\documentclass[11pt]{article}%
\usepackage{amsmath}
\usepackage{amsfonts}
\usepackage{amssymb}
\usepackage{amsthm}
\usepackage{graphicx}
\usepackage{fullpage}%
\providecommand{\U}[1]{\protect\rule{.1in}{.1in}}
\newtheorem{theorem}{Theorem}

\newtheorem{claim}[theorem]{Claim}

\newtheorem{conjecture}[theorem]{Conjecture}
\newtheorem{corollary}[theorem]{Corollary}

\newtheorem{definition}[theorem]{Definition}

\newtheorem{lemma}[theorem]{Lemma}

\newtheorem{proposition}[theorem]{Proposition}

\newtheorem*{numlessthm}{Theorem}
\newtheorem*{numlessconj}{Conjecture}
\begin{document}

\title{Quantum Money from Hidden Subspaces}
\author{Scott Aaronson\thanks{MIT. \ Email: aaronson@csail.mit.edu. \ \ This material
is based upon work supported by the National Science Foundation under Grant
No. 0844626. \ Also supported by a DARPA YFA grant, an NSF STC grant, a TIBCO
Chair, and a Sloan Fellowship.}
\and Paul Christiano\thanks{This work was done while the author was a student at
MIT. \ Email: paulfchristiano@gmail.com.}}
\date{}
\maketitle

\begin{abstract}
Forty years ago, Wiesner pointed out that quantum mechanics raises the
striking possibility of money that cannot be counterfeited according to the
laws of physics. \ We propose the first quantum money scheme that is

(1) \textit{public-key}---meaning that anyone can verify a banknote as
genuine, not only the bank that printed it, and

(2) \textit{cryptographically secure}, under a\ \textquotedblleft
classical\textquotedblright\ hardness assumption that has nothing to do with
quantum money.

Our scheme is based on \textit{hidden subspaces},\ encoded as the zero-sets of
random multivariate polynomials. \ A main technical advance is to show that
the \textquotedblleft black-box\textquotedblright\ version of our scheme,
where the polynomials are replaced by classical oracles, is
\textit{unconditionally} secure. \ Previously, such a result had only been
known relative to a \textit{quantum} oracle (and even there, the proof was
never published).

Even in Wiesner's original setting---quantum money that can only be verified
by the bank---we are able to use our techniques to patch a major security hole
in Wiesner's scheme. \ We give the first private-key quantum money scheme that
allows unlimited verifications and that remains unconditionally secure, even
if the counterfeiter can interact adaptively with the bank.

Our money scheme is simpler than previous public-key quantum money schemes,
including a knot-based scheme of Farhi et al. \ The verifier needs to perform
only two tests, one in the standard basis and one in the Hadamard
basis---matching the original intuition for quantum money, based on the
existence of complementary observables.

Our security proofs use a new variant of Ambainis's quantum adversary method,
and several other tools that might be of independent interest.

\end{abstract}

\tableofcontents

\section{Introduction\label{INTRO}}

\textit{\textquotedblleft Information wants to be free\textquotedblright%
}---this slogan expresses the idea that classical bits, unlike traditional
economic goods, can be copied an unlimited number of times. \ The copyability
of classical information is one of the foundations of the digital economy, but
it is also a nuisance to governments, publishers, software companies, and
others who wish to prevent copying. \ Today, essentially all electronic
commerce involves a trusted third party, such as a credit card company, to
mediate transactions. \ Without such a third party entering at \textit{some}
stage, it is impossible to prevent electronic cash from being counterfeited,
regardless of what cryptographic assumptions one makes.\footnote{The recent
Bitcoin system is an interesting illustration of this principle: it gets rid
of the centralized third party, but still uses a \textquotedblleft third
party\textquotedblright\ distributed over the community of Bitcoin users.}

Famously, though, quantum bits do\textit{ not }\textquotedblleft want to be
free\textquotedblright\ in the same sense that classical bits do:\ in many
respects, they behave more like gold, oil, or other traditional economic
goods. \ Indeed, the \textit{No-Cloning Theorem}, which is an immediate
consequence of the linearity of quantum mechanics, says that there is no
physical procedure that takes as input an unknown\footnote{The adjective
\textquotedblleft unknown\textquotedblright\ is needed because, if we knew a
classical description of a procedure to \textit{prepare} $\left\vert
\psi\right\rangle $, then of course we could run that procedure multiple times
to prepare multiple copies.} quantum pure state $\left\vert \psi\right\rangle
$, and that produces as output two unentangled copies of $\left\vert
\psi\right\rangle $, or even a close approximation thereof. \ The No-Cloning
Theorem is closely related to the \textit{uncertainty principle}, which says
that there exist \textquotedblleft complementary\textquotedblright\ properties
of a quantum state (for example, its position and momentum) that cannot both
be measured to unlimited accuracy.\footnote{Indeed, if we could copy
$\left\vert \psi\right\rangle $, then we could violate the uncertainty
principle by measuring one observable (such as position) on some copies, and a
complementary observable (such as momentum) on other copies. \ Conversely, if
we could measure all the properties of $\left\vert \psi\right\rangle $\ to
unlimited accuracy, then we could use the measurement results to create
additional copies of $\left\vert \psi\right\rangle $.}

\subsection{The History of Quantum Money\label{HISTORY}}

But can one actually \textit{exploit} the No-Cloning Theorem to achieve
classically-impossible cryptographic tasks? \ This question was first asked by
Wiesner \cite{wiesner},\ in a remarkable paper written around 1970 (but only
published in 1983) that arguably founded quantum information science. \ In
that paper, Wiesner proposed a scheme for \textit{quantum money} that would be
physically impossible to clone. \ In Wiesner's scheme, each \textquotedblleft
banknote\textquotedblright\ would consist of a classical serial number $s$,
together with a quantum state $\left\vert \psi_{s}\right\rangle $\ consisting
of $n$ unentangled qubits, each one $\left\vert 0\right\rangle $, $\left\vert
1\right\rangle $, $\frac{\left\vert 0\right\rangle +\left\vert 1\right\rangle
}{\sqrt{2}}$, or $\frac{\left\vert 0\right\rangle -\left\vert 1\right\rangle
}{\sqrt{2}}$ with equal probability. \ The issuing bank would maintain a giant
database, which stored a classical description\ of $\left\vert \psi
_{s}\right\rangle $\ for each serial number $s$. \ Whenever someone wanted to
\textit{verify} a banknote, he or she would take it back to the
bank---whereupon the bank would use its knowledge of how $\left\vert \psi
_{s}\right\rangle $\ was prepared to measure each qubit in the appropriate
basis, and check that it got the correct outcomes. \ On the other hand, it can
be proved \cite{molina}\ that someone who did \textit{not} know the
appropriate bases could copy the banknote with success probability at most
$\left(  3/4\right)  ^{n}$.

Though historically revolutionary, Wiesner's money scheme suffered at least
three drawbacks:

\begin{enumerate}
\item[(1)] \textbf{The \textquotedblleft Verifiability
Problem\textquotedblright:} The only entity that can verify a banknote is the
bank that printed it.

\item[(2)] \textbf{The \textquotedblleft Online Attack
Problem\textquotedblright:} A counterfeiter able to submit banknotes for
verification, and get them back afterward, can easily break Wiesner's scheme
(\cite{lutomirski:attack,aar:qcopy}; see also Section \ref{QUERYSEC}).

\item[(3)] \textbf{The \textquotedblleft Giant Database
Problem\textquotedblright:} The bank needs to maintain a database with an
entry for every banknote in circulation.
\end{enumerate}

In followup work in 1982, Bennett, Brassard, Breidbart, and Wiesner
\cite{bbbw} (henceforth BBBW) at least showed how to eliminate the giant
database problem: namely, by generating the state $\left\vert \psi
_{s}\right\rangle =\left\vert \psi_{f_{k}\left(  s\right)  }\right\rangle
$\ using a \textit{pseudorandom function} $f_{k}$, with key $k$ known only by
the bank. \ Unlike Wiesner's original scheme, the BBBW scheme is no longer
\textit{information-theoretically }secure: a counterfeiter can recover $k$
given exponential computation time. \ On the other hand, a counterfeiter
cannot break the scheme in polynomial time, unless it can \textit{also}
distinguish $f_{k}$\ from a random function.

These early ideas about quantum money inspired the field of \textit{quantum
cryptography} \cite{bb84}. \ But strangely, the subject of quantum money
itself lay dormant for more than two decades, even as interest in quantum
computing exploded. \ However, the past few years have witnessed a
\textquotedblleft quantum money renaissance.\textquotedblright\ \ Some recent
work has offered partial solutions to the verifiability problem. \ For
example, Mosca and Stebila \cite{moscastebila} suggested that the bank use a
\textit{blind quantum computing} protocol to offload the verification of
banknotes to local merchants, while Gavinsky \cite{gavinsky:money}\ (see also
followup work by Molina et al.\ \cite{molina}\ and Pastawski et al.\
\cite{pastawski}) proposed a variant of Wiesner's scheme that requires only
\textit{classical} communication between the merchant and bank.

However, most of the focus today is on a more ambitious goal: namely, creating
what Aaronson \cite{aar:qcopy}\ called \textit{public-key quantum money},\ or
quantum money that \textit{anyone} could authenticate, not just the bank that
printed it. As with public-key cryptography in the 1970s, it is far from
obvious \textit{a priori} whether public-key quantum money is possible at all.
\ Can a bank publish a description of a quantum circuit\ that lets people
feasibly \textit{recognize} a state $\left\vert \psi\right\rangle $, but does
not let them feasibly prepare or even copy $\left\vert \psi\right\rangle $?

Aaronson \cite{aar:qcopy} gave the first formal treatment of public-key
quantum money, as well as related notions such as copy-protected quantum
software. \ He proved that there exists a \textit{quantum oracle} relative to
which secure public-key quantum money is possible. \ Unfortunately, that
result, though already involved, did not lead in any obvious way to an
explicit (or \textquotedblleft real-world\textquotedblright) quantum money
scheme.\footnote{Also, the proof of Aaronson's result never appeared---an
inexcusable debt that this paper finally repays, with interest.} \ He raised
as an open problem whether secure public-key quantum money is possible
relative to a \textit{classical} oracle. \ In the same paper, Aaronson also
proposed an explicit scheme, based on random stabilizer states, but could not
offer any evidence for its security. \ And indeed, the scheme was broken about
a year afterward by Lutomirski et al.\ \cite{breaking}, using an algorithm for
finding planted cliques in random graphs due to Alon, Krivelevich, and Sudakov
\cite{alonclique}.

Recently, Farhi et al.\ \cite{knots} took a completely different approach to
public-key quantum money. \ They proposed a quantum money scheme based on
\textit{knot theory}, where each banknote is a superposition over
exponentially-many oriented link diagrams. \ Within a given banknote, all the
link diagrams $L$\ have the same Alexander polynomial $p\left(  L\right)  $ (a
certain knot invariant).\footnote{Instead of knots, Farhi et al.\ \cite{knots}%
\ could also have used, say, superpositions over $n$-vertex\ \textit{graphs}
having the same eigenvalue spectrum. \ But in that case, their scheme would
have been breakable, the reason being that the \textit{graph isomorphism
problem} is easy for random graphs. \ By contrast, it is not known how to
solve \textit{knot isomorphism} efficiently, even with a quantum computer and
even for random knots.} \ This $p\left(  L\right)  $, together with a digital
signature of $p\left(  L\right)  $, serves as the banknote's \textquotedblleft
classical serial number.\textquotedblright\ \ Besides the unusual mathematics
employed, the work of Farhi et al.\ \cite{knots} (building on \cite{breaking}%
)\ also introduced an idea that will play a major role in our work. \ That
idea is to construct public-key quantum money schemes by composing two
\textquotedblleft simpler\textquotedblright\ ingredients: first, objects that
we call \textit{mini-schemes}; and second, classical digital signature schemes.

The main disadvantage of the knot-based scheme, which it shares with every
previous scheme, is that no one can say much about its security---other than
that it has not yet been broken, and that various known counterfeiting
strategies fail. \ Indeed, even characterizing \textit{which quantum states
Farhi et al.'s verification procedure accepts} remains a difficult open
problem, on which progress seems likely to require major advances in knot
theory! \ In other words, there might be states that look completely different
from \textquotedblleft legitimate banknotes,\textquotedblright\ but are still
accepted with high probability.

In followup work, Lutomirski \cite{lutomirski:scp} proposed an
\textquotedblleft abstract\textquotedblright\ version of the knot scheme,
which gets rid of the link diagrams and Alexander polynomials, and simply uses
a classical oracle to achieve the same purposes. \ Lutomirski raised the
challenge of proving that this \textit{oracle} scheme is secure---in which
case, it would have yielded the first public-key quantum money scheme that was
proven secure relative to a classical oracle. \ Unfortunately, proving the
security of Lutomirski's scheme remains open, and seems hard.\footnote{One way
to understand the difficulty is that any security proof for Lutomirski's
scheme would need to contain, as a special case, a quantum lower bound for the
so-called \textit{index erasure} problem \cite{amrr}. \ In other words, any
fast quantum algorithm for index erasure would imply a break of Lutomirski's
scheme.
\par
At present, the simplest known proof of a quantum lower bound for index
erasure is via a reduction from Aaronson's quantum lower bound for the\textit{
collision problem} \cite{aar:col}. \ The latter is proved using the polynomial
method of Beals et al.\ \cite{bbcmw}. \ In this work, by contrast, we will
only manage to prove the security of \textit{our}\ oracle scheme using a
specially-designed variant of Ambainis's quantum adversary method
\cite{ambainis}. \ There \textit{is} a recent lower bound for the index
erasure problem using the adversary method \cite{amrr}, but it is quite
involved.}

As alluded to earlier, there is already some research on ways to
\textit{break} quantum money schemes. \ Besides the papers by Lutomirski
\cite{lutomirski:attack} and Lutomirski et al.\ \cite{breaking} mentioned
before, let us mention\ the beautiful work of\ Farhi et al.\ on
\textit{quantum state restoration} \cite{farhi:restore}. \ As we discuss in
Section \ref{QUERYSEC}, quantum state restoration can be used to break many
public-key quantum money schemes: roughly speaking, any scheme where the
banknotes contain only limited entanglement, and where verification consists
of a rank-$1$ projective measurement. \ This fact explains why our scheme,
like the knot-based scheme of Farhi et al.\ \cite{knots}, will require
highly-entangled banknotes.

\subsection{The Challenge\label{CHALLENGE}}

Work over the past few years has revealed a surprising richness in the quantum
money problem---both in the ideas that have been used to construct public-key
quantum money schemes, \textit{and} in the ideas that have been used to break
them. \ Of course, this record also underscores the need for caution! \ To
whatever extent we can, we ought to hold quantum money schemes to modern
cryptographic standards, and not be satisfied with \textquotedblleft we tried
to break it and failed.\textquotedblright

It is easy to see that, if public-key quantum money is possible, then it must
rely on \textit{some} computational assumption, in addition to the No-Cloning
Theorem.\footnote{This is because a counterfeiter with unlimited time could
simply search for a state $\left\vert \psi\right\rangle $\ that the
(publicly-known) verification procedure accepted.} \ The best case would be to
show that secure, public-key quantum money is possible, \textit{if} (for
example) there exist one-way functions resistant to quantum attack.
\ Unfortunately, we seem a long way from showing anything of the kind. \ The
basic problem is that \textit{uncloneability} is a novel cryptographic
requirement: something that would not even make sense in a classical context.
\ Indeed, work by Farhi et al.\ \cite{farhi:restore}\ and Aaronson
\cite{aar:qcopy}\ has shown that it is sometimes possible to copy quantum
banknotes, via attacks that do not even \textit{measure} the banknotes in an
attempt to learn a classical secret! \ Rather, these attacks simply perform
some unitary transformation on a legitimate banknote $\left\vert
\$\right\rangle $\ together with an ancilla $\left\vert 0\right\rangle $, the
end result of which is to produce $\left\vert \$\right\rangle ^{\otimes2}$.
\ Given such a strange attack, how can one deduce the failure of any
\textquotedblleft standard\textquotedblright\ cryptographic assumption?

Yet despite the novelty of the quantum money problem---or perhaps because of
it---it seems reasonable to want \textit{some} non-tautological evidence that
a public-key quantum money scheme is secure. \ A minimal wish-list might include:

\begin{enumerate}
\item[(1)] Security under \textit{some} plausible assumption, of a sort
cryptographers know how to evaluate. \ Such an assumption should talk only
about computing a classical output from a classical input; it should have
nothing to do with cloning of quantum states.

\item[(2)] A proof that the money scheme is secure against \textit{black-box
counterfeiters}: those that do not exploit the structure of some cryptographic
function $f$ used in verifying the banknotes.

\item[(3)] A \textquotedblleft simple\textquotedblright\ verification process,
which accepts all valid banknotes $\left\vert \$\right\rangle $\ with
probability $1$, and rejects all banknotes that are far from $\left\vert
\$\right\rangle $.
\end{enumerate}

\subsection{Our Results\label{RESULTS}}

Our main contribution is a new public-key quantum money scheme, which achieves
all three items in the wish-list above, and which is the first to achieve (1)
or (2). \ Regardless of whether our particular scheme stands or falls, we
introduce at least four techniques that should be useful for the design and
analysis of \textit{any} public-key quantum money scheme. \ These are:

\begin{itemize}
\item The \textquotedblleft inner-product adversary method,\textquotedblright%
\ a new variant of Ambainis's quantum adversary method \cite{ambainis} that
can be used to rule out black-box counterfeiting strategies.

\item A formal proof that full-fledged quantum money schemes can be
constructed out of two simpler ingredients: (a) objects that we call
\textit{mini-schemes}, and (b) conventional digital signature schemes secure
against quantum attack. \ (Note that this construction itself, \textit{sans}
the analysis, was introduced in earlier work on quantum money, by Lutomirski
et al.\ \cite{breaking} and Farhi et al.\ \cite{knots}.)

\item A method to \textit{amplify} weak counterfeiters into strong ones, so
that one only needs to rule out the latter to show security.

\item A new connection between (a) the security of quantum money schemes, and
(b) the security of conventional cryptosystems against attacks that succeed
with exponentially-small probabilities.
\end{itemize}

A second contribution is to construct the first \textit{private}-key quantum
money schemes that remain \textit{unconditionally} secure, even if the
counterfeiter can interact adaptively with the bank. \ This gives the first
solution to the \textquotedblleft online attack problem,\textquotedblright\ a
major security hole in the Wiesner \cite{wiesner}\ and BBBW \cite{bbbw}%
\ schemes pointed out by Lutomirski \cite{lutomirski:attack}\ and Aaronson
\cite{aar:qcopy}. \ These private-key schemes are direct adaptations of our
public-key scheme.

In more detail, our quantum money scheme is based on \textit{hidden subspaces}
of the vector space $\mathbb{F}_{2}^{n}$. \ Each of our money states is a
uniform superposition of the vectors in a random $n/2$-dimensional subspace
$A\leq\mathbb{F}_{2}^{n}$. \ We denote this superposition by $\left\vert
A\right\rangle $. \ Crucially, we can recognize the state $\left\vert
A\right\rangle $ using only membership oracles for $A$ and for its dual
subspace $A^{\perp}$. \ To do so, we apply the membership oracle for $A$, then
a Fourier transform, then the membership oracle for $A^{\perp}$, and then a
second Fourier transform to restore the original state. \ We prove that this
operation computes a rank-$1$ projection onto $\left\vert A\right\rangle $.

Underlying the security of our money schemes is the assertion that the states
$\left\vert A\right\rangle $ are difficult to clone, even given membership
oracles for $A$ and $A^{\perp}$. \ Or more concretely: \textit{any quantum
algorithm that maps }$\left\vert A\right\rangle $\textit{\ to }$\left\vert
A\right\rangle ^{\otimes2}$\textit{\ must make }$2^{\Omega\left(  n\right)  }%
$\textit{\ queries to the }$A,A^{\perp}$\textit{\ oracles.}

In order to prove this statement, we introduce a new method for proving lower
bounds on quantum query complexity, which we call the \textit{inner-product
adversary method}. \ This technique considers a single counterfeiting
algorithm being run in parallel to clone two distinct states $\left\vert
A\right\rangle $ and $\left\vert A^{\prime}\right\rangle $, with each having
access to the membership oracles for $A,A^{\perp}$ or $A^{\prime}%
,A^{\prime\perp}$, as appropriate. \ To measure how much progress the
algorithm has made, we consider the inner product between the states produced
by the parallel executions: because $\left\langle A\right\vert ^{\otimes
2}\left\vert A^{\prime}\right\rangle ^{\otimes2}<\left\langle A|A^{\prime
}\right\rangle $ for many pairs of subspaces $A,A^{\prime}$, in order to
succeed a counterfeiter will have to reduce this inner product substantially.
\ We prove that when averaged over a suitable distribution of pairs
$A,A^{\prime}$, the \textit{expected inner product} between the two states
produced by the counterfeiter cannot decrease too much with a single query to
the membership oracles. \ We conclude that in order to produce $\left\vert
A\right\rangle ^{\otimes2}$ given $\left\vert A\right\rangle $ and membership
oracles for $A,A^{\perp}$, a counterfeiter must use exponentially many queries.

Having ruled out the possibility of nearly perfect cloning, we introduce a new
amplification protocol, which allows us to transform a counterfeiter who
succeeds with $\Omega\left(  1/\operatorname*{poly}\left(  n\right)  \right)
$ success probability into a counterfeiter who succeeds with probability
arbitrarily close to $1$. \ This technique is based on combining standard
Grover search with a monotonic state amplification protocol of Tulsi, Grover,
and Patel \cite{tgp}, to obtain monotonic convergence with the quadratic
speedup of Grover search.\footnote{Although the \textquotedblleft quadratic
speedup\textquotedblright\ part is not strictly necessary for us, it improves
our lower bound on the number of queries the counterfeiter needs to make---to
the tight one, in fact---and might be of independent interest.} \ Combining
this amplification with the inner-product adversary method, and applying a
random linear transformation to convert the counterfeiter's worst case to its
average case, we conclude that no counterfeiting algorithm can succeed with
any non-negligible probability on a non-negligible fraction of states
$\left\vert A\right\rangle $.

Using these results, how do we produce a secure quantum money scheme? \ We now
need to step back, and discuss some general constructions that have nothing to
do with hidden subspaces in particular. \ Before constructing full-fledged
quantum money schemes, we find it useful---following \cite{breaking,knots}%
---to construct simpler objects called \textit{quantum money mini-schemes}, in
which the bank issues only a single money state and maintains no secret
information. \ Formally, a mini-scheme is a protocol $\mathsf{Bank}$ for
outputting pairs $(s,\rho_{s})$ and a verification procedure $\mathsf{Ver}%
_{s}$ for identifying $\rho_{s}$. \ We say a mini-scheme is \textit{complete}
if the state $\rho_{s}$ passes the verification $\mathsf{Ver}_{s}$ with high
probability, and we say the scheme is \textit{secure} if furthermore no
counterfeiter can take a single state $\rho_{s}$, and produce two
(possibly-entangled) states $\rho_{1}$ and $\rho_{2}$ which simultaneously
pass the verification procedure with non-negligible probability.

In the case of hidden subspace money, for example, we can use our
uncloneability result to produce a secure mini-scheme relative to a classical
oracle. \ The algorithm $\mathsf{Bank}$ queries the classical oracle to obtain
a serial number $s$ and the description of a subspace $A$. \ Using this
description, it prepares $\left\vert A\right\rangle $ and publishes $\left(
s,\left\vert A\right\rangle \right)  $. \ The verification procedure uses the
serial number $s$ as an index into another classical oracle, which allows it
to test membership in $A$ and $A^{\perp}$. \ We prove that the uncloneability
of the states $\left\vert A\right\rangle $ implies that this mini-scheme is secure.

Using a construction introduced by Lutomirski et al.\ \cite{breaking} and Farhi
et al.\ \cite{knots}, we also show that, given any mini-scheme $\mathcal{M}$,
one can obtain a full-fledged quantum money scheme\ by combining $\mathcal{M}%
$\ with any (classical) digital signature scheme secure against quantum
attacks. \ In the construction of \cite{breaking,knots}, the issuing bank
first uses the mini-scheme to produce a pair $\left(  s,\rho_{s}\right)  $;
then it digitally signs the serial number $s$ and distributes $\left(
s,\rho_{s},\mathsf{S{}ign}\left(  s\right)  \right)  $\ as its banknote. \ Our
contribution is to prove rigorously that, \textit{if} a counterfeiter can
break the money scheme, then it must have been able to break either the
underlying mini-scheme or else the signature scheme.

By combining this reduction with our mini-scheme, we are able to obtain a
\textquotedblleft black-box\textquotedblright\ public key quantum money scheme
relative to a classical oracle, which is unconditionally secure:

\begin{numlessthm}
[Security of Hidden Subspace Money]Relative to some (classical) oracle $A$,
there exists a secure public-key quantum money scheme.

More precisely, there is an algorithm $\mathsf{KeyGen}^{A}$ which outputs
pairs $\left(  k_{\operatorname*{private}},k_{\operatorname*{public}}\right)
$ with security parameter $n$; an algorithm $\mathsf{Bank}^{A}\left(
k_{\operatorname*{private}}\right)  $ which generates a \textquotedblleft
quantum banknote\textquotedblright\ $\left\vert \$\right\rangle $; and a
verification algorithm $\mathsf{Ver}^{A}\left(  k_{\operatorname*{public}%
},\left\vert \$\right\rangle \right)  $ which tests the authenticity of a
purported banknote. \ These algorithms are polynomial-time and have the
following properties:

\textbf{Completeness:} If $(k_{\operatorname*{private}}%
,k_{\operatorname*{public}})$ is produced by $\mathsf{KeyGen}^{A}$, then
$\mathsf{Ver}^{A}\left(  k_{\operatorname*{public}},\mathsf{Bank}^{A}\left(
k_{\operatorname*{private}}\right)  \right)  $ accepts with certainty.

\textbf{Soundness:} Suppose a would-be polynomial-time counterfeiter with
access to $A$ and $k_{\operatorname*{public}}$ is given $q$ valid banknotes.
\ If this counterfeiter outputs any number of (possibly-entangled) quantum
states, there is at most a $1/\exp\left(  n\right)  $ probability that
$\mathsf{Ver}^{A}$ will accept more than $q$ of them.
\end{numlessthm}

By adapting these ideas to the private-key setting, we are also able to
provide the first \textit{private}-key quantum money scheme that is
unconditionally secure, even if the counterfeiter is able to interact
adaptively with the bank. \ This patches a security hole in Wiesner's original
scheme which was observed in \cite{lutomirski:attack,aar:qcopy}, but which has
not previously been addressed in a provably-secure way.

Finally, we provide a candidate cryptographic protocol for obfuscating the
indicator functions of subspaces $A\leq\mathbb{F}_{2}^{n}$. \ In order to
obfuscate a membership oracle for $A$, we provide a random system of
polynomials $p_{1},\ldots,p_{m}$\ that vanish on $A$. \ Membership in $A$ can
be tested by evaluating the $p_{i}$'s, but given only the $p_{i}$'s, we
conjecture that it is difficult to recover $A$. \ Combining this protocol with
the black-box money scheme, we obtain an \textit{explicit} quantum money
scheme. \ This scheme is also the first public-key quantum money scheme whose
security can be based on a plausible \textquotedblleft
classical\textquotedblright\ cryptographic assumption. \ Here is the assumption:

\begin{numlessconj}
[*]Suppose $A$ is a uniformly-random $n/2$-dimensional subspace of
$\mathbb{F}_{2}^{n}$, and that $\left\{  p_{i}\right\}  _{1\leq i\leq
2n},\left\{  q_{i}\right\}  _{1\leq i\leq2n}$ are systems of degree-$d$
polynomials from $\mathbb{F}_{2}^{n}$ to $\mathbb{F}_{2}$,\ which vanish on
$A$ and $A^{\perp}$ respectively but are otherwise uniformly-random. \ Then
for large enough constant $d$, there is no polynomial-time quantum algorithm
that takes as input descriptions of the $p_{i}$'s and $q_{i}$'s, and that
outputs a basis for $A$ with success probability $\Omega\left(  2^{-n/2}%
\right)  $.
\end{numlessconj}

Note that we can trivially guess a \textit{single} nonzero $A$\ element with
success probability $2^{-n/2}$, but guessing a whole \textit{basis} for $A$
would succeed with probability only $2^{-\Omega\left(  n^{2}\right)  }$.
\ Conjecture (*)\ asserts that it is harder to find many elements of $A$ than
to find just one element.

The following theorem says that, \textit{if} a counterfeiter could break our
quantum money scheme, then with nontrivial success probability, it could
\textit{also} recover a description of $A$ from the $p_{i}$'s and $q_{i}$'s
alone---even without having access to a bank that provides a valid money state
$\left\vert A\right\rangle $.

\begin{numlessthm}
Assuming Conjecture~(*), there exists a public-key quantum money scheme with
perfect completeness and $1/\exp\left(  n\right)  $ soundness error. \ That
is, the verifier always accepts valid banknotes, and a would-be counterfeiter
succeeds only with $1/\exp\left(  n\right)  $ probability.\footnote{This
theorem remains true even if the statement of Conjecture~(*) is weakened by
adding random noise to the $p_{i}$'s and $q_{i}$'s, so that only a constant
fraction of them vanish on $A$ or $A^{\perp}$. \ The presence of noise
interferes substantially with known techniques for solving systems of
equations, though an attacker who was able to recover $A$ from a
\textit{single} polynomial would of course not be hindered by such noise.}
\end{numlessthm}

The problem of recovering a subspace $A$, given a system of equations that
vanish on $A$, is closely related to \textit{algebraic cryptanalysis}, and in
particular to the so-called \textit{polynomial isomorphism problem}. In the
latter problem, we are given as input two polynomials $p,q:\mathbb{F}%
^{n}\rightarrow\mathbb{F}$ related by an unknown linear change of basis $L$;
the challenge is to find $L$. \ When $\deg\left(  p\right)  =\deg\left(
q\right)  =3$, the best known algorithms for the polynomial isomorphism
problem require exponential time \cite{cgp:ip1s,gms:ip1s,bffp:ip1s}. \ An
attacker \textit{might} be able to use known techniques to effectively reduce
the degree of the polynomials in our scheme by $1$, at the expense of an
exponentially reduced success probability \cite{bffp:ip1s}. \ Provided the
degree is at least $4$, however, recovering $A$ seems to be well beyond
existing techniques.

\subsection{Motivation\label{MOTIVATION}}

Unlike the closely-related task of \textit{quantum key distribution}
\cite{bb84} (which is already practical), quantum money currently seems to be
a long way off. \ The basic difficulty is how to maintain the
\textit{coherence} of a quantum money state for an appreciable length of time.
\ \textit{All} money eventually loses its value unless it is spent, but money
that decohered on a scale of microseconds would be an extreme example!

So one might wonder: why develop rigorous foundations for a cryptographic
functionality that seems so far from being practical? \ One answer is that,
just as quantum key distribution uses many of the same ideas as private-key
quantum money, but without requiring long-lasting coherence, so it is not hard
to imagine protocols that would use many of the same ideas as \textit{public}%
-key quantum money without requiring long-lasting coherence. \ Indeed,
depending on the problem, rapid decoherence might be a \textit{feature} rather
than a bug!

As one example, public-key quantum money that decohered quickly could be used
to create \textbf{non-interactive uncloneable signatures}. \ These are
$n$-qubit quantum states $\left\vert \psi\right\rangle $ that an agent can
efficiently\ prepare using a private key, then freely hand out to passersby.
\ By feeding $\left\vert \psi\right\rangle $, together with the agent's
\textit{public} key, into suitable measuring equipment, anyone can verify on
the spot that the agent is who she says she is and not an impostor. \ Compared
with \textit{classical} identification protocols, the novel feature here is
that the agent does not need to respond to a \textit{challenge}---for example,
digitally signing a random string---but can instead just hand out a fixed
$\left\vert \psi\right\rangle $\ non-interactively. \ Furthermore, because
$\left\vert \psi\right\rangle $\ decoheres in a matter of seconds,\ and
recovering a classical \textit{description}\ of $\left\vert \psi\right\rangle
$\ from measurements on it is computationally intractable, \noindent someone
who is given $\left\vert \psi\right\rangle $ cannot use it later to
impersonate the agent.

Of course, if an attacker managed to solve the technological problem of
keeping $\left\vert \psi\right\rangle $\ coherent for very long times, then he
could break this system, by collecting one or more copies of $\left\vert
\psi\right\rangle $\ that an agent had handed out, and using them to
impersonate the agent. \ But in that case, whatever method the attacker was
using to keep the states\ coherent could \textit{also}---once discovered---be
used to create a secure public-key quantum money scheme!

However, we believe the \textquotedblleft real\textquotedblright\ reason to
study quantum money is basically the same as the \textquotedblleft
real\textquotedblright\ reason to study quantum computing as a whole---or for
that matter, to study the many interesting aspects of \textit{classical}
cryptography that are equally far from application. \ As theoretical computer
scientists, we are in the business of mapping out the inherent capabilities
and limits of information processing.

In our case, what quantum money provides is a near-ideal playground for
understanding the implications of the uncertainty principle and the No-Cloning
Theorem. \ In the early days of quantum mechanics, Bohr \cite{bohr}\ and
others argued that the uncertainty principle requires us to change our
conception of science itself---their basic argument being that, in physics,
predictions are only ever as good as our knowledge of a system's initial state
$\left\vert \psi\right\rangle $, but the uncertainty principle might mean that
the initial state is unknowable even with arbitrarily-precise measurements.

But does this argument have any \textquotedblleft teeth\textquotedblright?
\ In other words: among the properties of a quantum state $\left\vert
\psi\right\rangle $\ that make the state impossible to learn precisely or to
duplicate, can any of those properties ever \textit{matter empirically}? \ To
us, quantum money is interesting precisely because it gives one of the
clearest examples where the answer to that question is yes.

\section{Preliminaries\label{PRELIM}}

To begin, we fix some notation. \ Let $\left[  N\right]  =\left\{
1,\ldots,N\right\}  $. \ We call a function $\delta\left(  n\right)
$\ \textit{negligible} if $\delta\left(  n\right)  =o\left(  1/p\left(
n\right)  \right)  $\ for every polynomial $p$. \ Given a subspace $S$ of a
vector space $V$, let $S^{\bot}$\ be the \textit{orthogonal complement} of
$S$: that is, the set of $y\in V$\ such that $x\cdot y=0$\ for all $x\in S$.
\ It is not hard to show that $S^{\bot}$\ is also a subspace of $V$, that
$\left(  S^{\bot}\right)  ^{\bot}=S$, and that these properties hold even if
$\cdot$\ is \textquotedblleft merely\textquotedblright\ a dot product rather
than an inner product. \ As a word of warning, this paper will use the same
notation $S^{\bot}$ in two very different contexts:

\begin{itemize}
\item When $V=\mathbb{C}^{2^{n}}$, the orthogonal complement $S^{\bot}$\ of,
e.g., the subspace $S\leq V$ spanned by a single computational basis state
$\left\vert x\right\rangle $, has $2^{n}-1$\ dimensions and is spanned by all
basis states $\left\vert y\right\rangle $ such that $y\neq x$.

\item When $V=\mathbb{F}_{2}^{n}$, the orthogonal complement $S^{\bot}$\ of,
e.g., the subspace $S\leq V$\ spanned by a single string $x=x_{1}\ldots x_{n}%
$, has $n-1$\ dimensions and consists of all strings $y=y_{1}\ldots y_{n}$
such that $x_{1}y_{1}+\cdots+x_{n}y_{n}\equiv0\left(  \operatorname{mod}%
2\right)  $.
\end{itemize}

By a \textit{classical oracle}, we will mean a unitary transformation of the
form $\left\vert x\right\rangle \rightarrow\left(  -1\right)  ^{f\left(
x\right)  }\left\vert x\right\rangle $, for some Boolean function $f:\left\{
0,1\right\}  ^{\ast}\rightarrow\left\{  0,1\right\}  $. \ Note that, unless
specified otherwise, even a classical oracle can be queried in quantum
superposition. \ A \textit{quantum oracle}, by contrast, is an arbitrary
$n$-qubit unitary transformation $U$ (or rather, a collection of such $U$'s,
one for each $n$) that a quantum algorithm can apply in a black-box fashion.
\ Quantum oracles were defined and studied by Aaronson and Kuperberg \cite{ak}.

\subsection{Cryptography\label{CRYPTO}}

Before we construct quantum money schemes, it will be helpful to have some
\textquotedblleft conventional\textquotedblright\ cryptographic primitives in
our toolbox. \ Foremost among these is a\textit{ digital signature scheme
secure against quantum chosen-message attacks}. \ We now define digital
signature schemes---both for completeness, and to fix the quantum attack model
that is relevant for us.

\begin{definition}
[Digital Signature Schemes]\label{sig}A (classical, public-key)
\textbf{digital signature scheme} $\mathcal{D}$\ consists of three
probabilistic polynomial-time classical algorithms:

\begin{itemize}
\item $\mathsf{KeyGen}$, which takes as input a security parameter $0^{n}$,
and generates a \textbf{key pair} $\left(  k_{\operatorname*{private}%
},k_{\operatorname*{public}}\right)  $.

\item $\mathsf{S{}ign}$, which takes as input $k_{\operatorname*{private}}%
$\ and a message $x$, and generates a \textbf{signature} $\mathsf{S{}%
ign}\left(  k_{\operatorname*{private}},x\right)  $.\footnote{We indulge in
slight abuse of notation, since if $\mathsf{S{}ign}$\ is randomized then the
signature need not be a function of $k_{\operatorname*{private}}$ and $x$.}

\item $\mathsf{Ver}$, which takes as input $k_{\operatorname*{public}}$,\ a
message $x$, and a claimed signature $w$, and either accepts or rejects.
\end{itemize}

We say $\mathcal{D}$ has \textbf{completeness error} $\varepsilon$\ if
$\mathsf{Ver}\left(  k_{\operatorname*{public}},x,\mathsf{S{}ign}\left(
x,k_{\operatorname*{private}}\right)  \right)  $\ accepts with probability at
least $1-\varepsilon$\ for all messages $x$ and key pairs $\left(
k_{\operatorname*{private}},k_{\operatorname*{public}}\right)  $. \ Here the
probability is over the behavior of $\mathsf{Ver}$\ and $\mathsf{S{}ign}$.

Let $C$\ (the \textbf{counterfeiter}) be a quantum circuit of size
$\operatorname*{poly}\left(  n\right)  $\ that takes
$k_{\operatorname*{public}}$\ as input\footnote{Actually, for our security
proofs, it suffices to consider a weaker attack model, in which $C$\ only
receives $k_{\operatorname*{public}}$ at the same time as it receives
$w_{1},\ldots,w_{m}$. \ This model was called \textquotedblleft existential
unforgeability under \textit{static} chosen-message attacks\textquotedblright%
\ by Cash et al.\ \cite{cash}. \ We thank an anonymous reviewer for this
observation.} and does the following:

\begin{enumerate}
\item[(1)] Probabilistically generates a classical list of messages
$x_{1},\ldots,x_{m}$, and submits them to a \textbf{signing oracle}
$\mathcal{O}$.

\item[(2)] Gets back independently-generated signatures $w_{1},\ldots,w_{m}$,
where $w_{i}:=\mathsf{S{}ign}\left(  k_{\operatorname*{private}},x_{i}\right)
$.

\item[(3)] Outputs a pair $\left(  x,w\right)  $.
\end{enumerate}

We say $C$\ \textbf{succeeds} if $x\notin\left\{  x_{1},\ldots,x_{m}\right\}
$\ and $\mathsf{Ver}\left(  k_{\operatorname*{public}},x,w\right)  $\ accepts.
\ We say $\mathcal{D}$ has \textbf{soundness error} $\delta$\ if every
counterfeiter $C$ succeeds with probability at most $\delta$. \ Here the
probability is over the key pair $\left(  k_{\operatorname*{private}%
},k_{\operatorname*{public}}\right)  $ and the behavior of $C$, $\mathsf{S{}%
ign}$, and $\mathsf{Ver}$.

We call $\mathcal{D}$ \textbf{secure against nonadaptive quantum
chosen-message attacks} if it has completeness error $\leq1/3$ and negligible
soundness error.
\end{definition}

Intuitively, we call a signature scheme \textquotedblleft
secure\textquotedblright\ if no \textit{quantum} counterfeiter with
\textit{nonadaptive, classical} access to a signing oracle $\mathcal{O}$\ can
forge a signature for any message that it did not submit to $\mathcal{O}$.
\ Depending on the application, one might want to generalize Definition
\ref{sig}\ in various ways: for example, by giving the counterfeiter
\textit{adaptive }or\textit{ quantum} access to $\mathcal{O}$, or by letting
$\mathsf{KeyGen}$, $\mathsf{S{}ign}$, and $\mathsf{Ver}$\ be quantum
algorithms themselves. \ For this paper, however, Definition \ref{sig}%
\ provides all we need.

Do signature schemes secure against quantum attack exist? \ Naturally,
signature schemes based on RSA or other number-theoretic problems can all be
broken by a quantum computer. \ However, building on earlier work by Naor and
Yung \cite{naoryung} (among many others), Rompel \cite{rompel}\ showed that a
secure public-key signature scheme can be constructed from \textit{any}
one-way function---not necessarily a trapdoor function. \ Furthermore,
Rompel's security reduction, from breaking the signature scheme to inverting
the one-way function, is \textit{black-box}: in particular, nothing in it
depends on the assumption that the adversary is classical rather than quantum.
\ We therefore get the following consequence:

\begin{theorem}
[Quantum-Secure Signature Schemes \cite{rompel}]\label{rompelthm}If there
exists a (classical) one-way function $f$\ secure against quantum attack, then
there also exists a digital signature scheme secure against quantum
chosen-message attacks.
\end{theorem}

Recently, Boneh et al.\ \cite{bdflsz} proved several results similar to
Theorem \ref{rompelthm}, and they needed nontrivial work to do so. \ However,
a crucial difference is that Boneh et al.\ were (justifiably) concerned with
quantum adversaries who can make \textit{quantum} queries to the signing
oracle $\mathcal{O}$. \ By contrast, as mentioned earlier, for our application
it suffices to consider adversaries who query $\mathcal{O}$
\textit{classically}---and in that case, the standard security reductions go
through essentially without change.

Let us state another consequence of Theorem \ref{rompelthm}, which will be
useful for our oracle construction in Section \ref{COR}.

\begin{theorem}
[Relativized Quantum-Secure Signatures]\label{sigoracle}Relative to a suitable
oracle $A$, there exists a digital signature scheme secure against quantum
chosen-message attacks.
\end{theorem}

\begin{proof}
[Proof Sketch]It is easy to give an oracle $A:\left\{  0,1\right\}  ^{\ast
}\rightarrow\left\{  0,1\right\}  $ relative to which there exists a one-way
function $f_{n}:\left\{  0,1\right\}  ^{n}\rightarrow\left\{  0,1\right\}
^{p\left(  n\right)  }$\ secure against quantum adversaries. \ Indeed, we can
let $A$\ be a \textit{random} oracle, and then define%
\[
f_{n}\left(  x\right)  :=A\left(  x,1\right)  \ldots A\left(  x,p\left(
n\right)  \right)
\]
directly in terms of $A$. \ Assume $p\left(  n\right)  \geq n$. \ Then the
lower bound on the quantum query complexity of function inversion, proved by
Bennett et al.\ \cite{bbbv} and Ambainis \cite{ambainis}, straightforwardly
implies that any quantum algorithm to invert $f_{n}$, with success probability
$\varepsilon>0$, must make $\Omega\left(  2^{n/2}\sqrt{\varepsilon}\right)
$\ quantum queries to $A$.

Now, the security reduction of Rompel \cite{rompel} is not only black-box but
\textit{relativizing}: that is, it goes through if all legitimate and
malicious parties have access to the same oracle $A$. \ So by Theorem
\ref{rompelthm}, starting from $\left\{  f_{n}\right\}  $\ one can construct a
digital signature scheme relative to the same oracle $A$, which is secure
against quantum chosen-message attacks.
\end{proof}

\subsection{Quantum Information\label{QI}}

Let us collect a few facts about quantum pure and mixed states that are used
in the paper. \ We assume basic familiarity with the formalism of bras, kets,
density matrices, etc.; see Nielsen and Chuang \cite{nc} for a good overview.

Given two mixed states $\rho$\ and $\sigma$, their \textit{trace distance} is
defined as $D\left(  \rho,\sigma\right)  :=\frac{1}{2}\sum_{i=1}^{N}\left\vert
\lambda_{i}\right\vert $, where $\lambda_{1},\ldots,\lambda_{N}$\ are the
eigenvalues of $\rho-\sigma$. \ Trace distance is a metric and satisfies
$0\leq D\left(  \rho,\sigma\right)  \leq1$. \ Also, the \textit{fidelity}
$0\leq F\left(  \rho,\sigma\right)  \leq1$\ is defined, in this paper, as the
maximum of $\left\vert \left\langle \psi|\varphi\right\rangle \right\vert
$\ over all purifications $\left\vert \psi\right\rangle $\ of $\rho$\ and
$\left\vert \varphi\right\rangle $\ of $\sigma$.\footnote{Some authors instead
define \textquotedblleft fidelity\textquotedblright\ as the maximum of
$\left\vert \left\langle \psi|\varphi\right\rangle \right\vert ^{2}$.} \ By
extension, given a subspace $S$, we let $F\left(  \rho,S\right)  $\ be the
maximum of $\left\vert \left\langle \psi|\varphi\right\rangle \right\vert
$\ over all purifications $\left\vert \psi\right\rangle $\ of $\rho$\ and all
unit vectors $\left\vert \varphi\right\rangle \in S$. \ Trace distance and
fidelity are related as follows \cite{nc}:

\begin{proposition}
\label{ineq}For all mixed states $\rho,\sigma$,%
\[
D\left(  \rho,\sigma\right)  \leq\sqrt{1-F\left(  \rho,\sigma\right)  ^{2}},
\]
with equality if $\rho$\ or $\sigma$\ is pure.
\end{proposition}

While fidelity is \textit{not} a metric, it does satisfy the following
inequality, which will be helpful in Section \ref{COR}.

\begin{lemma}
[\textquotedblleft Triangle Inequality\textquotedblright\ for Fidelity]%
\label{triangle}Suppose $\left\langle \psi\right\vert \rho\left\vert
\psi\right\rangle \geq1-\varepsilon$\ and $\left\langle \varphi\right\vert
\sigma\left\vert \varphi\right\rangle \geq1-\varepsilon$. \ Then $F\left(
\rho,\sigma\right)  \leq\left\vert \left\langle \psi|\varphi\right\rangle
\right\vert +2\varepsilon^{1/4}$.
\end{lemma}

\begin{proof}
By Proposition \ref{ineq},%
\[
D\left(  \rho,\left\vert \psi\right\rangle \right)  \leq\sqrt{1-\left\langle
\psi\right\vert \rho\left\vert \psi\right\rangle }\leq\sqrt{\varepsilon},
\]
and likewise $D\left(  \sigma,\left\vert \varphi\right\rangle \right)
\leq\sqrt{\varepsilon}$. \ Thus, since trace distance satisfies the triangle
inequality,%
\begin{align*}
D\left(  \rho,\sigma\right)   &  \geq D\left(  \left\vert \psi\right\rangle
,\left\vert \varphi\right\rangle \right)  -D\left(  \rho,\left\vert
\psi\right\rangle \right)  -D\left(  \sigma,\left\vert \varphi\right\rangle
\right)  \\
&  \geq\sqrt{1-\left\vert \left\langle \psi|\varphi\right\rangle \right\vert
^{2}}-2\sqrt{\varepsilon}.
\end{align*}
Then%
\begin{align*}
F\left(  \rho,\sigma\right)   &  \leq\sqrt{1-D\left(  \rho,\sigma\right)
^{2}}\\
&  \leq\sqrt{1-\left(  \sqrt{1-\left\vert \left\langle \psi|\varphi
\right\rangle \right\vert ^{2}}-2\sqrt{\varepsilon}\right)  ^{2}}\\
&  \leq\sqrt{\left\vert \left\langle \psi|\varphi\right\rangle \right\vert
^{2}+4\sqrt{\varepsilon}}\\
&  \leq\left\vert \left\langle \psi|\varphi\right\rangle \right\vert
+2\varepsilon^{1/4}.
\end{align*}

\end{proof}

Finally, the following lemma of Aaronson \cite{aar:adv} will imply that, as
long as a quantum money scheme has small \textit{completeness error} (i.e.,
small probability of rejecting a valid banknote),\ the banknotes can be reused
many times.

\begin{lemma}
[\textquotedblleft Almost As Good As New Lemma\textquotedblright%
\ \cite{aar:adv}]\label{goodasnew}Suppose a measurement on a mixed state
$\rho$\ yields a particular outcome with probability $1-\varepsilon$. \ Then
after the measurement, one can recover a state $\widetilde{\rho}$\ such that
$\left\Vert \widetilde{\rho}-\rho\right\Vert _{\operatorname*{tr}}\leq
\sqrt{\varepsilon}$.
\end{lemma}

\subsection{Quantum Search\label{SEARCH}}

In our security proof for quantum money, an important step will be to
\textit{amplify} a counterfeiter\ who copies a banknote $\$$\ with any
non-negligible fidelity to a counterfeiter who copies $\$$ almost perfectly.
\ Taking the contrapositive, this will imply that to rule out the former sort
of counterfeiter, it suffices to rule out the latter.

In this section, we first review two variants of Grover's search algorithm
\cite{grover}\ that are useful for amplifying the fidelity of quantum states.
\ We then introduce a variant that combines the advantages of both.

Assume we are given a pure initial state $\left\vert \operatorname*{Init}%
\right\rangle $, in some Hilbert space $\mathcal{H}$. \ Our goal is to map
$\left\vert \operatorname*{Init}\right\rangle $\ to a final state $\left\vert
\Psi\right\rangle $\ that lies in (or close to) a \textquotedblleft good
subspace\textquotedblright\ $G\leq\mathcal{H}$. \ We have oracle access to two
unitary transformations:

\begin{itemize}
\item $U_{\operatorname*{Init}}$, which maps $\left\vert \operatorname*{Init}%
\right\rangle $\ to $-\left\vert \operatorname*{Init}\right\rangle $, and acts
as the identity on all $\left\vert v\right\rangle $\ orthogonal to $\left\vert
\operatorname*{Init}\right\rangle $.

\item $U_{G}$, which maps $\left\vert v\right\rangle $\ to $-\left\vert
v\right\rangle $\ for all $\left\vert v\right\rangle \in G$, and acts as the
identity on all $\left\vert v\right\rangle $ orthogonal to $G$.
\end{itemize}

We are promised that the fidelity of the initial state with $G$,%
\[
F\left(  \left\vert \operatorname*{Init}\right\rangle ,G\right)
=\max_{\left\vert \psi\right\rangle \in G}\left\langle \operatorname*{Init}%
|\psi\right\rangle ,
\]
is at least some $\varepsilon>0$.

In this scenario, \textit{provided }$F\left(  \left\vert \operatorname*{Init}%
\right\rangle ,G\right)  $\textit{\ is known}, the amplitude amplification
framework of Brassard, H\o yer, Mosca, and Tapp \cite{bhmt} lets us prepare a
state close to $G$ using only $\Theta\left(  1/\varepsilon\right)  $\ iterations:

\begin{lemma}
[Amplitude Amplification \cite{bhmt}]\label{aa}Write $\left\vert
\operatorname*{Init}\right\rangle $\ as $\sin\theta\left\vert
\operatorname*{Good}\right\rangle +\cos\theta\left\vert \operatorname*{Bad}%
\right\rangle $, where $\left\vert \operatorname*{Good}\right\rangle $\ is the
unit vector formed by projecting $\left\vert \operatorname*{Init}\right\rangle
$\ onto $G$, and $\left\vert \operatorname*{Bad}\right\rangle $\ is orthogonal
to $\left\vert \operatorname*{Good}\right\rangle $. \ Then by using $O\left(
T\right)  $ oracle calls\ to $U_{\operatorname*{Init}}$ and $U_{G}$, we can
prepare the state%
\[
\left\vert \Phi_{T}\right\rangle :=\sin\left[  \left(  2T+1\right)
\theta\right]  \left\vert \operatorname*{Good}\right\rangle +\cos\left[
\left(  2T+1\right)  \theta\right]  \left\vert \operatorname*{Bad}%
\right\rangle
\]

\end{lemma}

Note that Grover's algorithm is simply a special case of Lemma \ref{aa}, where
$\left\vert \operatorname*{Init}\right\rangle $\ is the uniform superposition
over $N$\ basis states $\left\vert 1\right\rangle ,\ldots,\left\vert
N\right\rangle $, and $G$ is the subspace spanned by \textquotedblleft
marked\textquotedblright\ states.

However, Lemma \ref{aa}\ has an annoying drawback, which it shares with
ordinary Grover search. \ Namely, the algorithm does \textit{not} converge
monotonically toward the target subspace $G$, but could instead
\textquotedblleft wildly overshoot it,\textquotedblright\ cycling around the
$2$-dimensional\ subspace spanned by $\left\vert \operatorname*{Bad}%
\right\rangle $\ and $\left\vert \operatorname*{Good}\right\rangle $. \ If we
know the fidelity $F\left(  \left\vert \operatorname*{Init}\right\rangle
,G\right)  $\ in advance (rather than just a lower bound\ on the fidelity),
\textit{or} if we can prepare new copies of $\left\vert \operatorname*{Init}%
\right\rangle $ \textquotedblleft free of charge\textquotedblright\ in case of
failure, then this overshooting is not a serious problem. \ Alas, neither of
those conditions will hold in our application.

Fortunately, for independent reasons, in 2005 Tulsi, Grover, and Patel
\cite{tgp} introduced a new quantum search algorithm that \textit{does}
guarantee monotonic convergence toward $G$, by alternating unitary
transformations with measurements. \ (Their algorithm was later simplified and
improved by Chakraborty, Radhakrishnan, and Raghunathan \cite{crr}.)

\begin{lemma}
[Fixed-Point Quantum Search \cite{tgp,crr}]\label{fixedpoint}By using $T$
oracle calls\ to $U_{\operatorname*{Init}}$ and $U_{G}$, we can prepare a
state $\left\vert \Psi\right\rangle $\ such that $F\left(  \left\vert
\Psi\right\rangle ,G\right)  \geq1-\exp\left(  -T\varepsilon^{2}\right)  $.
\end{lemma}

Rearranging, Lemma \ref{fixedpoint}\ lets us prepare a state $\left\vert
\Psi\right\rangle $\ such that $F\left(  \left\vert \Psi\right\rangle
,G\right)  \geq1-\delta$\ using $T=O\left(  \frac{1}{\varepsilon^{2}}\log
\frac{1}{\delta}\right)  $ iterations. \ On the positive side, the dependence
on $1/\delta$\ in this bound is logarithmic: we get not only monotonic
convergence toward $G$, but \textit{exponentially-fast} convergence. \ On the
negative side, notice that the dependence on $\varepsilon$\ has worsened from
$1/\varepsilon$\ to $1/\varepsilon^{2}$---negating the quadratic speedup that
was the original point of quantum search!

In the rest of this section, we give a \textquotedblleft
hybrid\textquotedblright\ quantum search algorithm that combines the
advantages of Lemmas \ref{aa}\ and \ref{fixedpoint}---i.e., it converges
monotonically toward the target subspace $G$ (rather than \textquotedblleft
overshooting\textquotedblright\ $G$), but \textit{also} achieves a quadratic
speedup. \ In the context of our security proof for quantum money, this hybrid
algorithm will lead to a quadratically-better (and in fact, tight) lower bound
on the number of queries that a counterfeiter needs to make, compared to what
we would get from using Lemma \ref{fixedpoint} alone. \ While this quadratic
improvement is perhaps only of moderate interest, we include the algorithm in
the hope that it will find other applications.

We first give a technical lemma needed to analyze our algorithm.

\begin{lemma}
\label{interval}For all $L,\beta,\eta,\gamma$, there are at most $\left(
L/\beta+1\right)  \left(  2\eta+1\right)  $\ integers $T\in\left\{
0,\ldots,L\right\}  $\ such that $\left\vert T-\left(  \beta n+\gamma\right)
\right\vert <\eta$\ for some integer $n$.
\end{lemma}

\begin{proof}
The real interval $\left[  0,L\right]  $ can intersect at most $L/\beta
+1$\ intervals $\left(  \beta n+\gamma-\eta,\beta n+\gamma+\eta\right)  $, and
each such interval can contain at most $2\eta+1$\ integer points.
\end{proof}

We now give our hybrid of Lemmas \ref{aa}\ and \ref{fixedpoint}.

\begin{theorem}
[Faster Fixed-Point Search]\label{combined}Let $\delta\geq2\varepsilon$.
\ Then by using $O\left(  \frac{\log1/\delta}{\varepsilon\delta^{2}}\right)  $
oracle calls\ to $U_{\operatorname*{Init}}$ and $U_{G}$, we can prepare a
state $\rho$\ such that $F\left(  \rho,G\right)  \geq1-\delta$.
\end{theorem}

\begin{proof}
Let $\xi:=\arcsin\varepsilon$; note that $\varepsilon\leq\xi\leq\frac{\pi}%
{2}\varepsilon$.\ \ Also let $L:=\left\lceil 100/\xi\right\rceil $ and
$R:=\frac{25}{\delta^{2}}\left(  2+\log\frac{1}{\delta}\right)  $. \ Then the
algorithm is as follows:

\begin{enumerate}
\item[(1)] Choose an integer $T\in\left\{  0,\ldots,L\right\}  $\ uniformly at random.

\item[(2)] Apply $T$\ iterations of amplitude amplification with $\left\vert
\operatorname*{Init}\right\rangle $\ as the initial state and $G$\ as the
target subspace (as in Lemma \ref{aa}), to obtain a state $\left\vert \Phi
_{T}\right\rangle $.

\item[(3)] Apply $R$\ iterations of fixed-point quantum search with
$\left\vert \Phi_{T}\right\rangle $\ as the initial state and $G$ as the
target subspace (as in Lemma \ref{fixedpoint}), to obtain a state $\left\vert
\Psi_{T}\right\rangle $.
\end{enumerate}

The final output of the above algorithm is%
\[
\rho=\operatorname*{E}_{T\in\left\{  0,\ldots,L\right\}  }\left[  \left\vert
\Psi_{T}\right\rangle \left\langle \Psi_{T}\right\vert \right]  .
\]
Also, the total number of oracle calls to $U_{\operatorname*{Init}}$ and
$U_{G}$\ is%
\[
O\left(  TR\right)  =O\left(  \frac{\log1/\delta}{\varepsilon\delta^{2}%
}\right)  .
\]
(The reason this number scales like $TR$\ rather than $T+R$\ is that, in step
(3), each time we reflect about the initial state $\left\vert \Phi
_{T}\right\rangle $\ we need to rerun step (2). \ Thus, we need $\Theta\left(
T\right)  $ oracle calls within each of the $R$ iterations.)

By Lemma \ref{aa}, after step (2) we have a state $\left\vert \Phi
_{T}\right\rangle $\ such that%
\[
F\left(  \left\vert \Phi_{T}\right\rangle ,G\right)  =\left\vert \left\langle
\Phi_{T}|\operatorname*{Good}\right\rangle \right\vert =\left\vert \sin\left[
\left(  2T+1\right)  \xi\right]  \right\vert .
\]
So for any $\alpha\in\left(  0,1\right)  $,%
\begin{align*}
\Pr_{T\in\left\{  0,\ldots,L\right\}  }\left[  F\left(  \left\vert \Phi
_{T}\right\rangle ,G\right)  <\alpha\right]   &  =\Pr_{T\in\left\{
0,\ldots,L\right\}  }\left[  \left\vert \sin\left[  \left(  2T+1\right)
\xi\right]  \right\vert <\alpha\right]  \\
&  =\Pr_{T\in\left\{  0,\ldots,L\right\}  }\left[  \exists n\in\mathbb{Z}%
:\left\vert \left(  2T+1\right)  \xi-\pi n\right\vert <\arcsin\alpha\right]
\\
&  \leq\frac{\left(  \frac{L}{\pi/2\xi}+1\right)  \left(  \frac{\arcsin\alpha
}{\xi}+1\right)  }{L+1}\\
&  \leq\frac{2}{\pi}\arcsin\alpha+\frac{2\xi}{\pi}+\frac{\arcsin\alpha}%
{100}+\frac{\xi}{100}\\
&  \leq1.02\left(  \alpha+\varepsilon\right)  ,
\end{align*}
where the third line uses Lemma \ref{interval}.

Now assume $F\left(  \left\vert \Phi_{T}\right\rangle ,G\right)  \geq\alpha$.
\ Then by Lemma \ref{fixedpoint}, after step (3) we have a state $\left\vert
\Psi_{T}\right\rangle $\ such that%
\[
F\left(  \left\vert \Phi_{T}\right\rangle ,G\right)  \geq1-\exp\left(
-R\alpha^{2}\right)  .
\]
Let us now make the choice $\alpha:=\delta/5$. \ Then by the union bound, the
\textquotedblleft average\textquotedblright\ output $\rho=\operatorname*{E}%
_{T}\left[  \left\vert \Psi_{T}\right\rangle \left\langle \Psi_{T}\right\vert
\right]  $\ satisfies%
\begin{align*}
1-F\left(  \rho,G\right)   &  \leq1.02\left(  \alpha+\varepsilon\right)
+\exp\left(  -R\alpha^{2}\right) \\
&  \leq0.204\delta+0.51\delta+\exp\left(  -\frac{R\delta^{2}}{25}\right) \\
&  <\delta.
\end{align*}

\end{proof}

Note that our hybrid loses the property of \textit{exponentially-fast}
convergence toward the target subspace $G$, but that property will not be
important for us anyway. \ We leave as an open problem whether there exists a
hybrid algorithm with exponentially-fast convergence.

\section{Formalizing Quantum Money\label{DEF}}

In this section, we first give a formal cryptographic definition of
\textit{public-key quantum money schemes}. \ Our definition is similar to that
of Aaronson \cite{aar:qcopy}. \ However, following \cite{breaking,knots}, we
next define the notion of a \textit{quantum money mini-scheme}, which is
easier to construct and analyze than a full-blown quantum money scheme. \ A
mini-scheme is basically a quantum money scheme where each banknote includes a
classical serial number; where the only security requirement is that producing
a second banknote with \textit{the same serial number} is intractable; and
where there is no public or private key (since given the lax security
requirement, there is no need for one). \ We then prove two general results:
the amplification of weak counterfeiters into strong ones (Theorem
\ref{miniamp}), and the construction of full-blown quantum money schemes from
mini-schemes together with quantumly-secure digital signature schemes (Theorem
\ref{compose}).

\subsection{Quantum Money Schemes\label{SCHEMES}}

Intuitively, a \textit{public-key quantum money scheme} is a scheme by which

\begin{enumerate}
\item[(1)] a trusted \textquotedblleft bank\textquotedblright\ can feasibly
generate an unlimited number of quantum banknotes,

\item[(2)] anyone can feasibly verify a valid banknote as having come from the
bank, but

\item[(3)] no one besides the bank can feasibly map $q=\operatorname*{poly}%
\left(  n\right)  $\ banknotes to $r>q$\ banknotes with any non-negligible
success probability.\footnote{Previously, Aaronson \cite{aar:qcopy}\ required
only that no polynomial-time counterfeiter could increase its
\textit{expected} number of valid banknotes. \ However, the stronger condition
required here is both achievable, and seemingly more natural from the
standpoint of security proofs.}
\end{enumerate}

We now make the notion more formal.

\begin{definition}
[Quantum Money Schemes]\label{moneydef}A \textbf{public-key quantum money
scheme}\ $\mathcal{S}$\ consists of three polynomial-time quantum algorithms:

\begin{itemize}
\item $\mathsf{KeyGen}$, which takes as input a security parameter $0^{n}$,
and probabilistically generates a key pair $\left(  k_{\operatorname*{private}%
},k_{\operatorname*{public}}\right)  $.

\item $\mathsf{Bank}$, which takes as input $k_{\operatorname*{private}}$, and
probabilistically generates a quantum state $\$$\ called a \textbf{banknote}.
\ (Usually $\$$ will be an ordered pair $\left(  s,\rho_{s}\right)  $,
consisting of a classical \textbf{serial number} $s$\ and a \textbf{quantum
money state} $\rho_{s}$, but this is not strictly necessary.)

\item $\mathsf{Ver}$, which takes as input $k_{\operatorname*{public}}$\ and
an alleged banknote $%
\hbox{\rm\rlap/c}%
$, and either accepts or rejects.
\end{itemize}

We say $\mathcal{S}$ has \textbf{completeness error} $\varepsilon$\ if
$\mathsf{Ver}\left(  k_{\operatorname*{public}},\$\right)  $\ accepts with
probability at least $1-\varepsilon$\ for all public keys
$k_{\operatorname*{public}}$ and valid banknotes\ $\$$.\ \ If $\varepsilon
=0$\ then $\mathcal{S}$ has \textbf{perfect completeness}.

Let $\mathsf{Count}$\ (the \textbf{money counter}) take as input
$k_{\operatorname*{public}}$\ as well as a collection of (possibly-entangled)
alleged banknotes $%
\hbox{\rm\rlap/c}%
_{1},\ldots,%
\hbox{\rm\rlap/c}%
_{r}$, and output the number of indices $i\in\left[  r\right]  $\ such that
$\mathsf{Ver}\left(  k_{\operatorname*{public}},%
\hbox{\rm\rlap/c}%
_{i}\right)  $\ accepts. \ Then we say $\mathcal{S}$ has \textbf{soundness
error} $\delta$\ if, given any\ quantum circuit $C\left(
k_{\operatorname*{public}},\$_{1},\ldots,\$_{q}\right)  $ of size
$\operatorname*{poly}\left(  n\right)  $\ (called the \textbf{counterfeiter}),
which maps $q=\operatorname*{poly}\left(  n\right)  $ valid banknotes
$\$_{1},\ldots,\$_{q}$\ \ to $r=\operatorname*{poly}\left(  n\right)  $
(possibly-entangled) alleged banknotes $%
\hbox{\rm\rlap/c}%
_{1},\ldots,%
\hbox{\rm\rlap/c}%
_{r}$,%
\[
\Pr\left[  \mathsf{Count}\left(  k_{\operatorname*{public}},C\left(
k_{\operatorname*{public}},\$_{1},\ldots,\$_{q}\right)  \right)  >q\right]
\leq\delta.
\]
Here the probability is over the key pair $\left(  k_{\operatorname*{private}%
},k_{\operatorname*{public}}\right)  $, valid banknotes $\$_{1},\ldots,\$_{q}%
$\ generated by $\mathsf{Bank}\left(  k_{\operatorname*{private}}\right)  $,
and the behavior of $\mathsf{Count}$\ and $C$.

We call $\mathcal{S}$ \textbf{secure} if it has completeness error $\leq1/3$
and negligible soundness error.
\end{definition}

In Appendix \ref{COMPLETENESS}, we show that the completeness error in any
quantum money scheme can be amplified to $1/2^{\operatorname*{poly}\left(
n\right)  }$, at the cost of only a small increase in the soundness error.
\ Note that, by Lemma \ref{goodasnew} (the \textquotedblleft Almost As Good As
New Lemma\textquotedblright), once we make the completeness error
exponentially small in $n$, we can also give our scheme the property that
\textit{any banknote }$\$$\textit{ can be verified }$\exp\left(  n\right)
$\textit{ times}, before $\$$ gets \textquotedblleft worn
out\textquotedblright\ by repeated measurements. \ This observation is part of
what justifies our use of the term \textquotedblleft money.\textquotedblright%
\footnote{By contrast, BBBW \cite{bbbw}\ introduced the term \textquotedblleft
subway tokens\textquotedblright\ for quantum money states\ that get destroyed
immediately upon verification.}

In this paper, we will often consider \textbf{relativized} quantum money
schemes, which simply means that the three procedures $\mathsf{KeyGen}$,
$\mathsf{Bank}$, $\mathsf{Ver}$---as well as the counterfeiter $C$---all get
access to exactly the same oracle $A:\left\{  0,1\right\}  ^{\ast}%
\rightarrow\left\{  0,1\right\}  $. \ We will also consider relativized
digital signature schemes, etc., which are defined analogously.

A\textbf{ private-key quantum money scheme} is the same as a public-key
scheme, except that the counterfeiter $C$ no longer gets access to
$k_{\operatorname*{public}}$. \ (Thus, we might as well set
$k:=k_{\operatorname*{public}}=k_{\operatorname*{private}}$, since the public
and private keys no longer play separate roles.) \ We call a private-key
scheme \textbf{query-secure}---a notion \textquotedblleft
intermediate\textquotedblright\ between private-key and public-key---if the
counterfeiter $C$\ is allowed to interact repeatedly with the bank.\ \ Given
any alleged banknote $\sigma$, the bank runs the verification procedure
$\mathsf{Ver}\left(  k,\sigma\right)  $, then returns to $C$\ both the
classical result (i.e., accept or reject) \textit{and} the post-measurement
quantum state $\widetilde{\sigma}$.

\subsection{Mini-Schemes\label{MINI}}

While Definition \ref{moneydef} captures our intuitive requirements for a
public-key quantum money scheme, experience has shown that it is cumbersome to
work with in practice. \ So following Lutomirski et al.\ \cite{breaking}\ and
Farhi et al.\ \cite{knots}, in this section we define a simpler primitive
called \textit{mini-schemes}. \ We also prove an amplification theorem for a
large class of mini-schemes. \ Then, in Section \ref{CANONICAL}, we will
explain how mini-schemes can be generically combined with\ conventional
digital signature schemes to create full public-key quantum money schemes.

\begin{definition}
[Mini-Schemes]\label{minischeme}A (public-key)\textbf{ mini-scheme}
$\mathcal{M}$\ consists of two polynomial-time quantum algorithms:

\begin{itemize}
\item $\mathsf{Bank}$, which takes as input a security parameter $0^{n}$, and
probabilistically generates a banknote $\$=\left(  s,\rho_{s}\right)  $, where
$s$ is a classical \textbf{serial number}, and $\rho_{s}$\ is a quantum money state.

\item $\mathsf{Ver}$, which takes as input an alleged banknote $%
\hbox{\rm\rlap/c}%
$, and either accepts or rejects.
\end{itemize}

We say $\mathcal{M}$ has \textbf{completeness error} $\varepsilon$\ if
$\mathsf{Ver}\left(  \$\right)  $\ accepts with probability at least
$1-\varepsilon$\ for all valid banknotes $\$$. \ If $\varepsilon=0$\ then
$\mathcal{M}$ has perfect completeness. \ If, furthermore, $\rho
_{s}=\left\vert \psi_{s}\right\rangle \left\langle \psi_{s}\right\vert $\ is
always a pure state, and $\mathsf{Ver}$ simply consists of a projective
measurement onto the rank-$1$ subspace spanned by $\left\vert \psi
_{s}\right\rangle $, then we say $\mathcal{M}$\ is \textbf{projective}%
.\footnote{We similarly call a full quantum money scheme projective, if
$\mathsf{Ver}\left(  \$\right)  $\ consists of a measurement on one part of
$\$$ in the computational basis, followed by a rank-$1$ projective measurement
on the remaining part.}

Let $\mathsf{Ver}_{2}$\ (the \textbf{double verifier}) take as input a single
serial number $s$ as well as two (possibly-entangled) states $\sigma_{1}$ and
$\sigma_{2}$, and accept if and only $\mathsf{Ver}\left(  s,\sigma_{1}\right)
$\ and $\mathsf{Ver}\left(  s,\sigma_{2}\right)  $\ both accept. \ We say
$\mathcal{M}$ has \textbf{soundness error} $\delta$\ if, given any quantum
circuit $C$ of size $\operatorname*{poly}\left(  n\right)  $\ (the
\textbf{counterfeiter}),\ $\mathsf{Ver}_{2}\left(  s,C\left(  \$\right)
\right)  $\ accepts with probability at most $\delta$. \ Here the probability
is over the banknote $\$$\ output by $\mathsf{Bank}\left(  0^{n}\right)  $, as
well as the behavior of $\mathsf{Ver}_{2}$\ and $C$.

We call $\mathcal{M}$ \textbf{secure} if it has completeness error $\leq
1/3$\ and negligible soundness error.{}
\end{definition}

We observe a simple relationship between Definitions \ref{moneydef}\ and
\ref{minischeme}:

\begin{proposition}
\label{easydir}If there exists a secure public-key money scheme $\mathcal{S}%
=\left(  \mathsf{KeyGen}_{\mathcal{S}},\mathsf{Bank}_{\mathcal{S}%
},\mathsf{Ver}_{\mathcal{S}}\right)  $, then there also exists a secure
mini-scheme $\mathcal{M}=\left(  \mathsf{Bank}_{\mathcal{M}},\mathsf{Ver}%
_{\mathcal{M}}\right)  $.
\end{proposition}

\begin{proof}
Each banknote output by $\mathsf{Bank}_{\mathcal{M}}\left(  0^{n}\right)
$\ will have the form $\left(  k_{\operatorname*{public}},\mathsf{Bank}%
_{\mathcal{S}}\left(  k_{\operatorname*{private}}\right)  \right)  $, where
$\left(  k_{\operatorname*{private}},k_{\operatorname*{public}}\right)  $\ is
a key pair output by $\mathsf{KeyGen}_{\mathcal{S}}\left(  0^{n}\right)  $.
\ Then $\mathsf{Ver}_{\mathcal{M}}\left(  s,\rho_{s}\right)  $\ will accept if
and only if $\mathsf{Ver}_{\mathcal{S}}\left(  s,\rho_{s}\right)  $\ does.
\ Any counterfeiter $C_{\mathcal{M}}$\ against $\mathcal{M}$\ can be converted
directly into a counterfeiter $C_{\mathcal{S}}$\ against $\mathcal{S}$.
\end{proof}

Call a mini-scheme $\mathcal{M}=\left(  \mathsf{Bank},\mathsf{Ver}\right)  $
\textbf{secret-based} if $\mathsf{Bank}$\ works by first generating a
uniformly-random classical string $r$, and then generating a banknote
$\$_{r}:=\left(  s_{r},\rho_{r}\right)  $. \ Intuitively, in a secret-based
scheme, the bank can generate many identical banknotes by simply reusing $r$,
while in a non-secret-based scheme, \textit{not even the bank} might be able
to generate two identical banknotes. \ Here is an interesting observation:

\begin{proposition}
\label{miniowf}If there exists a secure, secret-based mini-scheme, then there
also exists a one-way function\ secure against quantum attack.
\end{proposition}

\begin{proof}
The desired OWF is $\mathsf{SerialNum}\left(  r\right)  :=s_{r}$. \ If there
existed a polynomial-time quantum algorithm\ to\ recover $r$\ given $s_{r}$,
then we could use that algorithm to produce an unlimited number of additional
banknotes\ $\$_{r}$.
\end{proof}

All of the mini-schemes developed in this paper will be secret-based. \ By
contrast, the earlier schemes of Lutomirski et al.\ \cite{breaking}\ and Farhi
et al.\ \cite{knots} are non-secret-based, since the serial number $s$\ is
only obtained as the outcome of a quantum measurement.

The following result is one of the most useful in the paper. \ Intuitively, it
says that in projective mini-schemes, a counterfeiter that copies a banknote
with \textit{any} non-negligible fidelity can be \textquotedblleft
amplified\textquotedblright\ to a counterfeiter that copies the banknote
almost \textit{perfectly}---or conversely, that to rule out the former sort of
counterfeiter, it suffices to rule out the latter. \ The proof makes essential
use of the amplitude amplification results from Section \ref{SEARCH}.

\begin{theorem}
[Amplification of Counterfeiters]\label{miniamp}Let $\mathcal{M}=\left(
\mathsf{Bank},\mathsf{Ver}\right)  $ be a projective mini-scheme, and let
$\$=\left(  s,\rho\right)  $\ be a valid banknote in $\mathcal{M}$. \ Suppose
there exists a counterfeiter $C$ that copies $\$$\ with probability
$\varepsilon>0$:\ that is,%
\[
\Pr\left[  \mathsf{Ver}_{2}\left(  s,C\left(  \$\right)  \right)  \text{
accepts}\right]  \geq\varepsilon.
\]
Then for all $\delta>0$, there is also a modified counterfeiter $C^{\prime}$
(depending only on $\varepsilon$ and $\delta$, not $\$$), which makes%
\[
O\left(  \frac{\log1/\delta}{\sqrt{\varepsilon}\left(  \sqrt{\varepsilon
}+\delta^{2}\right)  }\right)
\]
queries to $C$, $C^{-1}$, and $\mathsf{Ver}$ and which satisfies%
\[
\Pr\left[  \mathsf{Ver}_{2}\left(  s,C^{\prime}\left(  \$\right)  \right)
\text{ accepts}\right]  \geq1-\delta.
\]

\end{theorem}

\begin{proof}
Write $\$$\ as a mixture of pure states:%
\[
\$=\sum p_{i}\left\vert \psi_{i}\right\rangle \left\langle \psi_{i}\right\vert
.
\]
By linearity, clearly it suffices to show that%
\[
\Pr\left[  \mathsf{Ver}_{2}\left(  s,C^{\prime}\left(  \left\vert \psi
_{i}\right\rangle \right)  \right)  \text{ accepts}\right]  \geq1-\delta
\]
for all $i$\ such that $p_{i}>0$. \ We focus on $\left\vert \psi\right\rangle
:=\left\vert \psi_{1}\right\rangle $\ without loss of generality.

By assumption, there exists a subspace $S$\ such that%
\[
\Pr\left[  \mathsf{Ver}\left(  \rho\right)  \text{ accepts}\right]  =F\left(
\rho,S\right)  ^{2}%
\]
for all $\rho$. \ Then $F\left(  \$,S\right)  =F\left(  \left\vert
\psi\right\rangle ,S\right)  =1$.

Now, just as $\mathsf{Ver}$ is simply a projector onto $S$, so $\mathsf{Ver}%
_{2}$\ is a projector onto $S^{\otimes2}$. \ Thus%
\[
F\left(  C\left(  \left\vert \psi\right\rangle \right)  ,S^{\otimes2}\right)
\geq\sqrt{\varepsilon}.
\]
So consider performing a fixed-point Grover search, with $C\left(  \left\vert
\psi\right\rangle \right)  $\ as the initial state and $S^{\otimes2}$\ as the
target subspace. \ By Lemma \ref{fixedpoint}, this will produce a state $\rho
$\ such that $F\left(  \rho,S^{\otimes2}\right)  \geq1-\delta$\ using
$O\left(  \frac{1}{\varepsilon}\log\frac{1}{\delta}\right)  $\ Grover
iterations. \ Each iteration requires a reflection about $C\left(  \left\vert
\psi\right\rangle \right)  $\ and a reflection about $S^{\otimes2}$, which can
be implemented using $O\left(  1\right)  $\ queries to $C,C^{-1}$ and
$\mathsf{Ver}$ respectively. \ Therefore the number of queries to $C,C^{-1}$
and $\mathsf{Ver}$\ is $O\left(  \frac{1}{\varepsilon}\log\frac{1}{\delta
}\right)  $\ as well.

If $\delta$ is large compared to $\varepsilon$, then we can instead use
Theorem \ref{combined}, which produces a state $\rho$\ such that $F\left(
\rho,S^{\otimes2}\right)  \geq1-\delta$\ using $O\left(  \frac{1}%
{\sqrt{\varepsilon}\delta^{2}}\log\frac{1}{\delta}\right)  $\ iterations.
\ Taking the minimum of the two bounds gives us the claimed bound on query complexity.
\end{proof}

Theorem\ \ref{miniamp} is unlikely to hold for \textit{arbitrary}
(non-projective) mini-schemes, for the simple reason that we can always create
a mini-scheme where $\mathsf{Ver}$\ accepts \textit{any} state with some small
nonzero probability $\varepsilon$. \ We leave it as an open problem to find
the largest class of mini-schemes for which Theorem\ \ref{miniamp}\ holds.

\subsection{The Standard Construction\label{CANONICAL}}

Following Lutomirski et al.\ \cite{breaking}\ and Farhi et al.\ \cite{knots},
we can now define a \textquotedblleft standard construction\textquotedblright%
\ of public-key quantum money schemes from mini-schemes and digital signature
schemes. \ Given a mini-scheme $\mathcal{M}=\left(  \mathsf{Bank}%
_{\mathcal{M}},\mathsf{Ver}_{\mathcal{M}}\right)  $, and a signature
$\mathcal{D}=\left(  \mathsf{KeyGen}_{\mathcal{D}},\mathsf{S{}ign}%
_{\mathcal{D}},\mathsf{Ver}_{\mathcal{D}}\right)  $, we define the quantum
money scheme $\mathcal{S}=\left(  \mathsf{KeyGen}_{\mathcal{S}},\mathsf{Bank}%
_{\mathcal{S}},\mathsf{Ver}_{\mathcal{S}}\right)  $ as follows:

\begin{itemize}
\item $\mathsf{KeyGen}_{\mathcal{S}}$ is simply $\mathsf{KeyGen}_{\mathcal{D}%
}$\ from the digital signature scheme.

\item $\mathsf{Bank}_{\mathcal{S}}$ first calls $\mathsf{Bank}_{\mathcal{M}}%
$\ from the mini-scheme to obtain a banknote $\left(  s,\rho\right)  $. \ It
then outputs $\left(  s,\rho\right)  $\ together with a digital signature of
the serial number $s$:%
\[
\mathsf{Bank}_{\mathcal{S}}\left(  k_{\operatorname*{private}}\right)
:=\left(  s,\mathsf{S{}ign}_{\mathcal{D}}\left(  k_{\operatorname*{private}%
},s\right)  ,\rho\right)  .
\]

\item $\mathsf{Ver}_{\mathcal{S}}$ accepts an alleged banknote $\left(
s,w,\sigma\right)  $, if and only if $\mathsf{Ver}_{\mathcal{M}}\left(
s,\sigma\right)  $\ and $\mathsf{Ver}_{\mathcal{D}}\left(
k_{\operatorname*{public}},s,w\right)  $\ both accept.
\end{itemize}

We now prove the above construction's security.

\begin{theorem}
[Security of the Standard Construction]\label{compose}Suppose $\mathcal{M}$ is
a secure mini-scheme, and $\mathcal{D}$\ is a digital signature scheme secure
against quantum chosen-message attacks. \ Then $\mathcal{S}$\ is a secure
public-key quantum money scheme.
\end{theorem}

\begin{proof}
The intuition behind the proof is extremely simple: by requiring digital
signatures for the serial numbers, we can force a counterfeiter to copy one of
its \textit{existing} banknotes, rather than creating a new banknote with a
new serial number. \ In this way, we force the counterfeiter to break the
underlying mini-scheme $\mathcal{M}$, rather than doing an \textquotedblleft
end run\textquotedblright\ around $\mathcal{M}$.

To formalize this intuition, suppose there exists a counterfeiter
$C_{\mathcal{S}}$\ against $\mathcal{S}$: that is, a polynomial-time quantum
algorithm such that%
\[
\Pr\left[  \mathsf{Count}\left(  k_{\operatorname*{public}},C_{\mathcal{S}%
}\left(  k_{\operatorname*{public}},\$_{1},\ldots,\$_{q}\right)  \right)
>q\right]  \geq\frac{1}{p\left(  n\right)  }.
\]
Here $\$_{i}:=\left(  s_{i},w_{i},\rho_{i}\right)  $\ is a valid banknote,
$\mathsf{Count}$\ is the money counter from Definition \ref{moneydef},\ and
$p$\ is some polynomial. \ Also, the probability is over the key pair $\left(
k_{\operatorname*{private}},k_{\operatorname*{public}}\right)  $, the valid
banknotes $\$_{1},\ldots,\$_{q}$, and the behavior of $\mathsf{Count}$\ and
$C_{\mathcal{S}}$. \ Suppose further that $\mathcal{D}$\ is secure. \ Then it
suffices to show that, by using $C_{\mathcal{S}}$, we can construct a
counterfeiter $C_{\mathcal{M}}$\ against the underlying mini-scheme
$\mathcal{M}$.

Let $\mathsf{New}\left(  k_{\operatorname*{public}},\$_{1},\ldots
,\$_{q}\right)  $\ be an algorithm that does the following:

\begin{enumerate}
\item[(1)] Records the serial numbers $s_{1},\ldots,s_{q}$\ of\ $\$_{1}%
,\ldots,\$_{q}$, and lets $U:=\left\{  s_{1},\ldots,s_{q}\right\}  $.

\item[(2)] Runs $C_{\mathcal{S}}\left(  k_{\operatorname*{public}}%
,\$_{1},\ldots,\$_{q}\right)  $, and examines the output states $%
\hbox{\rm\rlap/c}%
_{1},\ldots,%
\hbox{\rm\rlap/c}%
_{r}$.

\item[(3)] Returns the number of $i\in\left[  r\right]  $\ such that
$\mathsf{Ver}_{\mathcal{S}}\left(
\hbox{\rm\rlap/c}%
_{i}\right)  $ accepts, \textit{and} $%
\hbox{\rm\rlap/c}%
_{i}$'s serial number $s_{i}^{\prime}$ does not belong to $U$.
\end{enumerate}

Then we claim that $\Pr\left[  \mathsf{New}\left(  k_{\operatorname*{public}%
},\$_{1},\ldots,\$_{q}\right)  >0\right]  $\ is negligibly small, where the
probability is over the same variables as before. \ The proof is simply that,
if this were not so, then we could easily create a counterfeiter
$C_{\mathcal{D}}$\ against the digital signature scheme $\mathcal{D}$. \ With
non-negligible probability, $C_{\mathcal{D}}$\ would generate a valid
signature $\mathsf{S{}ign}_{\mathcal{D}}\left(  k_{\operatorname*{private}%
},s_{i}^{\prime}\right)  $, for a message $s_{i}^{\prime}$\ for which
$C_{\mathcal{D}}$ had never before seen a valid signature, by running
$C_{\mathcal{S}}\left(  k_{\operatorname*{public}},\$_{1},\ldots
,\$_{q}\right)  $, then measuring $%
\hbox{\rm\rlap/c}%
_{i}=\left(  s_{i}^{\prime},w_{i}^{\prime},\rho_{i}^{\prime}\right)  $\ for a
uniformly random $i\in\left[  r\right]  $. \ (Note that $C_{\mathcal{D}}$\ can
generate $q$ money states\ $\$_{1},\ldots,\$_{q}$, without knowledge of
$k_{\operatorname*{private}}$, by generating the $s_{i}$'s and $\rho_{i}%
$'s\ on its own, then calling the signing oracle $\mathcal{O}$\ to get the
$w_{i}$'s.)

But now we can define a counterfeiter $C_{\mathcal{M}}$\ against the
mini-scheme $\mathcal{M}$, which works as follows:

\begin{itemize}
\item[(i)] Run $\mathsf{KeyGen}_{\mathcal{D}}\left(  0^{n}\right)  $, to
generate a \textit{new} key pair $\left(  k_{\operatorname*{private}}^{\prime
},k_{\operatorname*{public}}^{\prime}\right)  $.

\item[(ii)] Label the banknote to be copied $\left(  s_{\ell},\rho_{\ell
}\right)  $, for some $\ell\in\left[  q\right]  $\ chosen uniformly at random.

\item[(iii)] Repeatedly call $\mathsf{Bank}_{\mathcal{M}}\left(  0^{n}\right)
$ to generate $q-1$\ serial numbers and quantum money states, labeled $\left(
s_{i},\rho_{i}\right)  $ for all $i\in\left[  q\right]  \setminus\left\{
\ell\right\}  $. \ Let $U:=\left\{  s_{1},\ldots,s_{q}\right\}  $.

\item[(iv)] Generate a digital signature $w_{i}:=\mathsf{S{}ign}_{\mathcal{D}%
}\left(  k_{\operatorname*{private}}^{\prime},s_{i}\right)  $\ for each
$i\in\left[  q\right]  $. \ Let $\$_{i}:=\left(  s_{i},w_{i},\rho_{i}\right)
$.

\item[(v)] Run the counterfeiter $C_{\mathcal{S}}\left(
k_{\operatorname*{public}},\$_{1},\ldots,\$_{q}\right)  $, to obtain
$r>q$\ alleged banknotes $%
\hbox{\rm\rlap/c}%
_{1},\ldots,%
\hbox{\rm\rlap/c}%
_{r}$\ where $%
\hbox{\rm\rlap/c}%
_{j}=\left(  s_{j}^{\prime},w_{j}^{\prime},\rho_{j}^{\prime}\right)  $.

\item[(vi)] Choose $j,k\in\left[  r\right]  $\ uniformly at random without
replacement, and output $\left(  \rho_{j}^{\prime},\rho_{k}^{\prime}\right)  $
as a candidate for two copies of $\rho_{\ell}$.
\end{itemize}

Suppose that $\mathsf{Count}>q$, as happens with probability at least
$\frac{1}{p\left(  n\right)  }$. \ Also suppose that $\mathsf{New}=0$, as
happens all but a negligible fraction of the time. \ Then by the pigeonhole
principle, there must exist indices $j\neq k$\ such that $s_{j}^{\prime}%
=s_{k}^{\prime}$. \ With probability at least $1/\binom{r}{2}$, the
counterfeiter $C_{\mathcal{M}}$\ will find such a $\left(  j,k\right)  $ pair.
\ Therefore $C_{\mathcal{M}}$\ succeeds with overall probability
$\Omega\left(  1/\operatorname*{poly}\left(  n\right)  \right)  $.
\end{proof}

Theorem \ref{compose} reduces the construction of a public-key quantum money
scheme to two \textquotedblleft smaller\textquotedblright\ problems:
constructing a mini-scheme, and constructing a signature scheme secure against
quantum attacks. \ In practice, however, the situation is even better, since
in this paper, all of our constructions of mini-schemes will \textit{also}
yield signature schemes \textquotedblleft free of charge\textquotedblright!
\ The following proposition explains why:

\begin{proposition}
\label{freeofcharge}If there exists a secure, secret-based mini-scheme
$\mathcal{M}$, then there also exists a secure public-key quantum money scheme
$\mathcal{S}$.
\end{proposition}

\begin{proof}
Starting from $\mathcal{M}$, we can get a one-way function secure against
quantum attack from Proposition \ref{miniowf}, and hence a digital signature
scheme $\mathcal{D}$\ secure against quantum chosen-message attack from
Theorem \ref{rompelthm}. \ Combining $\mathcal{M}$\ and\ $\mathcal{D}$\ now
yields $\mathcal{S}$\ by Theorem \ref{compose}.
\end{proof}

Finally, let us make explicit what Theorem \ref{compose}\ means for oracle construction.

\begin{corollary}
\label{minitofull}Suppose there exists a mini-scheme $\mathcal{M}$\ that is
provably secure relative to some oracle $A_{\mathcal{M}}$ (i.e., any
counterfeiter $C_{\mathcal{M}}$\ against $\mathcal{M}$\ must make
superpolynomially many queries to $A_{\mathcal{M}}$). \ Then there exists a
public-key quantum money scheme $\mathcal{S}$\ that is provably secure
relative to some other oracle $A_{\mathcal{S}}$.
\end{corollary}

\begin{proof}
By Theorem \ref{sigoracle},\ relative to a suitable oracle $A_{\mathcal{D}}$
(in fact, a \textit{random} oracle suffices), there exists a signature scheme
$\mathcal{D}$, such that any quantum chosen-message attack against
$\mathcal{D}$\ must make superpolynomially many queries to $A_{\mathcal{D}}$.
The oracle $A_{\mathcal{S}}$ will simply be a concatenation of $A_{\mathcal{M}%
}$ with $A_{\mathcal{D}}$. \ Relative to $A_{\mathcal{S}}$, we claim that the
mini-scheme $\mathcal{M}$ and signature scheme $\mathcal{D}$ are \textit{both}
secure---and therefore, by Theorem \ref{compose}, we can construct a secure
public-key quantum money scheme $\mathcal{S}$.

The only worry is that a counterfeiter $C_{\mathcal{M}}$\ against
$\mathcal{M}$\ might gain some advantage by querying $A_{\mathcal{D}}$; or
conversely, a counterfeiter $C_{\mathcal{D}}$\ against $\mathcal{D}$\ might
gain some advantage by querying $A_{\mathcal{M}}$. \ However, this worry is
illusory, for the simple reason that the oracles $A_{\mathcal{D}}$ and
$A_{\mathcal{M}}$ are generated independently. \ Thus, if $C_{\mathcal{M}}%
$\ can break $\mathcal{M}$\ by querying $A_{\mathcal{D}}$, then it can
\textit{also} break $\mathcal{M}$\ by querying a randomly-generated
\textquotedblleft mock-up\textquotedblright\ $A_{\mathcal{D}}^{\prime}$\ of
$A_{\mathcal{D}}$; and conversely, if $C_{\mathcal{D}}$\ can break
$\mathcal{D}$\ by querying $A_{\mathcal{M}}$, then it can also break
$\mathcal{D}$\ by querying a randomly-generated mock-up\ $A_{\mathcal{M}%
}^{\prime}$\ of $A_{\mathcal{M}}$. \ Regardless of the \textit{computational}
cost of generating these mock-ups, they give us a break against $\mathcal{D}%
$\ or $\mathcal{M}$\ that makes only $\operatorname*{poly}\left(  n\right)
$\ oracle queries, thereby giving the desired contradiction.
\end{proof}

\section{Inner-Product Adversary Method\label{METHOD}}

At least in the black-box setting, our goal is to create quantum money
(mini-)schemes that we can \textit{prove} are secure---by showing that any
counterfeiter would need to make exponentially many queries to some oracle.
\ Proving security results of this kind turns out to require interesting
quantum lower bound machinery. \ In this section, we introduce the
\textit{inner-product adversary method}, a new variant of Ambainis's quantum
adversary method \cite{ambainis} that is\ well-adapted to proving the security
of quantum money schemes, and that seems likely to find other applications.

Let us explain the difficulty we need to overcome. \ In a public-key quantum
money scheme, a counterfeiter $C$ has \textit{two} powerful resources available:

\begin{enumerate}
\item[(1)] One or more copies of a \textquotedblleft
legitimate\textquotedblright\ quantum money state $\left\vert \psi
\right\rangle $.

\item[(2)] Access to a \textit{verification procedure} $V$, which accepts
$\left\vert \psi\right\rangle $ and rejects every state orthogonal to
$\left\vert \psi\right\rangle $.
\end{enumerate}

Indeed, for us, the situation is even better for $C$ (i.e., worse for us!),
since $C$ can query not only the verification procedure $V$ itself, but also
an underlying \textit{classical} oracle $U$\ that the legitimate buyers and
sellers use to implement $V$. \ But let us ignore that issue for now.

As a first step, of course, we should understand how to rule out
counterfeiting given (1) or (2) separately. \ If $C$\ has a copy of
$\left\vert \psi\right\rangle $, but no oracle access to $V$,\ then the
impossibility of preparing $\left\vert \psi\right\rangle \left\vert
\psi\right\rangle $\ essentially amounts to the No-Cloning
Theorem.\ \ Conversely, if $C$ has oracle access to $V$, but no copy of
$\left\vert \psi\right\rangle $,\ then given unlimited time, $C$ \textit{can}
prepare as many\ copies of $\left\vert \psi\right\rangle $\ as it wants,\ by
using Grover's algorithm to search for a quantum state that $V$\ accepts.
\ The problem is \textquotedblleft merely\textquotedblright\ that, if
$\left\vert \psi\right\rangle $\ has $n$ qubits, then Grover's algorithm
requires $\Theta\left(  2^{n/2}\right)  $ iterations, and the BBBV hybrid
argument \cite{bbbv} shows that Grover's algorithm is optimal.

What we need, then, is a theorem showing that any counterfeiter needs
exponentially many queries to $V$ to prepare $\left\vert \psi\right\rangle
\left\vert \psi\right\rangle $,\ \textit{even if} the counterfeiter has a copy
of $\left\vert \psi\right\rangle $\ to start with. \ Such a theorem would
contain both the No-Cloning Theorem and the BBBV hybrid argument as special
cases. \ Aaronson \cite{aar:qcopy} called the desired generalization the
\textit{Complexity-Theoretic No-Cloning Theorem}, and sketched a proof of it
using Ambainis's adversary method. \ Based on that result, Aaronson also
argued that there exists a \textit{quantum oracle} (i.e., a black-box unitary
transformation $V$) relative to which secure public-key quantum money is
possible. \ However, the details were never published.

In this section, we prove a result---Theorem \ref{lbthm}---that is much more
general than Aaronson's previous Complexity-Theoretic No-Cloning Theorem
\cite{aar:qcopy}. \ Then, in Section \ref{COR}, we apply Theorem \ref{lbthm}
to prove the security of public-key quantum money relative to a
\textit{classical} oracle. \ In Appendix \ref{QOR}, we also apply Theorem
\ref{lbthm} to prove the \textquotedblleft original\textquotedblright%
\ Complexity-Theoretic No-Cloning Theorem \cite{aar:qcopy}, which involves
Haar-random $n$-qubit states $\left\vert \psi\right\rangle $, rather than
superpositions $\left\vert A\right\rangle $\ over subspaces $A\leq
\mathbb{F}_{2}^{n}$.\footnote{For whatever it is worth, we get a lower bound
of $\Omega\left(  2^{n/2}\right)  $\ on the number of queries needed to copy a
Haar-random state, which is quadratically better than the $\Omega\left(
2^{n/4}\right)  $\ that we get for subspace states.}

\subsection{Idea of Method\label{IDEA}}

So, what \textit{is} the inner-product adversary method? \ In Ambainis's
adversary method \cite{ambainis}---like in the BBBV hybrid argument
\cite{bbbv}\ from which it evolved---the basic idea is to upper-bound how much
\textquotedblleft progress\textquotedblright\ a quantum algorithm $Q$\ can
make at distinguishing pairs of oracles, as the result of a single query.
\ Let $\left\vert \Psi_{t}^{U}\right\rangle $\ be $Q$'s state after $t$
queries, assuming that the oracle is $U$. \ Then \textit{normally}, before any
queries have been made, we can assume that $\left\vert \Psi_{0}^{U}%
\right\rangle =\left\vert \Psi_{0}^{V}\right\rangle $\ for all oracles $U$ and
$V$. \ By contrast, after the final query $T$, for all oracle pairs $\left(
U,V\right)  $\ that $Q$ is trying to distinguish, we must have (say)
$\left\vert \left\langle \Psi_{T}^{U}|\Psi_{T}^{V}\right\rangle \right\vert
\leq1/2$. \ Thus, \textit{if} we can show that the inner product $\left\vert
\left\langle \Psi_{t}^{U}|\Psi_{t}^{V}\right\rangle \right\vert $\ can
decrease by at most $\varepsilon$\ as the result of a single query, then it
follows that $Q$ must make $\Omega\left(  1/\varepsilon\right)  $\ queries.

But when we try to apply the above framework to quantum money, we run into
serious difficulties. \ Most obviously, it is no longer true that $\left\vert
\Psi_{0}^{U}\right\rangle =\left\vert \Psi_{0}^{V}\right\rangle $\ for all
oracles $U,V$. \ Indeed, before $Q$ makes even a \textit{single} query to its
oracle $V$, it already has a great deal of information about $V$, in the form
of a legitimate money state $\left\vert \psi\right\rangle $\ that
$V$\ accepts. \ The task is \textquotedblleft merely\textquotedblright\ to
prepare a second copy of a state that $Q$\ already has! \ Worse yet, once we
fix two oracles $U$\ and $V$, we find that $Q$ generally \textit{can} exploit
the \textquotedblleft head start\textquotedblright\ provided by its initial
state to decrease the inner product $\left\vert \left\langle \Psi_{t}^{U}%
|\Psi_{t}^{V}\right\rangle \right\vert $\ by a constant amount, by making just
a single query to $U$ or $V$ respectively.

Our solution is as follows. \ We first carefully choose a distribution
$\mathcal{D}$\ over oracle pairs $\left(  U,V\right)  $. We then analyze how
much the \textit{expected} inner product%
\[
\operatorname*{E}_{\left(  U,V\right)  \thicksim\mathcal{D}}\left[  \left\vert
\left\langle \Psi_{t}^{U}|\Psi_{t}^{V}\right\rangle \right\vert \right]
\]
can decrease as the result of a single query to $U$ or $V$. \ We will find
that, even if $Q$ can substantially decrease the inner product between
$\left\vert \Psi_{t}^{U}\right\rangle $ and $\left\vert \Psi_{t}%
^{V}\right\rangle $\ for \textit{some} $\left(  U,V\right)  $\ pairs by making
a single query, it cannot do so for \textit{most} pairs.

To illustrate, let $\left\vert \psi\right\rangle $\ and $\left\vert
\varphi\right\rangle $\ be two possible quantum money states, which satisfy
(say) $\left\langle \psi|\varphi\right\rangle =1/2$. \ Then if a
counterfeiting algorithm succeeds perfectly, it must map $\left\vert
\psi\right\rangle $\ to $\left\vert \psi\right\rangle ^{\otimes2}$, and
$\left\vert \varphi\right\rangle $\ to $\left\vert \varphi\right\rangle
^{\otimes2}$. \ Since%
\[
\left\langle \psi\right\vert ^{\otimes2}\left\vert \varphi\right\rangle
^{\otimes2}=\left(  \left\langle \psi|\varphi\right\rangle \right)  ^{2}%
=\frac{1}{4},
\]
this means that the counterfeiter must \textit{decrease the corresponding
inner product }by at least\textit{ }$1/4$. \ However, we will show that the
\textit{average} inner product can decrease by at most $1/\exp\left(
n\right)  $ as the result of a single query. \ From this it will follow that
the counterfeiter needs to make $2^{\Omega\left(  n\right)  }$\ queries.

Let us mention that today, there are several \textquotedblleft
sophisticated\textquotedblright\ versions of the quantum adversary method
\cite{amrr,lmrss},\ which \textit{can} yield lower bounds for quantum state
generation tasks not unlike the ones we consider. \ However, a drawback of
these methods is that they are extremely hard to apply to concrete problems:
doing so typically requires eigenvalue bounds, and often the use of
representation theory. \ For this reason, even if one of the \textquotedblleft
sophisticated\textquotedblright\ adversary methods (or a variant thereof)
could be applied to the quantum money problem, our approach might still be preferable.

\subsection{The Method\label{THEMETHOD}}

We now introduce the inner-product adversary method. \ Let $\mathcal{O}$\ be a
set of quantum oracles acting on $n$ qubits each. \ For each $U\in\mathcal{O}%
$, assume there exists a subspace $S_{U}\leq\mathbb{C}^{2^{n}}$\ such that

\begin{enumerate}
\item[(i)] $U\left\vert \psi\right\rangle =-\left\vert \psi\right\rangle
$\ for all $\left\vert \psi\right\rangle \in S_{U}$, and

\item[(ii)] $U\left\vert \eta\right\rangle =\left\vert \eta\right\rangle
$\ for all $\left\vert \eta\right\rangle \in S_{U}^{\bot}$.
\end{enumerate}

Let $R\subset\mathcal{O}\times\mathcal{O}$ be a symmetric binary relation on
$\mathcal{O}$, with the properties that

\begin{enumerate}
\item[(i)] $\left(  U,U\right)  \notin R$\ for all $U\in\mathcal{O}$, and

\item[(ii)] for every $U\in\mathcal{O}$\ there exists a $V\in\mathcal{O}%
$\ such that $\left(  U,V\right)  \in R$.\
\end{enumerate}

Suppose that for all $U\in\mathcal{O}$\ and all $\left\vert \eta\right\rangle
\in S_{U}^{\bot}$, we have%
\[
\operatorname*{E}_{V~:~\left(  U,V\right)  \in R}\left[  F\left(  \left\vert
\eta\right\rangle ,S_{V}\right)  ^{2}\right]  \leq\varepsilon,
\]
where $F\left(  \left\vert \eta\right\rangle ,S_{V}\right)  =\max_{\left\vert
\psi\right\rangle \in S_{V}}\left\vert \left\langle \eta|\psi\right\rangle
\right\vert $\ is the fidelity between $\left\vert \eta\right\rangle $\ and
$S_{V}$. \ Let $Q$\ be a quantum oracle algorithm, and let $Q^{U}$\ denote $Q$
run with the oracle $U\in\mathcal{O}$. \ Suppose $Q^{U}$ begins in the state
$\left\vert \Psi_{0}^{U}\right\rangle $\ (possibly already dependent on $U$).
\ Let $\left\vert \Psi_{t}^{U}\right\rangle $\ denote the state of $Q^{U}%
$\ immediately after the $t^{th}$\ query. \ Also, define a \textit{progress
measure} $p_{t}$\ by%
\[
p_{t}:=\operatorname*{E}_{U,V~:~\left(  U,V\right)  \in R}\left[  \left\vert
\left\langle \Psi_{t}^{U}|\Psi_{t}^{V}\right\rangle \right\vert \right]  .
\]
The following lemma bounds how much $p_{t}$\ can decrease as the result of a
single query.

\begin{lemma}
[Bound on Progress Rate]\label{innerprod}%
\[
p_{t}\geq p_{t-1}-4\sqrt{\varepsilon}.
\]

\end{lemma}

\begin{proof}
Let $\left\vert \Phi_{t}^{U}\right\rangle $\ denote the state of $Q^{U}%
$\ immediately \textit{before} the $t^{th}$\ query. \ Then for all $t$, it is
clear that $\left\langle \Phi_{t}^{U}|\Phi_{t}^{V}\right\rangle =\left\langle
\Psi_{t-1}^{U}|\Psi_{t-1}^{V}\right\rangle $: in other words, the unitary
transformations that $Q$\ performs in between query steps have no effect on
the inner products. \ So to prove the lemma, it suffices to show the following
inequality:%
\begin{equation}
\operatorname*{E}_{U,V~:~\left(  U,V\right)  \in R}\left[  \left\vert
\left\langle \Phi_{t}^{U}|\Phi_{t}^{V}\right\rangle \right\vert \right]
-\operatorname*{E}_{U,V~:~\left(  U,V\right)  \in R}\left[  \left\vert
\left\langle \Psi_{t}^{U}|\Psi_{t}^{V}\right\rangle \right\vert \right]
\leq4\sqrt{\varepsilon}. \tag{*}%
\end{equation}
Let $\left\{  \left\vert i\right\rangle \right\}  _{i\in\left[  B\right]  }%
$\ be an arbitrary orthonormal basis for $Q$'s workspace register. \ Then we
can write%
\begin{align*}
\left\vert \Phi_{t}^{U}\right\rangle  &  =\sum_{i\in\left[  B\right]  }%
\alpha_{t,i}^{U}\left\vert i\right\rangle \left\vert \Phi_{t,i}^{U}%
\right\rangle \\
&  =\sum_{i\in\left[  B\right]  }\left\vert i\right\rangle \left(  \beta
_{t,i}^{U}\left\vert \eta_{t,i}^{U}\right\rangle +\gamma_{t,i}^{U}\left\vert
\psi_{t,i}^{U}\right\rangle \right)  ,
\end{align*}
where $\left\vert \eta_{t,i}^{U}\right\rangle \in S_{U}^{\bot}$\ and
$\left\vert \psi_{t,i}^{U}\right\rangle \in S_{U}$. \ (By normalization,
$\left\vert \beta_{t,i}^{U}\right\vert ^{2}+\left\vert \gamma_{t,i}%
^{U}\right\vert ^{2}=\left\vert \alpha_{t,i}^{U}\right\vert ^{2}$.) \ A query
transforms the above state to%
\[
\left\vert \Psi_{t}^{U}\right\rangle =\sum_{i\in\left[  B\right]  }\left\vert
i\right\rangle \left(  \beta_{t,i}^{U}\left\vert \eta_{t,i}^{U}\right\rangle
-\gamma_{t,i}^{U}\left\vert \psi_{t,i}^{U}\right\rangle \right)  .
\]
So for all $U,V\in\mathcal{O}$,%
\begin{align*}
\left\langle \Phi_{t}^{U}|\Phi_{t}^{V}\right\rangle -\left\langle \Psi_{t}%
^{U}|\Psi_{t}^{V}\right\rangle  &  =\sum_{i\in\left[  B\right]  }\left(
\overline{\beta}_{t,i}^{U}\left\langle \eta_{t,i}^{U}\right\vert
+\overline{\gamma}_{t,i}^{U}\left\langle \psi_{t,i}^{U}\right\vert \right)
\left(  \beta_{t,i}^{V}\left\vert \eta_{t,i}^{V}\right\rangle +\gamma
_{t,i}^{V}\left\vert \psi_{t,i}^{V}\right\rangle \right) \\
&  ~~~~~~~~~~-\sum_{i\in\left[  B\right]  }\left(  \overline{\beta}_{t,i}%
^{U}\left\langle \eta_{t,i}^{U}\right\vert -\overline{\gamma}_{t,i}%
^{U}\left\langle \psi_{t,i}^{U}\right\vert \right)  \left(  \beta_{t,i}%
^{V}\left\vert \eta_{t,i}^{V}\right\rangle -\gamma_{t,i}^{V}\left\vert
\psi_{t,i}^{V}\right\rangle \right) \\
&  =2\sum_{i\in\left[  B\right]  }\left(  \overline{\beta}_{t,i}^{U}%
\gamma_{t,i}^{V}\left\langle \eta_{t,i}^{U}|\psi_{t,i}^{V}\right\rangle
+\overline{\gamma}_{t,i}^{U}\beta_{t,i}^{V}\left\langle \psi_{t,i}^{U}%
|\eta_{t,i}^{V}\right\rangle \right)  .
\end{align*}
By Cauchy-Schwarz, the above implies that%
\[
\left\vert \left\langle \Phi_{t}^{U}|\Phi_{t}^{V}\right\rangle \right\vert
-\left\vert \left\langle \Psi_{t}^{U}|\Psi_{t}^{V}\right\rangle \right\vert
\leq2\max_{i\in\left[  B\right]  }\left\vert \left\langle \eta_{t,i}^{U}%
|\psi_{t,i}^{V}\right\rangle \right\vert +2\max_{i\in\left[  B\right]
}\left\vert \left\langle \psi_{t,i}^{U}|\eta_{t,i}^{V}\right\rangle
\right\vert .
\]
Now fix $U\in\mathcal{O}$\ and $i\in\left[  B\right]  $. \ Then again applying
Cauchy-Schwarz,%
\begin{align*}
\operatorname*{E}_{V~:~\left(  U,V\right)  \in R}\left[  \left\vert
\left\langle \eta_{t,i}^{U}|\psi_{t,i}^{V}\right\rangle \right\vert \right]
&  \leq\sqrt{\operatorname*{E}_{V~:~\left(  U,V\right)  \in R}\left[
\left\vert \left\langle \eta_{t,i}^{U}|\psi_{t,i}^{V}\right\rangle \right\vert
^{2}\right]  }\\
&  \leq\sqrt{\operatorname*{E}_{V~:~\left(  U,V\right)  \in R}\left[
\max_{\left\vert \psi\right\rangle \in S_{V}}\left\vert \left\langle
\eta_{t,i}^{U}|\psi\right\rangle \right\vert ^{2}\right]  }\\
&  \leq\sqrt{\varepsilon}.
\end{align*}
Hence%
\[
\operatorname*{E}_{U,V~:~\left(  U,V\right)  \in R}\left[  \left\vert
\left\langle \eta_{t,i}^{U}|\psi_{t,i}^{V}\right\rangle \right\vert \right]
\leq\sqrt{\varepsilon}%
\]
as well, and likewise%
\[
\operatorname*{E}_{U,V~:~\left(  U,V\right)  \in R}\left[  \left\vert
\left\langle \psi_{t,i}^{U}|\eta_{t,i}^{V}\right\rangle \right\vert \right]
\leq\sqrt{\varepsilon}%
\]
by symmetry. \ Putting everything together,%
\begin{align*}
p_{t-1}-p_{t}  &  =\operatorname*{E}_{U,V~:~\left(  U,V\right)  \in R}\left[
\left\vert \left\langle \Phi_{t}^{U}|\Phi_{t}^{V}\right\rangle \right\vert
-\left\vert \left\langle \Psi_{t}^{U}|\Psi_{t}^{V}\right\rangle \right\vert
\right] \\
&  \leq2\operatorname*{E}_{U,V~:~\left(  U,V\right)  \in R}\left[  \max
_{i\in\left[  B\right]  }\left\vert \left\langle \eta_{t,i}^{U}|\psi_{t,i}%
^{V}\right\rangle \right\vert \right]  +2\operatorname*{E}_{U,V~:~\left(
U,V\right)  \in R}\left[  \max_{i\in\left[  B\right]  }\left\vert \left\langle
\psi_{t,i}^{U}|\eta_{t,i}^{V}\right\rangle \right\vert \right] \\
&  \leq4\sqrt{\varepsilon}.
\end{align*}
This proves inequality (*) and hence the lemma.
\end{proof}

From Lemma \ref{innerprod}\ we immediately deduce the following.

\begin{theorem}
[Inner-Product Adversary Method]\label{lbthm}Suppose that initially
$\left\vert \left\langle \Psi_{0}^{U}|\Psi_{0}^{V}\right\rangle \right\vert
\geq c$\ for all $\left(  U,V\right)  \in R$, whereas by the end we need
$\left\vert \left\langle \Psi_{T}^{U}|\Psi_{T}^{V}\right\rangle \right\vert
\leq d$\ for all $\left(  U,V\right)  \in R$. \ Then $Q$\ must make
$T=\Omega\left(  \frac{c-d}{\sqrt{\varepsilon}}\right)  $\ oracle queries.
\end{theorem}

\section{Classical Oracle Scheme\label{COR}}

In this section, we construct a mini-scheme, called the \textit{Hidden
Subspace Mini-Scheme}, that requires only a classical oracle. \ We then use
the inner-product adversary method from Section \ref{METHOD}\ to show that our
mini-scheme is secure---indeed, that any counterfeiter must make
$\Omega\left(  2^{n/4}\right)  $\ queries to copy a banknote. \ By the results
of Sections \ref{CANONICAL}\ and \ref{CRYPTO}, our mini-scheme will
automatically imply a full-blown public-key quantum money scheme, which
requires only a classical oracle and is unconditionally secure.

\subsection{The Hidden Subspace Mini-Scheme\label{HSM}}

We identify $n$-bit strings $x\in\left\{  0,1\right\}  ^{n}$\ with elements of
the vector space $\mathbb{F}_{2}^{n}$\ in the standard way. \ Then in our
mini-scheme, each $n$-qubit money state will have the form%
\[
\left\vert A\right\rangle :=\frac{1}{\sqrt{\left\vert A\right\vert }}%
\sum_{x\in A}\left\vert x\right\rangle ,
\]
where $A$ is some randomly-chosen subspace of $\mathbb{F}_{2}^{n}$\ (i.e., a
set of codewords of a linear code),\ with $\dim A=n/2$. \ Let $A^{\bot}$\ be
the orthogonal complement of $A$, so that $\dim A^{\bot}=n/2$ as
well.\ \ Notice that we can transform $\left\vert A\right\rangle $\ to
$\left\vert A^{\bot}\right\rangle $\ and vice versa by simply applying
$H_{2}^{\otimes n}$: a Hadamard gate on each of the $n$ qubits, or
equivalently a quantum Fourier transform over $\mathbb{F}_{2}^{n}$.

The basic idea of the mini-scheme is as follows: the bank can easily prepare
the quantum money state $\left\vert A\right\rangle $, starting from a
classical description $\left\langle A\right\rangle $\ of $A$\ (e.g., a list of
$n/2$\ generators). \ The bank distributes the state $\left\vert
A\right\rangle $, but keeps the classical description $\left\langle
A\right\rangle $\ secret. \ Along with $\left\vert A\right\rangle $\ itself,
the bank also publishes details of how to \textit{verify} $\left\vert
A\right\rangle $ by querying two classical oracles, $U_{A}$\ and $U_{A^{\bot}%
}$. \ The first oracle, $U_{A}$, decides membership in $A$: for all $n$-qubit
basis states $\left\vert x\right\rangle $,%
\[
U_{A}\left\vert x\right\rangle =\left\{
\begin{array}
[c]{cc}%
-\left\vert x\right\rangle  & \text{if }x\in A\\
\left\vert x\right\rangle  & \text{otherwise}%
\end{array}
\right.
\]
The second oracle, $U_{A^{\bot}}$, decides membership in $A^{\bot}$ in the
same way.

Using $U_{A}$, it is easy to implement a projector $\mathbb{P}_{A}$\ onto the
set of basis states in $A$. \ To do so, simply initialize a control qubit to
$\left\vert +\right\rangle =\frac{\left\vert 0\right\rangle +\left\vert
1\right\rangle }{\sqrt{2}}$, then apply $U_{A}$\ conditioned on the control
qubit being in state $\left\vert 1\right\rangle $, then measure the control
qubit in the $\left\{  \left\vert +\right\rangle ,\left\vert -\right\rangle
\right\}  $ basis, and postselect on getting the outcome $\left\vert
-\right\rangle $. \ Likewise, using $U_{A^{\bot}}$, it is easy to implement a
projector $\mathbb{P}_{A^{\bot}}$\ onto the set of basis states in $A^{\bot}$.
\ Then $V_{A}$, the public verification algorithm for the money state
$\left\vert A\right\rangle $, will simply consist of $\mathbb{P}_{A}$, then a
Fourier transform, then $\mathbb{P}_{A^{\bot}}$, and finally a second Fourier
transform to return the legitimate money state back to $\left\vert
A\right\rangle $:%
\[
V_{A}:=H_{2}^{\otimes n}\mathbb{P}_{A^{\bot}}H_{2}^{\otimes n}\mathbb{P}_{A}.
\]
We show in Lemma \ref{testworks}\ that $V_{A}$\ is just a projector onto
$\left\vert A\right\rangle $. \ This means, in particular, that $V_{A}%
\left\vert A\right\rangle =\left\vert A\right\rangle $, and that $V_{A}%
$\ accepts an arbitrary state $\left\vert \psi\right\rangle $\ with
probability $\left\vert \left\langle \psi|A\right\rangle \right\vert ^{2}$.
\ Thus, our mini-scheme is \textit{projective} and has \textit{perfect
completeness}.

But what about security? \ Intuitively, a counterfeiter could query $U_{A}$ or
$U_{A^{\bot}}$\ to find a generating set\ for $A$ or $A^{\bot}$---but that
would require an exponentially-long Grover search, since $\left\vert
A\right\vert =\left\vert A^{\bot}\right\vert =2^{n/2}\ll2^{n}$. Alternatively,
the counterfeiter could measure $\left\vert A\right\rangle $\ in the standard
or Hadamard bases---but that would reveal just \textit{one} random element of
$A$\ or $A^{\bot}$. \ Neither ability seems useful for \textit{copying}
$\left\vert A\right\rangle $, let alone recovering a full classical
description of $A$.\footnote{Obviously, if the counterfeiter had
$\Omega\left(  n\right)  $\ copies of $\left\vert A\right\rangle $, then it
\textit{could} recover a generating set for $A$, by simply measuring each copy
independently\ in the standard basis. \ That is why, in our full quantum money
scheme, the counterfeiter will \textit{not} have $\Omega\left(  n\right)
$\ copies of $\left\vert A\right\rangle $. \ Instead, each banknote\ will
involve a completely different subspace $A_{s}\leq\mathbb{F}_{2}^{n}$
(parameterized by its unique serial number $s$), so that measuring one
banknote\ reveals nothing about the others.}

And indeed, using the inner-product adversary method plus some other tools, we
will prove the following tight lower bound (Theorem \ref{averagecase}): even
if given a single copy of $\left\vert A\right\rangle $, as well as oracle
access to $U_{A}$\ and $U_{A^{\bot}}$, a counterfeiter still needs
$\Omega\left(  \epsilon2^{n/4}\right)  $\ queries to prepare a state that
has\ fidelity $\epsilon$\ with $\left\vert A\right\rangle ^{\otimes2}%
$.\textit{ \ }This will imply that our mini-scheme has $1/\exp\left(
n\right)  $\textit{ }soundness error.

\subsection{Formal Specification\label{SPEC}}

We are not quite done, since we never explained how the bank provides access
to $U_{A}$ and $U_{A^{\bot}}$. \ Thus, in our \textquotedblleft
final\textquotedblright\ mini-scheme $\mathcal{M}=\left(  \mathsf{Bank}%
_{\mathcal{M}},\mathsf{Ver}_{\mathcal{M}}\right)  $, the bank, verifier, and
counterfeiter will all have access to a \textit{single} classical oracle $U$,
which consists of four components:

\begin{itemize}
\item A \textbf{banknote generator} $\mathcal{G}\left(  r\right)  $, which
takes as input a random string $r\in\left\{  0,1\right\}  ^{n}$, and outputs a
set of linearly independent generators $\left\langle A_{r}\right\rangle
=\left\{  x_{1},\ldots,x_{n/2}\right\}  $ for a subspace $A_{r}\leq
\mathbb{F}_{2}^{n}$, as well as a unique $3n$-bit \textit{serial number}
$s_{r}\in\left\{  0,1\right\}  ^{3n}$. \ The function $\mathcal{G}$\ is chosen
uniformly at random, subject to the constraint that the serial numbers are all
distinct.\footnote{Note that one can implement $\mathcal{G}$\ using an
ordinary random oracle. \ In that case, the requirement that the serial
numbers are distinct will be satisfied with probability\ $1-O\left(
2^{-n}\right)  $.}

\item A \textbf{serial number checker} $\mathcal{H}\left(  s\right)  $, which
outputs $1$ if $s=s_{r}$\ is a valid serial number for some $\left\langle
A_{r}\right\rangle $, and $0$ otherwise.

\item A \textbf{primal subspace tester} $\mathcal{T}_{\operatorname*{primal}}%
$, which takes an input of the form $\left\vert s\right\rangle \left\vert
x\right\rangle $, applies $U_{A_{r}}$\ to $\left\vert x\right\rangle $ if
$s=s_{r}$\ is a valid serial number for some $\left\langle A_{r}\right\rangle
$, and does nothing otherwise.

\item A \textbf{dual subspace tester} $\mathcal{T}_{\operatorname*{dual}}$,
identical to $\mathcal{T}_{\operatorname*{primal}}$\ except that it applies
$U_{A_{r}^{\bot}}$\ instead of $U_{A_{r}}$.
\end{itemize}

Then $\mathcal{M}=\left(  \mathsf{Bank}_{\mathcal{M}},\mathsf{Ver}%
_{\mathcal{M}}\right)  $\ is defined as follows:

\begin{itemize}
\item $\mathsf{Bank}_{\mathcal{M}}\left(  0^{n}\right)  $ chooses
$r\in\left\{  0,1\right\}  ^{n}$ uniformly at random. \ It then looks up
$\mathcal{G}\left(  r\right)  =\left(  s_{r},\left\langle A_{r}\right\rangle
\right)  $, and outputs the banknote $\left\vert \$_{r}\right\rangle
=\left\vert s_{r}\right\rangle \left\vert A_{r}\right\rangle $.

\item $\mathsf{Ver}_{\mathcal{M}}\left(
\hbox{\rm\rlap/c}%
\right)  $ first uses $\mathcal{H}$\ to check that $%
\hbox{\rm\rlap/c}%
$\ has the form $\left(  s,\rho\right)  $, where $s=s_{r}$ is a valid serial
number. \ If so, then it uses $\mathcal{T}_{\operatorname*{primal}}$\ and
$\mathcal{T}_{\operatorname*{dual}}$\ to apply $V_{A_{r}}=H_{2}^{\otimes
n}\mathbb{P}_{A_{r}^{\bot}}H_{2}^{\otimes n}\mathbb{P}_{A_{r}}$, and accepts
if and only if $V_{A_{r}}\left(  \rho\right)  $\ accepts.
\end{itemize}

\subsection{Analysis\label{ANALYSIS}}

We now analyze the mini-scheme\ defined in Sections \ref{HSM} and \ref{SPEC}.
\ For convenience, we assume for most of the proof that the subspace
$A\leq\mathbb{F}_{2}^{n}$\ is \textit{fixed}, and that the counterfeiter (who
does not know $A$) only has access to the oracles $U_{A}$\ and $U_{A^{\bot}}$.
\ Then, at the end, we will explain how to generalize the conclusions to the
\textquotedblleft final\textquotedblright\ mini-scheme $\mathcal{M}$.

It will be convenient to consider the subset $A^{\ast}\subset\left\{
0,1\right\}  ^{n+1}$, defined by%
\[
A^{\ast}:=\left(  0,A\right)  \cup(1,A^{\bot}).
\]
Let $S_{A^{\ast}}$\ be the subspace of $\mathbb{C}^{2^{n+1}}$\ that is spanned
by basis states $\left\vert x\right\rangle $\ such that $x\in A^{\ast}$.
\ Then we can think of the pair of oracles $\left(  U_{A},U_{A^{\bot}}\right)
$\ as being a \textit{single} oracle $U_{A^{\ast}}$,\ which satisfies
$U_{A^{\ast}}\left\vert \psi\right\rangle =-\left\vert \psi\right\rangle
$\ for all $\left\vert \psi\right\rangle \in S_{A^{\ast}}$,\ and $U_{A^{\ast}%
}\left\vert \eta\right\rangle =\left\vert \eta\right\rangle $\ for all
$\left\vert \eta\right\rangle \in S_{A^{\ast}}^{\bot}$ (where here $\bot
$\ means the orthogonal complement in $\mathbb{C}^{2^{n+1}}$, \textit{not} the
orthogonal complement in $\mathbb{F}_{2}^{n}$!).

Recall the definition of the verifier $V_{A}$:%
\[
V_{A}:=H_{2}^{\otimes n}\mathbb{P}_{A^{\bot}}H_{2}^{\otimes n}\mathbb{P}_{A},
\]
where $\mathbb{P}_{A}$\ and $\mathbb{P}_{A^{\bot}}$\ denote projective
measurements that accept a basis state $\left\vert x\right\rangle $\ if and
only if $x$ belongs to $A$\ or $A^{\bot}$\ respectively. \ The following lemma
shows that $V_{A}$ \textquotedblleft works,\textquotedblright\ and indeed that
it gives us a projective mini-scheme.

\begin{lemma}
\label{testworks}$V_{A}=\left\vert A\right\rangle \left\langle A\right\vert $
is simply a projector onto $\left\vert A\right\rangle $. \ So in particular,
$\Pr\left[  V_{A}\left(  \left\vert \psi\right\rangle \right)  \text{
accepts}\right]  =\left\vert \left\langle \psi|A\right\rangle \right\vert
^{2}$.
\end{lemma}

\begin{proof}
It suffices to show that $V_{A}\left\vert A\right\rangle =\left\vert
A\right\rangle $ and that $V_{A}\left\vert \psi\right\rangle =0$ for all
$\left\vert \psi\right\rangle $\ orthogonal to $\left\vert A\right\rangle $.
\ First,%
\begin{align*}
V_{A}\left\vert A\right\rangle  &  =H_{2}^{\otimes n}\mathbb{P}_{A^{\bot}%
}H_{2}^{\otimes n}\mathbb{P}_{A}\left\vert A\right\rangle \\
&  =H_{2}^{\otimes n}\mathbb{P}_{A^{\bot}}H_{2}^{\otimes n}\left\vert
A\right\rangle \\
&  =H_{2}^{\otimes n}\mathbb{P}_{A^{\bot}}|A^{\bot}\rangle\\
&  =H_{2}^{\otimes n}|A^{\bot}\rangle\\
&  =\left\vert A\right\rangle .
\end{align*}
Second, if $\left\langle \psi|A\right\rangle =0$ then we can write%
\[
\left\vert \psi\right\rangle =\sum_{x\in2^{n}}{c_{x}\left\vert x\right\rangle
}%
\]
where $\sum_{x\in A}c_{x}=0$. \ Then%
\begin{align*}
V_{A}\left\vert \psi\right\rangle  &  =H_{2}^{\otimes n}\mathbb{P}_{A^{\bot}%
}H_{2}^{\otimes n}\mathbb{P}_{A}\sum_{x\in2^{n}}{c_{x}\left\vert
x\right\rangle }\\
&  =H_{2}^{\otimes n}\mathbb{P}_{A^{\bot}}H_{2}^{\otimes n}\sum_{x\in A}%
{c_{x}\left\vert x\right\rangle }\\
&  =\frac{1}{\sqrt{2^{n}}}H_{2}^{\otimes n}\mathbb{P}_{A^{\bot}}\sum_{x\in
A}{c_{x}\sum_{y\bot x}\left\vert y\right\rangle }\\
&  =\frac{1}{\sqrt{2^{n}}}H_{2}^{\otimes n}\sum_{y\in A^{\perp}}{\left\vert
y\right\rangle \sum_{x\in A}c_{x}}\\
&  =0.
\end{align*}

\end{proof}

We now show that perfect counterfeiting requires exponentially many queries to
$U_{A^{\ast}}$.

\begin{theorem}
[Lower Bound for Perfect Counterfeiting]\label{nocopyc}Given one copy of
$\left\vert A\right\rangle $, as well as oracle access to $U_{A^{\ast}}$, a
counterfeiter needs $\Omega\left(  2^{n/4}\right)  $ queries to prepare
$\left\vert A\right\rangle ^{\otimes2}$ with certainty (for a worst-case
$\left\vert A\right\rangle $).
\end{theorem}

\begin{proof}
We will apply Theorem \ref{lbthm}. \ Let the set $\mathcal{O}$\ contain
$U_{A^{\ast}}$\ for every possible subspace $A\leq\mathbb{F}_{2}^{n}$\ with
$\dim A=n/2$. \ Also, put $\left(  U_{A^{\ast}},U_{B^{\ast}}\right)  \in
R$\ if and only if $\dim\left(  A\cap B\right)  =n/2-1$. \ Then given
$U_{A^{\ast}}\in\mathcal{O}$\ and $\left\vert \eta\right\rangle \in
S_{A^{\ast}}^{\bot}$, let%
\[
\left\vert \eta\right\rangle =\sum_{x\in\left\{  0,1\right\}  ^{n+1}\setminus
A^{\ast}}\alpha_{x}\left\vert x\right\rangle .
\]
We have%
\begin{align*}
\operatorname*{E}_{U_{B^{\ast}}~:~\left(  U_{A^{\ast}},U_{B^{\ast}}\right)
\in R}\left[  F\left(  \left\vert \eta\right\rangle ,S_{B^{\ast}}\right)
^{2}\right]   &  =\operatorname*{E}_{B~:~\dim\left(  B\right)  =n/2,\dim
\left(  A\cap B\right)  =n/2-1}\left[  \sum_{x\in B^{\ast}\setminus A^{\ast}%
}\left\vert \alpha_{x}\right\vert ^{2}\right] \\
&  \leq\max_{x\in\left\{  0,1\right\}  ^{n+1}\setminus A^{\ast}}\left(
\Pr_{B~:~\dim\left(  B\right)  =n/2,\dim\left(  A\cap B\right)  =n/2-1}\left[
x\in B^{\ast}\right]  \right) \\
&  =\max_{x\in\left\{  0,1\right\}  ^{n}\setminus A}\left(  \Pr_{B~:~\dim
\left(  B\right)  =n/2,\dim\left(  A\cap B\right)  =n/2-1}\left[  x\in
B\right]  \right) \\
&  =\frac{\left\vert B\setminus A\right\vert }{\left\vert \left\{
0,1\right\}  ^{n}\setminus A\right\vert }~~\text{(for }\dim\left(  B\right)
=n/2,~\dim\left(  A\cap B\right)  =n/2-1\text{)}\\
&  =\frac{2^{n/2-1}}{2^{n}-2^{n/2}}\\
&  \leq\frac{1}{2^{n/2}}.
\end{align*}
Here the first line uses the definition of fidelity, the second line uses the
easy direction of the minimax theorem, the third line uses the symmetry
between $A$\ and $A^{\bot}$, and the fourth line uses the symmetry among all
$2^{n}-2^{n/2}$\ strings\ $x\in\left\{  0,1\right\}  ^{n}\setminus A$. \ The
conclusion is that we can set $\varepsilon:=2^{-n/2}$.

Fix $\left(  U_{A^{\ast}},U_{B^{\ast}}\right)  \in R$. \ Then $\left\vert
\left\langle A|B\right\rangle \right\vert =1/2$. \ On the other hand, if the
counterfeiter succeeds, it must map $\left\vert A\right\rangle $\ to some
state $\left\vert f_{A}\right\rangle :=\left\vert A\right\rangle \left\vert
A\right\rangle \left\vert \operatorname*{garbage}_{A}\right\rangle $, and
$\left\vert B\right\rangle $\ to some state $\left\vert f_{B}\right\rangle
:=\left\vert B\right\rangle \left\vert B\right\rangle \left\vert
\operatorname*{garbage}_{B}\right\rangle $. \ Therefore $\left\vert
\left\langle f_{A}|f_{B}\right\rangle \right\vert \leq1/4$. \ So setting
$c=1/2$\ and $d=1/4$, Theorem \ref{lbthm}\ tells us that the counterfeiter
must make%
\[
\Omega\left(  \frac{c-d}{\sqrt{\varepsilon}}\right)  =\Omega\left(
2^{n/4}\right)
\]
queries to $U_{A^{\ast}}$.
\end{proof}

A simple modification to the proof of Theorem \ref{nocopyc}\ shows that even
to counterfeit money \textit{almost} perfectly, one still needs exponentially
many queries to $U_{A^{\ast}}$.

\begin{corollary}
[Lower Bound for Small-Error Counterfeiting]\label{smallerror}Given one copy
of $\left\vert A\right\rangle $, as well as oracle access to $U_{A^{\ast}}$, a
counterfeiter needs $\Omega\left(  2^{n/4}\right)  $ queries to prepare a
state $\rho$\ such that $\left\langle A\right\vert ^{\otimes2}\rho\left\vert
A\right\rangle ^{\otimes2}\geq0.9999$ (for a worst-case $\left\vert
A\right\rangle $).
\end{corollary}

\begin{proof}
Let $\left\vert \left\langle A|B\right\rangle \right\vert =c$, and let
$\epsilon=0.0001$. \ If the counterfeiter succeeds, it must map $\left\vert
A\right\rangle $\ to some state $\rho_{A}$, and $\left\vert B\right\rangle
$\ to some state $\rho_{B}$, such that $\left\langle A\right\vert ^{\otimes
2}\rho_{A}\left\vert A\right\rangle ^{\otimes2}$\ and $\left\langle
B\right\vert ^{\otimes2}\rho_{B}\left\vert B\right\rangle ^{\otimes2}$\ are
both at least $1-\epsilon$. \ So letting $\left\vert f_{A}\right\rangle $\ and
$\left\vert f_{B}\right\rangle $\ be purifications of $\rho_{A}$\ and
$\rho_{B}$\ respectively, we have%
\begin{align*}
\left\vert \left\langle f_{A}|f_{B}\right\rangle \right\vert  &  \leq F\left(
\rho_{A},\rho_{B}\right)  \\
&  \leq\left\vert \left\langle A\right\vert ^{\otimes2}\left\vert
B\right\rangle ^{\otimes2}\right\vert +2\epsilon^{1/4}\\
&  =c^{2}+2\epsilon^{1/4}%
\end{align*}
where the second line follows from Lemma \ref{triangle}. \ So setting
$d:=c^{2}+2\epsilon^{1/4}$, Theorem \ref{lbthm}\ tells us that the
counterfeiter must make%
\[
\Omega\left(  \frac{c-c^{2}-2\epsilon^{1/4}}{\sqrt{2^{-n/2}}}\right)
\]
queries to $U_{A^{\ast}}$. \ Fixing $c:=1/2$, the above is $\Omega\left(
2^{n/4}\right)  $.
\end{proof}

Since the verifier $V_{A}$\ is projective, we can now combine Corollary
\ref{smallerror} with Theorem \ref{miniamp} to obtain the following
\textquotedblleft amplified\textquotedblright\ lower bound.

\begin{corollary}
[Lower Bound for High-Error Counterfeiting]\label{largeerror}Let
$1/\varepsilon=o\left(  2^{n/2}\right)  $. \ Given one copy of $\left\vert
A\right\rangle $, as well as oracle access to $U_{A^{\ast}}$, a counterfeiter
needs $\Omega\left(  \sqrt{\varepsilon}2^{n/4}\right)  $ queries to prepare a
state $\rho$\ such that $\left\langle A\right\vert ^{\otimes2}\rho\left\vert
A\right\rangle ^{\otimes2}\geq\varepsilon$ (for a worst-case $\left\vert
A\right\rangle $).
\end{corollary}

\begin{proof}
Suppose we have a counterfeiter $C$ that makes $o\left(  \sqrt{\varepsilon
}2^{n/4}\right)  $\ queries to $U_{A^{\ast}}$, and prepares a state\ $\sigma
$\ such that $\left\langle A\right\vert ^{\otimes2}\sigma\left\vert
A\right\rangle ^{\otimes2}\geq\varepsilon$. \ Let $\delta:=0.00001$. \ Then by
Theorem \ref{miniamp}, there exists an amplified counterfeiter $C^{\prime}%
$\ that makes%
\[
O\left(  \frac{\log1/\delta}{\sqrt{\varepsilon}\left(  \sqrt{\varepsilon
}+\delta^{2}\right)  }\right)  =O\left(  \frac{1}{\sqrt{\varepsilon}}\right)
\]
calls to $C$ and $V_{A}$, and that prepares a state $\rho$\ such that
$\left\langle A\right\vert ^{\otimes2}\rho\left\vert A\right\rangle
^{\otimes2}\geq1-\delta$. \ Now, counting the $o\left(  \sqrt{\varepsilon
}2^{n/4}\right)  $\ queries from each $C$ invocation and $O\left(  1\right)
$\ queries from each $V_{A}$\ invocation, the total number of queries that
$C^{\prime}$\ makes to $U_{A^{\ast}}$\ is%
\[
\left[  o\left(  \sqrt{\varepsilon}2^{n/4}\right)  +O\left(  1\right)
\right]  \cdot O\left(  \frac{1}{\sqrt{\varepsilon}}\right)  =o\left(
2^{n/4}\right)  .
\]
But this contradicts Corollary \ref{smallerror}.
\end{proof}

So far, we have only made statements about the \textit{worst} case for a
would-be counterfeiter. \ But such guarantees are clearly not enough: it could
be that \textit{most} money states $\left\vert A\right\rangle $ are easy to
duplicate, without contradicting any of the results we have seen so far.

We will show that the problem faced by a counterfeiter is \textit{random
self-reducible}: if a counterfeiter could duplicate a uniformly-random money
state $\left\vert A\right\rangle $, then it could duplicate \textit{any}
$\left\vert A\right\rangle $. \ Thus the bank can ensure security by creating
uniformly-random money states.

In what follows, let $\mathcal{S}$\ be the set of all subspaces $A\leq
\mathbb{F}_{2}^{n}$\ such that $\dim A=n/2$. \ Also, let $V_{A}^{\otimes
2}=\left(  \left\vert A\right\rangle \left\langle A\right\vert \right)
^{\otimes2}$\ be the projector onto $\left\vert A\right\rangle ^{\otimes2}$.

\begin{theorem}
[Lower Bound for Average-Case Counterfeiting]\label{averagecase}Let
$A\leq\mathbb{F}_{2}^{n}$\ be a uniformly-random element of $\mathcal{S}$.
\ Then given one copy of $\left\vert A\right\rangle $, as well as oracle
access to $U_{A^{\ast}}$, a counterfeiter $C$\ needs $\Omega\left(
\sqrt{\varepsilon}2^{n/4}\right)  $ queries to prepare a $2n$-qubit state
$\rho$\ that $V_{A}^{\otimes2}$\ accepts with probability at least
$\varepsilon$, for all $1/\varepsilon=o\left(  2^{n/2}\right)  $. \ Here the
probability is taken over the choice of $A\in\mathcal{S}$, as well as the
behavior of $C$\ and $V_{A}^{\otimes2}$.
\end{theorem}

\begin{proof}
Suppose we had a counterfeiter $C$\ that violated the above. \ Using $C$ as a
black box, we will show how to construct a new counterfeiter $C^{\prime}%
$\ that violates Corollary \ref{largeerror}.

Given a (deterministically-chosen) money state $\left\vert A\right\rangle $
and oracle access to $U_{A^{\ast}}$, first choose an invertible linear map
$f:\mathbb{F}_{2}^{n}\rightarrow\mathbb{F}_{2}^{n}$ uniformly at random.
\ Then $f\left(  A\right)  $, the image of $A$ under $f$,\ is a
uniformly-random element of $\mathcal{S}$. \ Furthermore, the state
$\left\vert A\right\rangle $ can be transformed into $\left\vert f\left(
A\right)  \right\rangle $ straightforwardly,\ the oracle $U_{f\left(
A\right)  }$ can be simulated by composing $f$ with $U_{A}$, and the oracle
$U_{f\left(  A\right)  ^{\bot}}$\ can likewise be simulated by composing
$f^{-T}$ with $U_{A}$ (where $f^{-T}$\ denotes the inverse transpose of $f$).
\ So by using the counterfeiter $C$ for uniformly-random states, we can
produce a state $\rho_{f}$\ that $V_{f\left(  A\right)  }^{\otimes2}$\ accepts
with probability at least $\varepsilon$. \ By applying $f^{-1}$ to both
registers of $\rho_{f}$, we can then obtain a state $\rho$\ that
$V_{A}^{\otimes2}$\ accepts with probability at least $\varepsilon$, thereby
contradicting Corollary \ref{largeerror}.
\end{proof}

We are now ready to prove security for the \textquotedblleft
final\textquotedblright\ mini-scheme $\mathcal{M}$\ defined in Section
\ref{SPEC}.

\begin{theorem}
[Security of Mini-Scheme]\label{moneygood}The mini-scheme $\mathcal{M}=\left(
\mathsf{Bank}_{\mathcal{M}},\mathsf{Ver}_{\mathcal{M}}\right)  $, which is
defined relative to the classical oracle $U$, has perfect completeness and
$1/\exp\left(  n\right)  $ soundness error.
\end{theorem}

\begin{proof}
That $\mathcal{M}$\ has perfect completeness follows from its definition and
from Lemma \ref{testworks}. \ That $\mathcal{M}$\ has $1/\exp\left(  n\right)
$ soundness error \textit{essentially} follows from Theorem \ref{averagecase}.
\ We only need to show that, given a banknote of the form $\left\vert
\$_{r}\right\rangle =\left\vert s_{r}\right\rangle \left\vert A_{r}%
\right\rangle $, a polynomial-time counterfeiter $C$\ can gain no additional
advantage by querying the \textquotedblleft full\textquotedblright\ oracles
$\mathcal{G},\mathcal{H},\mathcal{T}_{\operatorname*{primal}},\mathcal{T}%
_{\operatorname*{dual}}$, beyond what it gains from querying $U_{A_{r}^{\ast}%
}=\left(  U_{A_{r}},U_{A_{r}^{\bot}}\right)  $.

Let $r\in\left\{  0,1\right\}  ^{n}$ be the random string chosen by the bank,
so that $\mathcal{G}\left(  r\right)  =\left(  s_{r},\left\langle
A_{r}\right\rangle \right)  $. \ Then observe that, even conditioned on
$s_{r}$\ and $A_{r}$, as well as complete descriptions of $\mathcal{T}%
_{\operatorname*{primal}},\mathcal{T}_{\operatorname*{dual}}$, and
$\mathcal{H}$, the string $r$ remains uniformly random. \ Nor can querying
$\mathcal{G}\left(  r^{\prime}\right)  $\ for $r^{\prime}\neq r$\ reveal any
information about $r$, since the values of $\mathcal{G}$\ are generated
independently. \ So suppose we \textit{modify} $\mathcal{G}$ by setting
$\mathcal{G}\left(  r\right)  :=\left(  s^{\prime},\left\langle A^{\prime
}\right\rangle \right)  $, for some new $3n$-bit\ serial number $s^{\prime}%
$\ and list of generators $\left\langle A^{\prime}\right\rangle $\ chosen
uniformly at random.\ \ Then the BBBV hybrid argument \cite{bbbv}\ tells us
that, in expectation over $r$, this can alter the final state output by the
counterfeiter $C\left(  \left\vert \$_{r}\right\rangle \right)  $\ by at most
$\operatorname*{poly}\left(  n\right)  /2^{n/2}$ in trace distance. \ So in
particular, if $C$\ succeeded with non-negligible probability before, then $C$
must \textit{still} succeed with non-negligible probability after we set
$\mathcal{G}\left(  r\right)  :=\left(  s^{\prime},\left\langle A^{\prime
}\right\rangle \right)  $.

However, once we make this modification, an adversary trying to counterfeit
$\left\vert A\right\rangle $ given $U_{A}$ and $U_{A^{\perp}}$ can easily
\textquotedblleft mock up\textquotedblright\ a serial number $s$, as well as
the oracles$\ \mathcal{G},\mathcal{H},\mathcal{T}_{\operatorname*{primal}}$
and $\mathcal{T}_{\operatorname*{dual}}$, for itself. \ For $s$, $\mathcal{G}%
$, and $\mathcal{H}$ are now drawn from a distribution completely independent
of $A$. \ The oracles $\mathcal{T}_{\operatorname*{primal}}$ and
$\mathcal{T}_{\operatorname*{dual}}$ are likewise independent of $A$,
\textit{except} that $\mathcal{T}_{\operatorname*{primal}}\left\vert
s\right\rangle \left\vert v\right\rangle =\left\vert s\right\rangle
U_{A}\left\vert v\right\rangle $ and $\mathcal{T}_{\operatorname*{dual}%
}\left\vert s\right\rangle \left\vert v\right\rangle =\left\vert
s\right\rangle U_{A^{\perp}}\left\vert v\right\rangle $---behaviors that an
adversary can easily simulate using $U_{A}$ and $U_{A^{\perp}}$, together with
its knowledge of $s$. \ Just like in Corollary \ref{minitofull}, since our
security guarantees are query complexity bounds, we do not care about the
\textit{computational} complexity of creating the mock-ups.

By using the mock-ups, one can convert any successful attack on $\mathcal{M}$
into successful counterfeiting of $\left\vert A\right\rangle $, given oracle
access to $U_{A}$ and $U_{A^{\perp}}$ only. \ But the latter contradicts
Theorem \ref{averagecase}.
\end{proof}

Finally, using Theorem \ref{moneygood} together with Corollary
\ref{minitofull}, we can obtain a secure public-key quantum money scheme,
relative to a classical oracle.

\begin{theorem}
[Security of Hidden Subspace Money]\label{moneyverygood}By combining the
mini-scheme $\mathcal{M}$\ with a digital signature scheme, it is possible to
construct a public-key quantum money scheme $\mathcal{S}=\left(
\mathsf{KeyGen}_{\mathcal{S}},\mathsf{Bank}_{\mathcal{S}},\mathsf{Ver}%
_{\mathcal{S}}\right)  $, defined relative to some classical oracle
$U^{\prime}$, which has perfect completeness and $1/\exp\left(  n\right)  $
soundness error.
\end{theorem}

\section{Explicit Quantum Money Scheme\label{EXPLICIT}}

We have shown how to construct a provably-secure public-key quantum money
scheme, when an appropriate classical oracle is available. \ In this section,
we propose a way to obtain the same functionality without an oracle. \ The key
challenge is this:

\begin{quotation}
\noindent\textit{Given a subspace }$A\leq\mathbb{F}_{2}^{n}$\textit{, how can
a bank distribute an \textquotedblleft obfuscated program\textquotedblright%
\ }$P_{A}$\textit{, which legitimate buyers and sellers can use to decide
membership in both }$A$\textit{ and }$A^{\bot}$\textit{, but which does not
reveal anything else about }$A$\textit{ that might facilitate counterfeiting?}
\end{quotation}

Note that, aside from the detail that we need security against quantum
adversaries, the above challenge is purely \textquotedblleft
classical\textquotedblright;\ it and its variants seem interesting even apart
from our quantum money application.

We will suggest a candidate protocol to achieve the challenge, based on
\textit{multivariate polynomial cryptography}. \ Given a collection
$p_{1},\ldots,p_{m}:\mathbb{F}_{2}^{n}\rightarrow\mathbb{F}_{2}$\ of
multivariate polynomials over $\mathbb{F}_{2}$, it is generally hard to find a
point $v\in\mathbb{F}_{2}^{n}$ on which all of the $p_{i}$'s vanish. \ On the
other hand, it is easy to check whether a \textit{particular} point $v$ has
that property. \ To \textquotedblleft hide\textquotedblright\ a subspace $A$,
we will provide uniformly-random low-degree polynomials $p_{1},\ldots,p_{m}%
$\ that vanish on each point of $A$. \ This information is sufficient to
decide membership in $A$. \ On the other hand, there is no known efficient
algorithm to \textit{find} $A$ given the polynomials, and current techniques
seem unlikely to yield even a quantum algorithm.

We can also introduce a constant fraction of \textit{noise} into our scheme
without interfering with its completeness. \ In other words, if only $\left(
1-\epsilon\right)  m$ of the polynomials $p_{1},\ldots,p_{m}$\ are chosen to
vanish on $A$, and the remaining $\epsilon m$ are random, then counting the
number of $p_{i}$'s that vanish at a point $v$ still suffices to determine
whether $v\in A$. \ Although we know of no attack even against our noise-free
scheme, adding noise in this way might improve security.

Crucially, we will state a \textquotedblleft classical\textquotedblright%
\ conjecture about the security of multivariate polynomial cryptography, and
show that the conjecture \textit{implies} the security of our explicit money
scheme. \ For the benefit of cryptographers, let us now state an
\textquotedblleft abstract\textquotedblright\ version of our conjecture, which
implies what we need, and which might hold even if our concrete conjecture
about multivariate polynomials fails.

\begin{conjecture}
[Subspace-Hiding Conjecture, Sufficient for Quantum Money]\label{abstractconj}%
There exists a polynomial-time algorithm that takes as input a description of
a uniformly-random subspace $A\leq\mathbb{F}_{2}^{n}$ with $\dim\left(
A\right)  =n/2$, and that outputs circuits $C_{A}$ and $C_{A^{\bot}}$, such
that the following holds.

\begin{enumerate}
\item[(i)] $C_{A}\left(  v\right)  $ decides whether $v\in A$,\ and
$C_{A^{\bot}}\left(  v\right)  $\ decides whether $v\in A^{\bot}$, for all
$v\in\mathbb{F}_{2}^{n}$.

\item[(ii)] Given descriptions of $C_{A}$ and $C_{A^{\bot}}$, no
polynomial-time quantum algorithm can find a generating set for $A$ with
success probability $\Omega\left(  2^{-n/2}\right)  $.
\end{enumerate}
\end{conjecture}

Later, Conjecture \ref{subspacedpt} will specialize Conjecture
\ref{abstractconj}\ to the setting of multivariate polynomials.

\subsection{Useful Facts About Polynomials\label{POLYFACTS}}

By viewing elements of $\mathbb{F}_{2}^{n}$ as $n$-tuples $\left(
x_{1},\ldots,x_{n}\right)  $, we can evaluate a polynomial $p\left(
x_{1},\ldots,x_{n}\right)  $ on points of $\mathbb{F}_{2}^{n}$.

Given a subspace $A\leq\mathbb{F}_{2}^{n}$ and a positive integer $d$, let
$\mathcal{I}_{d,A}$ be the set of degree-$d$ polynomials (not necessarily
homogeneous) that vanish on $A$. \ Since we are working over $\mathbb{F}_{2}$,
note that $x_{i}^{2}=x_{i}$, so it suffices to consider \textit{multilinear}
polynomials (in which no $x_{i}$\ is ever raised to a higher power than $1$).

Before presenting our scheme, we need to establish some basic properties of
polynomials over $\mathbb{F}_{2}^{n}$. \ First, we observe that the set of
polynomials does not depend on the choice of basis.

\begin{proposition}
\label{basis}Let $L$ be any invertible linear transformation on $\mathbb{F}%
_{2}^{n}$. \ Then the map $p\left(  v\right)  \mapsto p\left(  Lv\right)  $
defines a permutation on the set of degree-$d$ polynomials, which maps
$\mathcal{I}_{d,A}$ to $\mathcal{I}_{d,L^{-1}A}$.
\end{proposition}

Implementing our scheme will require sampling uniformly from $\mathcal{I}%
_{d,A}$, which the next lemma shows is possible.

\begin{lemma}
\label{sampling}It is possible to sample a uniformly-random element of
$\mathcal{I}_{d,A}$ in time $O(n^{d})$.
\end{lemma}

\begin{proof}
By Proposition~\ref{basis}, we can instead sample from the space of
polynomials which vanish on $\operatorname*{span}\left(  x_{1},\ldots
,x_{n/2}\right)  $, and then apply an appropriate change of basis to obtain a
sample from $\mathcal{I}_{d,A}$. \ So assume without loss of generality that
$A=\operatorname*{span}\left(  x_{1},\ldots,x_{n/2}\right)  $.

We claim that a polynomial $p$\ vanishes on $A$ if and only if every monomial
of $p$\ intersects $\left\{  x_{n/2+1},\ldots,x_{n}\right\}  $. \ This will
immediately give an $O\left(  n^{d}\right)  $-time sampling algorithm, because
we can consider each of the $O\left(  n^{d}\right)  $ degree-$d$ monomials in
turn, and include each one independently with probability $1/2$ if it
intersects $\left\{  x_{n/2+1},\ldots,x_{n}\right\}  $.

To prove the claim: first, if every monomial intersects $\left\{
x_{n/2+1},\ldots,x_{n}\right\}  $, then clearly $p$ vanishes on $A$.
\ Otherwise, let $m$ be a minimal monomial that does \textit{not} intersect
$\left\{  x_{n/2+1},\ldots,x_{n}\right\}  $. \ Consider the vector $v=\left(
v_{1},\ldots,v_{n}\right)  $ with $v_{i}=1$ if and only if $x_{i}\in m$.
\ Since $m$ does not intersect $\left\{  x_{n/2+1},\ldots,x_{n}\right\}  $,
clearly $v\in A$. \ Also, since $m$ is minimal, every other monomial must
evaluate to $0$\ on $v$. \ Thus $p\left(  v\right)  =m\left(  v\right)  =1$,
so $p$ is not identically zero on $A$.
\end{proof}

In addition to sampling polynomials that vanish on $A$, we would like to
guarantee that a sufficiently large system of such polynomials uniquely
determines the space $A$, so that such a system can be effectively used as a
membership oracle.

\begin{lemma}
\label{unique}Fix $A\leq\mathbb{F}_{2}^{n}$ and $\beta>1$, and choose $\beta
n$ polynomials $p_{1},\ldots,p_{\beta n}$\ uniformly and independently from
$\mathcal{I}_{d,A}$. \ Let $Z$\ be the set of $v\in\mathbb{F}_{2}^{n}$\ such
that $p_{i}\left(  v\right)  =0$\ for all $i\in\left[  \beta n\right]  $.
\ Then $A\subseteq Z$, and $\Pr\left[  Z=A\right]  =1-2^{-\Omega\left(
n\right)  }$.
\end{lemma}

\begin{proof}
$A\subseteq Z$ is clear. \ For the probabilistic part, fix a point $v\notin
A$.\ \ Then by the union bound, it suffices to show that $\Pr\left[  v\in
Z\right]  <c^{-n}$ for some $c>2$.

There must be some $w\in A^{\perp}$ such that $w\cdot v=1$. \ Then the map
$p\left(  v\right)  \mapsto p\left(  v\right)  +w\cdot v$ defines an
involution of $\mathcal{I}_{d,A}$,\ such that exactly one of $p\left(
v\right)  $ and $p\left(  v\right)  +w\cdot v$\ is zero. \ This means that
exactly half of the polynomials in $\mathcal{I}_{d,A}$ vanish at $v$. \ Hence%
\[
\Pr\left[  p_{1}\left(  v\right)  =\cdots=p_{\beta n}\left(  v\right)
=0\right]  =2^{-\beta n}%
\]
and we are done.
\end{proof}

As mentioned earlier, we would also like to allow sampling from \textit{noisy}
systems of equations, defined as follows: let $\mathcal{R}_{d,A,m,\epsilon}$
be the probability distribution over $m$-tuples $\left(  p_{1},\ldots
,p_{m}\right)  $ that sets exactly $\left(  1-\epsilon\right)  m$\ of the
polynomials $p_{i}$ (chosen uniformly at random) to be uniformly-random
samples from $\mathcal{I}_{d,A}$, and that sets the remaining $\epsilon m$\ of
the polynomials $p_{i}$ to be uniformly-random samples from $\mathcal{I}%
_{d,A^{\prime}}$, for a uniformly-random subspace $A^{\prime}\leq
\mathbb{F}_{2}^{n}$\ of dimension $\dim\left(  A\right)  $. \ (Note that a
\textit{different} $A^{\prime}$\ is chosen for every such $p_{i}$.) \ Then
using a Chernoff bound, it is not hard to show that, provided $m$ is large
enough compared to $n$, a sample from $\mathcal{R}_{d,A,m,\epsilon}$
\textit{also} uniquely defines the subspace $A$ with overwhelming probability.

\begin{lemma}
\label{noisyunique}Fix $A\leq\mathbb{F}_{2}^{n}$ and $\epsilon<1/2$, let
$\beta\geq\frac{3}{\left(  1-2\epsilon\right)  ^{2}}$, and choose polynomials
$p_{1},\ldots,p_{\beta n}$\ from $\mathcal{R}_{d,A,\beta n,\epsilon}$. \ Let
$w\left(  v\right)  :=\sum_{i=1}^{\beta n}p_{i}\left(  v\right)  $, and let
$Z$\ be the set of $v\in\mathbb{F}_{2}^{n}$\ such that $w\left(  v\right)
\leq\epsilon\beta n$. \ Then $A\subseteq Z$, and $\Pr\left[  Z=A\right]
=1-2^{-\Omega\left(  n\right)  }$.
\end{lemma}

\begin{proof}
Again, $A\subseteq Z$ is clear. \ For the probabilistic part, fix $v\notin A$.
\ Then by the union bound, it suffices to show that $\Pr\left[  v\in Z\right]
<\alpha^{-n}$ for some $\alpha<1/2$.

Observe that $v$\ is a zero of little more than half the polynomials
$p_{1},\ldots,p_{\beta n}$. \ If $p_{i}$\ was chosen to vanish on $A$, then
$\operatorname*{E}\left[  p_{i}\left(  v\right)  \right]  =1/2$, by the
argument of Lemma \ref{unique}. \ If $p_{i}$ was chosen to vanish on a
uniformly-random $A^{\prime}$, then%
\begin{align*}
\operatorname*{E}\left[  p_{i}\left(  v\right)  \right]   &  \geq\frac{1}%
{2}-\Pr\left[  v\in A^{\prime}\right] \\
&  =\frac{1}{2}-\frac{1}{2^{n/2}}.
\end{align*}
Hence%
\[
\operatorname*{E}\left[  p_{1}\left(  v\right)  +\cdots+p_{\beta n}\left(
v\right)  \right]  \geq\beta n\left(  \frac{1}{2}-\frac{1}{2^{n/2}}\right)  .
\]
Furthermore, the $p_{i}$'s are chosen independently, up to an irrelevant
ordering. \ Choose $\delta=1-2\epsilon$\ to satisfy $\frac{1}{2}\left(
1-\delta\right)  =\epsilon$. \ Then by a Chernoff bound,%
\begin{align*}
\Pr\left[  v\in Z\right]   &  =\Pr\left[  p_{1}\left(  v\right)
+\cdots+p_{\beta n}\left(  v\right)  \leq\epsilon\beta n\right] \\
&  \leq\exp\left(  -\frac{1}{2}\frac{\beta n}{2}\left(  1-\frac{2}{2^{n/2}%
}\right)  \delta^{2}\right) \\
&  \leq\exp\left(  -\frac{3n}{4}\left(  1-\frac{2}{2^{n/2}}\right)  \right) \\
&  <0.48^{n}%
\end{align*}
for large enough $n$, and we are done.
\end{proof}

\subsection{Explicit Hidden-Subspace Mini-Scheme\label{EXPLICITMS}}

In our explicit mini-scheme, the bank chooses a subspace $A$ randomly and
publishes sets of polynomials drawn from $\mathcal{R}_{d,A,\beta n,\epsilon}%
$\ and $\mathcal{R}_{d,A^{\perp},\beta n,\epsilon}$, along with the quantum
money state $\left\vert A\right\rangle $. \ By Lemma~\ref{noisyunique}, a user
can use these polynomials to test membership in $A$ and $A^{\perp}$, and can
therefore implement the oracle mini-scheme in Section~\ref{HSM}.

Formally, the mini-scheme $\mathcal{E}$\ is defined as follows. \ Parameters
$\epsilon\in\left[  0,1/2\right)  $, $\beta\geq\frac{3}{\left(  1-2\epsilon
\right)  ^{2}}$, and $d\geq4$ are fixed. \ The complexity of the verification
procedure will grow like $O\left(  \beta n^{d+1}\right)  $, but security might
also improve for larger $\epsilon$\ and $d$. \ Then:

\begin{itemize}
\item $\mathsf{Bank}\left(  0^{n}\right)  $ selects an $n/2$-dimensional
subspace $A\leq\mathbb{F}_{2}^{n}$ uniformly at random, say by selecting $n/2$
random linearly-independent generators. \ It then sets $s:=\left(
s_{A},s_{A^{\perp}}\right)  $, where $s_{A}$ and $s_{A^{\perp}}$\ are lists of
polynomials drawn from $\mathcal{R}_{d,A,\beta n,\epsilon}$\ and
$\mathcal{R}_{d,A^{\perp},\beta n,\epsilon}$\ respectively. \ It prepares the
money state $\left\vert A\right\rangle $ and outputs the banknote $\left\vert
\$_{s}\right\rangle :=\left\vert s\right\rangle \left\vert A\right\rangle $.

\item $\mathsf{Ver}\left(
\hbox{\rm\rlap/c}%
\right)  $ first checks that $%
\hbox{\rm\rlap/c}%
$ has the form $\left(  s_{A},s_{A^{\perp}},\rho\right)  $ where
$s_{A}=\left(  p_{1},\ldots,p_{\beta n}\right)  $ and $s_{A^{\perp}}=\left(
q_{1},\ldots,q_{\beta n}\right)  $ are lists of $\beta n$ polynomials over
$\mathbb{F}_{2}^{n}$. \ If not, it rejects. \ If so, then it defines $Z$\ and
$Z^{\perp}$\ to be the sets of points $v\in\mathbb{F}_{2}^{n}$\ such that
$\sum_{i=1}^{\beta n}p_{i}\left(  v\right)  \leq\epsilon\beta n$ and
$\sum_{i=1}^{\beta n}q_{i}\left(  v\right)  \leq\epsilon\beta n$%
\ respectively. \ (Recall that with overwhelming probability, $Z=A$\ and
$Z^{\perp}=A^{\perp}$. \ Also, while $\mathsf{Ver}$\ will not have explicit
listings of the exponentially-large sets $Z$\ and $Z^{\perp}$, all that
matters for us is that it can efficiently apply the projections $\mathbb{P}%
_{Z}$\ and $\mathbb{P}_{Z^{\perp}}$.) \ It then applies the operation
$V_{Z}:=H_{2}^{\otimes n}\mathbb{P}_{Z^{\perp}}H_{2}^{\otimes n}\mathbb{P}%
_{Z}$ to $\rho$, and accepts $%
\hbox{\rm\rlap/c}%
$ if and only if $V_{Z}\left(  \rho\right)  $ accepts.
\end{itemize}

\subsection{Analysis\label{ANALEX}}

We first observe that the mini-scheme $\mathcal{E}$\ has perfect completeness.

\begin{theorem}
\label{explicitcompleteness}$\mathcal{E}$ has perfect completeness.
\end{theorem}

\begin{proof}
This follows from Lemmas \ref{unique} and\ \ref{noisyunique}, and particularly
from the fact that $A\subseteq Z$\ and $A^{\perp}\subseteq Z^{\perp}$\ with
certainty. \ From this it follows that $V_{Z}:=H_{2}^{\otimes n}%
\mathbb{P}_{Z^{\perp}}H_{2}^{\otimes n}\mathbb{P}_{Z}$\ accepts the state
$\left\vert A\right\rangle $\ with probability $1$.
\end{proof}

Let us remark that, if we want the fraction $\epsilon$\ of \textquotedblleft
decoy\textquotedblright\ polynomials to be even greater than $1/2$, then we
can define a variant of our scheme that works for all $\epsilon<1$. $\ $In
this variant scheme, $\mathsf{Ver}$\ will guess that $v\in A$\ (i.e., put
$v\in Z$) if%
\[
p_{1}\left(  v\right)  +\cdots+p_{\beta n}\left(  v\right)  \leq\frac{\left(
1+\epsilon\right)  \beta n}{4},
\]
and will guess that $v\notin A$\ (i.e., put $v\notin Z$) otherwise. \ By
direct analogy with Lemma \ref{noisyunique}, one can prove using a Chernoff
bound that this rule will guarantee $\Pr\left[  Z=A\right]  =1-2^{-\Omega
\left(  n\right)  }$, and likewise $\Pr\left[  Z^{\perp}=A^{\perp}\right]
=1-2^{-\Omega\left(  n\right)  }$, provided we set $\beta\geq\frac{12}{\left(
1-\epsilon\right)  ^{2}}$. \ However, the disadvantage is that if
$\epsilon\geq\frac{1}{3}$, then we lose the property that $A\subseteq Z$\ and
$A^{\perp}\subseteq Z^{\perp}$\ with probability $1$, since $\epsilon\geq
\frac{1+\epsilon}{4}$. \ This means, in particular, that we lose perfect
completeness, and can only ensure a completeness error of $2^{-\Omega\left(
n\right)  }$.

We now wish to argue about $\mathcal{E}$'s soundness. \ Naturally, we can only
hope to prove soundness assuming some computational hardness conjecture.
\ What is nice, though, is that we can base $\mathcal{E}$'s soundness on a
conjecture that talks only about the hardness of a \textquotedblleft
classical\textquotedblright\ cryptographic problem (i.e., a problem with
classical inputs and outputs). \ Let us now state that conjecture, which is
simply the abstract Conjecture \ref{abstractconj} specialized to the setting
of multivariate polynomials.

\begin{conjecture}
[Direct Product for Finding Subspace Elements]\label{subspacedpt}Let
$\epsilon<1/2$ and $\beta:=\frac{3}{\left(  1-2\epsilon\right)  ^{2}}$.
\ Given samples from $\mathcal{R}_{d,A,\beta n,\epsilon}$\ and $\mathcal{R}%
_{d,A^{\perp},\beta n,\epsilon}$, no polynomial-time quantum algorithm can
find a complete list of generators for $A$ with success probability
$\Omega\left(  2^{-n/2}\right)  $.
\end{conjecture}

Note that it is easy to find \textit{one} nonzero element of $A$ with success
probability $2^{-n/2}$, by choosing $x\in\mathbb{F}_{2}^{n}$ randomly.
\ Conjecture \ref{subspacedpt}\ asserts both that it is impossible to do too
much better using $\mathcal{R}_{d,A,\beta n,\epsilon}$\ and $\mathcal{R}%
_{d,A^{\perp},\beta n,\epsilon}$, and that finding multiple elements of $A$ is
significantly harder than finding one element.

The security of mini-scheme $\mathcal{E}$ follows easily from Conjecture
\ref{subspacedpt}, \textit{despite} the fact that a would-be counterfeiter has
access to a valid quantum banknote, whereas Conjecture \ref{subspacedpt}
involves no such assumption.

\begin{theorem}
[Security Reduction for Explicit Mini-Scheme]\label{explicitsoundness}If
Conjecture~\ref{subspacedpt} holds, then $\mathcal{E}$ is secure.
\end{theorem}

\begin{proof}
Let $C_{\mathcal{E}}$\ be a counterfeiter against $\mathcal{E}$. \ Then we
need to show that, using $C_{\mathcal{E}}$, we can find a complete list of
generators for $A$\ with $\Omega\left(  2^{-n/2}\right)  $\ success probability.

Given $A\leq\mathbb{F}_{2}^{n}$\ with $\dim\left(  A\right)  =n/2$, let
$s:=\left(  s_{A},s_{A^{\perp}}\right)  $ where $s_{A}$\ and $s_{A^{\perp}}%
$\ are samples from $\mathcal{R}_{d,A,\beta n,\epsilon}$\ and $\mathcal{R}%
_{d,A^{\perp},\beta n,\epsilon}$\ respectively. \ Recall from Lemma
\ref{noisyunique}\ that $\Pr\left[  A=Z\right]  =1-2^{-\Omega\left(  n\right)
}$\ and $\Pr\left[  A^{\perp}=Z^{\perp}\right]  =1-2^{-\Omega\left(  n\right)
}$. \ Provided both of these events occur, we can use $s$ to decide membership
in $A$, and can therefore apply the projective measurement $\mathbb{P}_{A}$.
\ So let us prepare the uniform superposition over\ all $2^{n}$ elements of
$\mathbb{F}_{2}^{n}$, and then apply $\mathbb{P}_{A}$ to it. \ With
probability $2^{-n/2}$, this produces the state $\left\vert A\right\rangle $.

Once we have $s$ and $\left\vert A\right\rangle $, we can then form the
banknote $\left\vert \$\right\rangle :=\left\vert s\right\rangle \left\vert
A\right\rangle $, and provide this banknote to the counterfeiter
$C_{\mathcal{E}}$. \ By hypothesis, $C_{\mathcal{E}}$ outputs a
(possibly-entangled) state $\rho$\ on two registers, such that $\left\langle
A\right\vert ^{\otimes2}\rho\left\vert A\right\rangle ^{\otimes2}\geq\Delta
$\ for some $\Delta=\Omega\left(  1/\operatorname*{poly}\left(  n\right)
\right)  $. \ But now, because the mini-scheme $\mathcal{E}$\ is projective,
Theorem \ref{miniamp}\ applies, and we can \textit{amplify} $\rho$\ to
increase its fidelity with $\left\vert A\right\rangle ^{\otimes2}$. \ After
$O\left(  \frac{1}{\Delta^{2}}\log n\right)  $ calls to $C_{\mathcal{E}}$,
this gives us a state $\sigma$\ such that%
\[
\left\langle A\right\vert ^{\otimes2}\sigma\left\vert A\right\rangle
^{\otimes2}\geq1-\frac{1}{n^{2}}.
\]
More generally, by alternating counterfeiting steps and amplification steps,
we can produce as many registers as we like that each have large overlap with
$\left\vert A\right\rangle $. \ In particular, we can produce a state $\xi
$\ such that%
\[
\left\langle A\right\vert ^{\otimes n}\xi\left\vert A\right\rangle ^{\otimes
n}\geq1-o\left(  1\right)  .
\]
If we now run $\mathsf{Ver}$\ on each of the registers of $\xi$, the
probability that every invocation accepts is $1-o\left(  1\right)  $.
\ Furthermore, supposing that happens, the state we are left with is simply
$\left\vert A\right\rangle ^{\otimes n}$.

Finally, we measure each register of $\left\vert A\right\rangle ^{\otimes n}%
$\ in the standard basis. \ This gives us $n$\ elements $x_{1},\ldots,x_{n}\in
A$, which are independent and uniformly random. \ So by standard estimates,
the probability that $x_{1},\ldots,x_{n}$\ do \textit{not} contain a complete
generating set for $A$\ is $1/\exp\left(  n\right)  $.

Overall, the procedure above succeeded with probability $2^{-n/2}\left(
1-o\left(  1\right)  \right)  $, thereby giving us the desired contradiction
with Conjecture~\ref{subspacedpt}.
\end{proof}

Using the standard construction of quantum money schemes, we can now produce a
complete explicit money scheme, whose security follows from
Conjecture~\ref{subspacedpt}.

\begin{theorem}
[Security Reduction for Explicit Scheme]\label{explicitsecurity}Assuming
Conjecture~\ref{subspacedpt}, there exists a public-key quantum money scheme
with perfect completeness and soundness error $2^{-\Omega\left(  n\right)  }$.
\end{theorem}

\begin{proof}
We apply the standard construction of Theorem~\ref{compose} with the
mini-scheme $\mathcal{E}$, whose completeness and soundness follow from
Theorems~\ref{explicitcompleteness} and~\ref{explicitsoundness} respectively,
assuming Conjecture~\ref{subspacedpt}.
\end{proof}

\subsection{Justifying Our Hardness Assumption\label{JUSTIFY}}

Though our hardness assumption is new, it is closely related to standard
assumptions in \textit{multivariate polynomial cryptography}. \ Given a system
of multivariate quadratics over $\mathbb{F}_{2}$, finding a common zero is
known to be $\mathsf{NP}$-hard; moreover, it is strongly believed that the
problem remains hard even for \textit{random} systems of multivariate
polynomials, and cryptosystems based on this hardness assumption are
considered promising candidates for post-quantum cryptography \cite{dy:mpkc}.
\ Therefore, if Conjecture~\ref{subspacedpt} fails, it will almost certainly
be because some additional structure in this problem facilitates a new attack.

There are several ways in which Conjecture~\ref{subspacedpt} is stronger than
the assumption that solving random systems of multivariate polynomials is
hard. \ First, our systems have large, well-structured solution spaces
$A$\ and $A^{\perp}$. \ Systems with many solutions are not normally
considered in the literature, and while there seem to be no known attacks that
exploit this structure, the possibility is not ruled out. \ Second, we provide
two related systems, one with zeroes in $A$ and one with zeroes in $A^{\perp}%
$. \ Again, this is a very specific structural property which has not been
considered, and there might be unexpected attacks exploiting it. \ Third,
Conjecture~\ref{subspacedpt} asserts that no adversary can succeed with
probability $2^{-n/2}$, which seems significantly easier than succeeding with
non-negligible probability.

On the other hand, Conjecture~\ref{subspacedpt} is \textit{weaker} than
typical assumptions in multivariate polynomial cryptography in at least one
respect: a would-be counterfeiter needs to solve a system of polynomial
equations with a constant fraction of noise. \ Solving noisy systems of
\textit{linear} equations over $\mathbb{F}_{2}$ is called the \textit{learning
parity with noise} problem, and is generally believed to be hard even for
quantum computers \cite{regev}. \ If true, this suggests that Gaussian
elimination is fundamentally hard to adapt to the presence of noise. \ But
computing a Gr\"{o}bner basis is a strict generalization of Gaussian
elimination to higher degree, and involves a nearly identical process of
elimination. \ It therefore seems unlikely that these approaches can be
efficiently adapted to the setting with noise. \ The problem of solving
polynomials with noise has been studied recently, and the best-known
approaches involve performing an exponential time search to determine which
equations are noisy \cite{ac:psn}.

But if solving linear systems with noise is already hard, why do we even use
higher-degree polynomials in our scheme? \ The reason is that, alas, the
\textquotedblleft dual\textquotedblright\ structure of our money scheme
facilitates a simple attack in the case $d=1$.

\begin{claim}
For all $\epsilon<1/2$, there exists a $\beta$ such that one can recover
$A$\ efficiently given samples from $\mathcal{R}_{d,A,\beta n,\epsilon}$\ and
$\mathcal{R}_{d,A^{\perp},\beta n,\epsilon}$.\footnote{This claim also goes
through, with no essential changes, for the variant of our scheme discussed
earlier with $\epsilon\in\left[  1/2,1\right)  $ (i.e., the variant without
perfect completeness).}
\end{claim}

\begin{proof}
Let $p_{1},\ldots,p_{m}$\ and $q_{1},\ldots,q_{m}$\ be homogeneous linear
polynomials, of which a $1-\epsilon$\ fraction vanish on $A$ and $A^{\perp}$
respectively. \ Then the key observation is that each $p_{i}$\ vanishes on
$A$\ if and only if it has the form $p_{i}\left(  v\right)  =u_{i}\cdot
v$\ for some $u_{i}\in A^{\perp}$, while each $q_{i}$\ vanishes on $A^{\perp}%
$\ if and only if it has the form $q_{i}\left(  v\right)  =w_{i}\cdot v$\ for
some $w_{i}\in A$. \ But by Lemma \ref{noisyunique}, if $\beta>\frac
{3}{\left(  1-2\epsilon\right)  ^{2}}$, then for each $i\in\left[  m\right]
$, we can efficiently \textit{decide} whether $u_{i}\in A^{\perp}$ by counting
the number of $j$'s for which $q_{j}\left(  u_{i}\right)  =0$, and can
likewise decide whether $w_{i}\in A$\ by counting the number of $j$'s for
which $p_{j}\left(  w_{i}\right)  =0$. \ Thus we can learn $\Theta\left(
n\right)  $ random elements of $A$ or $A^{\perp}$, and thereby recover a basis
for $A$.
\end{proof}

There \textit{might} be a more sophisticated attack for higher degrees, but
this is suggested only weakly by the existence of an attack in the linear
case. \ Indeed, the relation between the complementary linear subspaces
$A$\ and $A^{\perp}$\ is precisely the sort of structure that should be
preserved by linear maps, but \textit{not} by higher-degree polynomials!

For degree-$2$ polynomials, it is possible to obtain a similar attack which
recovers $A$ from only a \textit{single }sample. \ This attack relies on the
observation that quadratics have an easily-computed canonical form
\cite{bffp:ip1s}, from which a basis for $A$ can be extracted in polynomial
time. \ The essential problem is that quadratic polynomials are very closely
related to bilinear forms, and that powerful methods from linear algebra can
therefore be applied to them.

Fortunately, the linear structure seems to be computationally obscured when
$d\geq3$. \ This phenomenon is related to the sharp discontinuity in the
difficulty of tensor problems with order $3$ and higher. \ More concretely,
the coefficients of a degree-$d$ polynomial can be viewed as the entries of an
order-$d$ tensor, and the existence of an attack in the degree $d=2$ case
corresponds to the possibility of efficient operations on order-$2$ tensors.
\ Basic operations on order-$3$ tensors are $\mathsf{NP}$-hard \cite{hl:tnp},
however, and this suggests that analogous attacks might not exist against
degree-$3$ polynomials.

This state of affairs is reflected in existing attacks on a standard
cryptographic assumption called \textit{polynomial isomorphism with one
secret}. \ Here we are given two polynomials $p,q$\ which are related by an
unknown linear change of coordinates $L$, and the task is to find such an $L$.
\ For degree-$2$ polynomials, this problem can be easily solved in polynomial
time \cite{bffp:ip1s}, but already for degree-$3$ polynomials the best known
attacks take exponential time \cite{cgp:ip1s,gms:ip1s,bffp:ip1s}. \ However,
if an attacker is given $n$ bits of partial information about the linear
transformation, then even in the $d=3$\ case, it becomes possible to find the
linear transformation that relates the polynomials \cite{bffp:ip1s}. \ This
does \textit{not} directly facilitate an attack on our assumption, but it
suggests that a similar attack \textit{might} be possible when $d=3$, since an
attacker is only required to succeed with $2^{-n/2}$ probability.
\ Fortunately, this attack seems to rely on the particular structure of degree
$2$ and $3$ polynomials. \ Of course it is possible that similar algorithms
may be discovered for higher-degree polynomials, but this would represent an
advance in algebraic cryptanalysis.

\section{Private-Key Quantum Money\label{QUERYSEC}}

Recall that a \textit{private-key} quantum money scheme is one where only the
bank itself is able to verify banknotes, using an $n$-bit key
$k=k_{\operatorname*{private}}=k_{\operatorname*{public}}$ that it keeps a
closely-guarded secret. \ Compensating for this disadvantage, private-key
schemes are known with much stronger security guarantees than seem possible
for public-key schemes.

In particular, as mentioned in Section \ref{HISTORY}, already forty years ago
Wiesner \cite{wiesner} described how to create private-key quantum money that
is \textit{information-theoretically secure}. \ In Wiesner's scheme, each
banknote consists of $n$ unentangled qubits together with a classical serial
number $s$. \ Wiesner's scheme also requires a giant database of serial
numbers maintained by the bank, or in our setting, access to a random oracle
$R$. \ But in followup work, BBBW \cite{bbbw}\ pointed out that we can replace
$R$ by any \textit{pseudorandom function family} $\left\{  f_{k}\right\}
_{k}$, to obtain a private-key quantum money scheme that is computationally
secure, \textit{unless} a polynomial-time algorithm can distinguish the
$f_{k}$'s\ from random functions.

Strangely, we are unaware of any rigorous proof of the security of Wiesner's
scheme until recently. \ However, answering a question by one of
us,\footnote{See
http://theoreticalphysics.stackexchange.com/questions/370/rigorous-security-proof-for-wiesners-quantum-money}
Molina, Vidick and Watrous \cite{molina}\ have now supplied the key ingredient
for a security proof. \ Specifically they show that, if a counterfeiter tries
to copy an $n$-qubit banknote $\left\vert \$\right\rangle $\ in Wiesner's
scheme, then the output can have squared fidelity at most $\left(  3/4\right)
^{n}$ with $\left\vert \$\right\rangle ^{\otimes2}$. \ (They also show that
this is tight: there \textit{exists} a non-obvious counterfeiting strategy
that succeeds with $\left(  3/4\right)  ^{n}$\ probability.)

To complete the security proof, one needs to show that, even given
$q$\ banknotes $\left\vert \$_{1}\right\rangle ,\ldots,\left\vert
\$_{q}\right\rangle $, a counterfeiter cannot prepare an additional banknote
with non-negligible probability (even with a new serial number). \ In a
forthcoming paper \cite{aar:private}, we will show how to adapt the methods of
Section \ref{DEF} to prove that claim. \ Briefly, one can first define a
notion of \textit{private-key mini-schemes}, in close analogy to public-key
mini-schemes. \ The work of Molina et al.\ \cite{molina}\ then directly
implies the security of what we call the \textquotedblleft Wiesner
mini-scheme.\textquotedblright\ \ Next, one can give a general reduction,
showing how to construct a full-blown private-key quantum money scheme
$\mathcal{S}$ starting from

\begin{enumerate}
\item[(1)] any private-key \textit{mini}-scheme $\mathcal{M}$, and

\item[(2)] any random or pseudorandom function family $R$.
\end{enumerate}

Though the details turn out to be more complicated in the private-key case,
the proof of correctness for this reduction is conceptually similar to the
proof of Theorem \ref{compose}. \ Namely, one shows that any counterfeiter
would yield \textit{either} a break of the underlying mini-scheme
$\mathcal{M}$, or \textit{else} a way to distinguish $R$ from a random
function. \ Notice that the analysis is completely unified: if $R$ is a
\textquotedblleft true\textquotedblright\ random oracle, then we get
information-theoretic security (as in Wiesner's scheme), while if $R$ is
pseudorandom, then we get computational security (as in the BBBW scheme).

Unfortunately, as pointed out by Lutomirski \cite{lutomirski:attack} and
Aaronson \cite{aar:qcopy}, the Wiesner and BBBW schemes both have a serious
security hole. \ Namely, suppose a counterfeiter $C$\ can repeatedly submit
alleged banknotes\ to a \textquotedblleft na\"{\i}ve and trusting
bank\textquotedblright\ for verification. \ Given a quantum state $\sigma$,
such a bank not only tells $C$ whether the verification procedure accepted or
rejected, but also, in either case, \textit{gives the post-measurement state
}$\widetilde{\sigma}$\textit{ back to }$C$. \ Then starting from a single
valid banknote $\left\vert \$\right\rangle $, we claim that $C$\ can recover a
complete classical description of $\left\vert \$\right\rangle $, using
$O\left(  n\log n\right)  $\ queries to the bank. \ Once it has such a
description, $C$ can of course prepare as many copies of $\left\vert
\$\right\rangle $\ as it likes.

The attack is simple: let $\left\vert \$\right\rangle =\left\vert \theta
_{1}\right\rangle \cdots\left\vert \theta_{n}\right\rangle $ (we omit the
classical serial number $s$, since it plays no role here). \ Then for each
$i\in\left[  n\right]  $, the counterfeiter tries \textquotedblleft swapping
out\textquotedblright\ the $i^{th}$\ qubit $\left\vert \theta_{i}\right\rangle
$\ and replacing it with $\left\vert b\right\rangle $, for each of the four
possibilities\ $\left\vert b\right\rangle \in\left\{  \left\vert
0\right\rangle ,\left\vert 1\right\rangle ,\left\vert +\right\rangle
,\left\vert -\right\rangle \right\}  $. \ It then uses $O\left(  \log
n\right)  $ queries to the bank, to estimate the probability that the state
$\left\vert \theta_{1}\right\rangle \cdots\left\vert \theta_{i-1}\right\rangle
\left\vert b\right\rangle \left\vert \theta_{i+1}\right\rangle \cdots
\left\vert \theta_{n}\right\rangle $\ passes the verification test. \ By doing
so, $C$ can learn a correct value of $\left\vert \theta_{i}\right\rangle
$\ with success probability $1-o\left(  1/n\right)  $. \ The crucial point is
that none of these queries damage the\ qubits \textit{not} being investigated
($\left\vert \theta_{j}\right\rangle $\ for $j\neq i$), since the bank
measures those qubits\ in the correct bases. \ Therefore $C$\ can reuse the
same banknote for each query.

More generally, recall from Section \ref{SCHEMES}\ that we call a private-key
quantum money scheme\ \textit{query-secure}, if it remains secure even
assuming the counterfeiter $C$\ can make adaptive queries\ to $\mathsf{Ver}%
\left(  k,\cdot\right)  $. \ Then we saw that the Wiesner and BBBW schemes are
\textit{not} query-secure. \ Recently, Farhi et al.\ \cite{farhi:restore}
proved a much more general \textquotedblleft no-go\textquotedblright%
\ theorem---which says intuitively that, if we want query-secure quantum
money, then the banknotes \textit{must} hide information in the
\textquotedblleft global correlations\textquotedblright\ between large numbers
of qubits.

\begin{theorem}
[Adaptive Attack on Wiesner-Like Schemes \cite{farhi:restore}]%
\label{queryinsecure}No quantum money scheme can be query-secure, if

\begin{enumerate}
\item[(i)] the banknotes have the form $\left\vert \$_{s}\right\rangle
=\left\vert s\right\rangle \left\vert \psi_{s}\right\rangle $,

\item[(ii)] verification of $\left(  s,\rho\right)  $ consists of projecting
$\rho$ onto $\left\vert \psi_{s}\right\rangle \left\langle \psi_{s}\right\vert
$, and

\item[(iii)] $\left\vert \psi_{s}\right\rangle $ can be reconstructed uniquely
from the statistics of $T=\operatorname*{poly}\left(  n\right)  $
efficiently-implementable measurements $M_{1},\ldots,M_{T}$, each of which has
at most $\operatorname*{poly}\left(  n\right)  $\ possible outcomes.
\end{enumerate}
\end{theorem}

On the positive side, any \textit{public-key} quantum money scheme---for
example, our multivariate polynomial scheme from Section \ref{EXPLICIT}%
---immediately yields a query-secure scheme with the same security guarantee.
\ This is because a counterfeiter who knows the code of $\mathsf{Ver}$ can
easily simulate oracle access to $\mathsf{Ver}$. \ But can we do any better
than that, and construct a query-secure money scheme whose security is
\textit{unconditional} (as in Wiesner's scheme), or else based on a
pseudorandom function\ (as in the BBBW scheme)?

In the forthcoming paper \cite{aar:private}, we will answer this question in
the affirmative, by \textit{directly}\ adapting the hidden subspace scheme
from Section \ref{COR} (i.e., the scheme based on a classical oracle). \ Since
the idea is an extremely simple one, let us sketch it here. \

\begin{theorem}
[Query-Secure Variant of Wiesner's Scheme]\label{wiesnerfix}Relative to a
random oracle $R$,\footnote{Or alternatively, assuming the bank has access to
a giant random number table, as in Wiesner's original setup \cite{wiesner}.}
there exists a private-key quantum money scheme, with perfect completeness and
$2^{-\Omega\left(  n\right)  }$ soundness error, that is
information-theoretically query-secure. \ One can also replace the random
oracle $R$\ by a pseudorandom function family $\left\{  f_{k}\right\}  _{k}$,
to obtain a private-key quantum money scheme, with no oracle, that is
query-secure assuming that the $f_{k}$'s cannot be distinguished from random
in quantum polynomial time.
\end{theorem}

\begin{proof}
[Proof Sketch]For each key $k$\ and a serial number $s$, we will think of the
random oracle $R$ as encoding a classical description $R\left(  k,s\right)
$\ of a subspace $A_{k,s}\leq\mathbb{F}_{2}^{n}$, which is uniformly random
subject to $\dim\left(  A_{k,s}\right)  =n/2$. \ Let $\left\vert
A_{k,s}\right\rangle $\ be a uniform superposition over $A_{k,s}$. \ Then the
private-key money scheme $\mathcal{S}=\left(  \mathsf{KeyGen},\mathsf{Bank}%
,\mathsf{Ver}\right)  $ is defined as follows:

\begin{itemize}
\item $\mathsf{KeyGen}\left(  0^{n}\right)  $ generates an $n$-bit key $k$
uniformly at random.

\item $\mathsf{Bank}\left(  k\right)  $ outputs a banknote $\left\vert
\$_{s}\right\rangle :=\left\vert s\right\rangle \left\vert A_{k,s}%
\right\rangle $, for a random serial number $s\in\left\{  0,1\right\}  ^{n}$.

\item $\mathsf{Ver}\left(  k,\left(  s,\rho\right)  \right)  $ applies a
projective measurement that accepts $\rho$ with probability $\left\langle
A_{k,s}|\rho|A_{k,s}\right\rangle $.
\end{itemize}

Now, suppose it were possible to break $\mathcal{S}$\ (i.e., to counterfeit
$\left\vert A_{k,s}\right\rangle $), using $\operatorname*{poly}\left(
n\right)  $\ adaptive queries to $\mathsf{Ver}\left(  k,\cdot\right)  $.
\ Then we claim that it would \textit{also} be possible to break our
\textit{public-key} scheme\ from Section \ref{COR}, and thereby contradict the
unconditional security proof for the latter! \ The reason is simply that any
query to $\mathsf{Ver}$, of the form $\mathsf{Ver}\left(  k,\left(
s,\rho\right)  \right)  $, can easily be simulated using queries to
$U_{A_{k,s}}$\ and $U_{A_{k,s}^{\bot}}$, the membership oracles for $A_{k,s}%
$\ and $A_{k,s}^{\bot}$\ respectively that are available to a
counterfeiter\ against the public-key scheme.

Finally, suppose we replace $R\left(  k,s\right)  $ by a pseudorandom function
$f_{k}\left(  s\right)  $. \ Then just like with the original BBBW scheme
\cite{bbbw}, we can argue as follows. \ Since we already showed that
$\mathcal{S}$\ is information-theoretically secure when instantiated with a
\textquotedblleft true\textquotedblright\ random function, any break of
$\mathcal{S}$\ in the pseudorandom case would thereby distinguish the function
$f_{k}$\ from random.
\end{proof}

\section{Open Problems\label{OPEN}}

The \textquotedblleft obvious\textquotedblright\ problem is to better
understand the security of our explicit scheme based on polynomials. \ Are
there nontrivial attacks, for example using Gr\"{o}bner-basis algorithms?
\ Can we base the security of our scheme---or a related scheme---on some
cryptographic assumption that does \textit{not} involve exponentially-small
success probabilities? \ What happens as we change the field size or
polynomial degree? \ Does \textquotedblleft hiding\textquotedblright\ a
subspace $A\leq\mathbb{F}_{2}^{n}$\ in the way we suggest, as the set of
common zeroes of multivariate polynomials $p_{1},\ldots,p_{m}:\mathbb{F}%
_{2}^{n}\rightarrow\mathbb{F}_{2}$, have other cryptographic applications, for
example to \textit{program obfuscation} \cite{baraketal}?

Of course, there is also tremendous scope for inventing new schemes, which
might be based on different assumptions and have different strengths and weaknesses.

Let us move on to some general questions about public-key quantum money.
\ First, is there an unconditionally-secure public-key quantum money scheme
relative to a \textit{random} oracle $R$? \ (Recall that Wiesner's original
scheme \cite{wiesner}\ was unconditionally-secure and used only a random
oracle, but was private-key. \ Meanwhile, our scheme from Section \ref{COR} is
unconditionally-secure and public-key, but requires a non-random oracle.)
\ Second, is there a public-key quantum money scheme where the banknotes
consist of \textit{single, unentangled qubits}, as in Wiesner's scheme? \ Note
that the results of Farhi et al.\ \cite{farhi:restore}\ imply that, if such a
scheme exists, then it cannot be projective. \ Third, is there a general way
to amplify soundness error in quantum money schemes?\footnote{Theorem
\ref{miniamp} gives \textit{some} soundness amplification for projective
schemes: namely, from constant to $1/\operatorname*{poly}\left(  n\right)  $.
\ Here we are asking whether one can do anything better.} \ (We show how to
amplify \textit{completeness} error in Appendix \ref{COMPLETENESS}.)

\subsection{Quantum Copy-Protection and More\label{COPY}}

Quantum money is just \textit{one} novel cryptographic use for the No-Cloning
Theorem. \ Given essentially \textit{any} object of cryptographic interest,
one can ask whether quantum mechanics lets us make the object uncloneable.
\ Section \ref{MOTIVATION} already discussed one example---\textit{uncloneable
signatures}---but there are many others, such as commitments and
proofs.\footnote{Even within complexity theory, it would be interesting to
study the class $\mathsf{QMA}$\ (Quantum Merlin-Arthur) subject to the
constraint that witnesses must be hard to clone---or alternatively, that
witnesses must be \textit{easy} to clone!}

Along those lines, Aaronson \cite{aar:qcopy} proposed a task that, if
achievable, would arguably be an even more dramatic application of the
No-Cloning Theorem than quantum money:\ namely, \textit{quantum software
copy-protection}. \ He gave explicit schemes---which have not yet been
broken---for copy-protecting a restricted class of functions, namely the
\textit{point functions}. \ In these schemes, given a \textquotedblleft
password\textquotedblright\ $s\in\left\{  0,1\right\}  ^{n}$, a software
vendor can prepare a quantum state $\left\vert \psi_{s}\right\rangle $, which
allows its holder to \textit{recognize} $s$: in other words, to decide whether
$x=s$ given $x\in\left\{  0,1\right\}  ^{n}$\ as input. \ On the other hand,
given $\left\vert \psi_{s}\right\rangle $, it seems intractable not only to
\textit{find }$s$\textit{ }for oneself, but even to prepare a \textit{second}
quantum state with which $s$ can be recognized.

Admittedly, recognizing passwords is an extremely restricted functionality.
\ However, relative to a quantum oracle, Aaronson \cite{aar:qcopy}\ also
described a scheme to quantumly copy-protect \textit{arbitrary} programs, just
as well as if the software vendor were able to hand out uncloneable black
boxes.\footnote{As usual, full details have not yet appeared yet.} \ In the
spirit of this paper, we can now ask: is there likewise a way to quantumly
copy-protect arbitrary programs relative to a \textit{classical} oracle? \ We
conjecture that the answer is yes, and in fact we have plausible candidate
constructions, which are directly related to the hidden-subspace money scheme
of Section \ref{COR}. \ However, the security of those constructions seems to
hinge on the following conjecture.

\begin{conjecture}
[Direct Product for Finding Black-Box Subspace Elements]\label{bbsubspace}Let
$A$ be a uniformly-random subspace of $\mathbb{F}_{2}^{n}$\ satisfying
$\dim\left(  A\right)  =n/2$. \ Then given membership oracles for both
$A$\ and $A^{\bot}$, any quantum algorithm needs $2^{\Omega\left(  n\right)
}$ queries to find two distinct nonzero elements $x,y\in A$, with success
probability $\Omega\left(  2^{-n/2}\right)  $.
\end{conjecture}

Besides its applications for copy-protection, a proof of\ Conjecture
\ref{bbsubspace}\ would be an important piece of formal evidence for
Conjecture \ref{subspacedpt}, on which we based the security of our explicit
money scheme.

\section{Appendix: Reducing Completeness Error\label{COMPLETENESS}}

When we defined quantum money schemes and mini-schemes in Section \ref{DEF},
we allowed the verifier to reject a \textit{legitimate} money state with
probability up to $1/3$. \ But of course, a money scheme with\ completeness
error $\varepsilon=1/3$\ is not very useful in practice! \ So in this
appendix, we prove that the completeness error $\varepsilon$\ can be made
exponentially small in $n$, at the cost of only a modest increase in the
soundness error $\delta$ (i.e., the probability of successful counterfeiting).

\begin{theorem}
[Completeness Amplification for Mini-Schemes]\label{ampcomp}Let $\mathcal{M}%
=\left(  \mathsf{Bank},\mathsf{Ver}\right)  $\ be a quantum money mini-scheme
with completeness error $\varepsilon<1/2$\ and soundness error $\delta
<1-2\varepsilon$. \ Then for all polynomials $p$ and all $\delta^{\prime
}>\frac{\delta}{1-2\varepsilon}$, we can construct an amplified mini-scheme
$\mathcal{M}^{\prime}=\left(  \mathsf{Bank}^{\prime},\mathsf{Ver}^{\prime
}\right)  $\ with completeness error $1/2^{p\left(  n\right)  }$\ and
soundness error $\delta^{\prime}$.
\end{theorem}

\begin{proof}
Let $k=\operatorname*{poly}\left(  n\right)  $ and $\eta>0$\ be parameters to
be determined later. \ Our construction of $\mathcal{M}^{\prime}$\ is the
\textquotedblleft obvious\textquotedblright\ one based on repetition:

\begin{itemize}
\item $\mathsf{Bank}^{\prime}\left(  0^{n}\right)  $\ outputs a composite
banknote $\$^{\prime}:=\left(  s_{1}\ldots s_{k},\rho_{s_{1}}\ldots\rho
_{s_{k}}\right)  $, where $\left(  s_{1},\rho_{s_{1}}\right)  ,\ldots,\left(
s_{k},\rho_{s_{k}}\right)  $\ are banknotes output independently by
$\mathsf{Bank}\left(  0^{n}\right)  $.

\item $\mathsf{Ver}^{\prime}\left(
\hbox{\rm\rlap/c}%
\right)  $ runs $\mathsf{Ver}\left(
\hbox{\rm\rlap/c}%
_{1}\right)  ,\ldots,\mathsf{Ver}\left(
\hbox{\rm\rlap/c}%
_{k}\right)  $, where $%
\hbox{\rm\rlap/c}%
_{1},\ldots,%
\hbox{\rm\rlap/c}%
_{k}$\ are the $\left(  s,\rho_{s}\right)  $\ pairs in the alleged composite
banknote $%
\hbox{\rm\rlap/c}%
$, and accepts if and only if at least $\left(  1-\varepsilon-\eta\right)  k$
invocations accept.
\end{itemize}

Note that $\mathsf{Ver}_{2}^{\prime}$, the amplified double verifier, then
takes as input a state of the form%
\[
\left(  s_{1}\ldots s_{k},\sigma_{1}\ldots\sigma_{k},\xi_{1}\ldots\xi
_{k}\right)  ,
\]
and accepts if and only if $\mathsf{Ver}^{\prime}\left(  s_{1}\ldots
s_{k},\sigma_{1}\ldots\sigma_{k}\right)  $\ and\ $\mathsf{Ver}^{\prime}\left(
s_{1}\ldots s_{k},\xi_{1}\ldots\xi_{k}\right)  $\ both accept. \ By choosing
$k$\ sufficiently large and applying a Chernoff bound, it is clear that we can
make the completeness error $1/2^{p\left(  n\right)  }$\ for any polynomial
$p$.

Meanwhile, suppose $\mathcal{M}^{\prime}$\ has soundness error $\delta
^{\prime}$:\ in other words, there exists a counterfeiter $C^{\prime}$\ such
that $\mathsf{Ver}_{2}^{\prime}\left(  s_{1}\ldots s_{k},C^{\prime}\left(
\$^{\prime}\right)  \right)  $\ accepts with probability $\delta^{\prime}$,
given a valid composite banknote $\$^{\prime}$. \ Then to prove the theorem,
it suffices to construct a counterfeiter $C$\ for the original mini-scheme
$\mathcal{M}$, such that $\mathsf{Ver}_{2}\left(  s,C\left(  \$\right)
\right)  $\ accepts with probability $\delta\geq\left(  1-2\varepsilon
-\eta\right)  \delta^{\prime}$, given a valid banknote $\$=\left(  s,\rho
_{s}\right)  $.

This $C$ works as follows:

\begin{enumerate}
\item[(1)] By calling $\mathsf{Bank}^{\prime}\left(  0^{n}\right)  $, generate
a new composite banknote $\$^{\prime}=\left(  \$_{1},\ldots,\$_{k}\right)  $.

\item[(2)] Let $\$_{\operatorname*{new}}^{\prime}$\ be the result of starting
with $\$^{\prime}$, then swapping out $\$_{i}$\ for the banknote $\$$\ to be
copied, for some $i\in\left[  k\right]  $\ chosen uniformly at random.

\item[(3)] Let $\left(  s_{1}\ldots s_{k},\sigma_{1}\ldots\sigma_{k},\xi
_{1}\ldots\xi_{k}\right)  :=C^{\prime}\left(  \$_{\operatorname*{new}}%
^{\prime}\right)  $.

\item[(4)] Output $\left(  s_{i},\sigma_{i},\xi_{i}\right)  $.
\end{enumerate}

By assumption,%
\[
\Pr\left[  \mathsf{Ver}_{2}^{\prime}\left(  s_{1}\ldots s_{k},\sigma_{1}%
\ldots\sigma_{k},\xi_{1}\ldots\xi_{k}\right)  \text{ accepts}\right]
\geq\delta^{\prime}.
\]
Now, suppose $\mathsf{Ver}_{2}^{\prime}$ does accept. \ Then by the definition
of $\mathsf{Ver}_{2}^{\prime}$, at least $\left(  1-\varepsilon-\eta\right)
k$\ of%
\[
\mathsf{Ver}\left(  s_{1},\sigma_{1}\right)  ,\ldots,\mathsf{Ver}\left(
s_{k},\sigma_{k}\right)
\]
must have accepted, along with at least $\left(  1-\varepsilon-\eta\right)
k$\ of%
\[
\mathsf{Ver}\left(  s_{1},\xi_{1}\right)  ,\ldots,\mathsf{Ver}\left(
s_{k},\xi_{k}\right)  .
\]
So there must be at least $\left(  1-2\varepsilon-2\eta\right)  k$\ indices
$j\in\left[  k\right]  $\ such that $\mathsf{Ver}\left(  s_{j},\sigma
_{j}\right)  $ and $\mathsf{Ver}\left(  s_{j},\xi_{j}\right)  $\ \textit{both}
accepted. \ Therefore%
\begin{align*}
\Pr\left[  \mathsf{Ver}_{2}\left(  s_{i},\sigma_{i},\xi_{i}\right)  \text{
accepts}\right]   &  =\Pr\left[  \mathsf{Ver}\left(  s_{i},\sigma_{i}\right)
\text{ and }\mathsf{Ver}\left(  s_{i},\xi_{i}\right)  \text{\ accept}\right]
\\
&  \geq\left(  1-2\varepsilon-2\eta\right)  \delta^{\prime}.
\end{align*}
Taking $\eta>0$\ sufficiently small now yields the theorem.
\end{proof}

A direct counterpart of Theorem \ref{ampcomp}, with exactly the same
parameters, can be proved for public-key quantum money schemes. \ Once again,
the main idea is to consider \textquotedblleft composite
banknotes\textquotedblright\ $\$^{\prime}=\left(  \$_{1},\ldots,\$_{k}\right)
$---and this time, to associate with each $\$_{i}$\ a \textit{different},
independently-chosen public/private key pair. \ Another counterpart of Theorem
\ref{ampcomp} can be proved for digital signature schemes, indeed with
slightly better parameters ($\delta^{\prime}>\frac{\delta}{1-\varepsilon}$
instead of $\delta^{\prime}>\frac{\delta}{1-2\varepsilon}$). \ We omit the details.

\section{Appendix: Complexity-Theoretic No-Cloning Theorem\label{QOR}}

In Section \ref{COR}, we applied the inner-product adversary method to show
that a uniform superposition $\left\vert A\right\rangle $\ over a random
subspace $A\leq\mathbb{F}_{2}^{n}$\ requires $\Omega\left(  2^{n/4}\right)  $
quantum queries to duplicate, even if we are given access to an oracle that
decides membership in both $A$\ and $A^{\perp}$. \ For completeness, in this
appendix we present a simpler application of the inner-product adversary
method: namely, we show that a Haar-random $n$-qubit state $\left\vert
\psi\right\rangle $\ requires $\Omega\left(  2^{n/2}\right)  $ queries to
duplicate, if we are given access to an oracle $U_{\psi}$\ that accepts
$\left\vert \psi\right\rangle $\ and that rejects every state orthogonal to
$\left\vert \psi\right\rangle $. \ The latter is the original result that
Aaronson \cite{aar:qcopy}\ called the \textquotedblleft Complexity-Theoretic
No-Cloning Theorem,\textquotedblright\ though a proof has not appeared until now.

In Section \ref{COR}, we used the lower bound for copying subspace states\ to
construct a quantum money mini-scheme that was provably secure relative to a
classical oracle. \ In the same way, one can use the Complexity-Theoretic
No-Cloning Theorem to construct a mini-scheme that is provably secure relative
to a \textit{quantum} oracle. \ We omit the details of that construction, not
only because it is superseded by the classical oracle construction in Section
\ref{COR}, but because the two constructions are essentially the same. \ The
one real difference is that the quantum oracle construction benefits from a
quadratically better lower bound on the number of queries needed to
counterfeit: $\Omega\left(  2^{n/2}\right)  $\ rather than $\Omega\left(
2^{n/4}\right)  $.

Choose an $n$-qubit pure state $\left\vert \psi\right\rangle $\ uniformly from
the Haar measure, and fix $\left\vert \psi\right\rangle $\ in what follows.
\ Let $U_{\psi}$\ be a unitary transformation such that $U_{\psi}\left\vert
\psi\right\rangle =-\left\vert \psi\right\rangle $\ and $U_{\psi}\left\vert
\eta\right\rangle =\left\vert \eta\right\rangle $ for all $\left\vert
\eta\right\rangle $\ orthogonal to $\left\vert \psi\right\rangle $. \ The
following is the direct analogue of Theorem \ref{nocopyc}.

\begin{theorem}
[Complexity-Theoretic No-Cloning]\label{nocopyq}Given one copy of $\left\vert
\psi\right\rangle $, as well as oracle access to $U_{\psi}$, a counterfeiter
needs $\Omega\left(  2^{n/2}\right)  $ queries to prepare $\left\vert
\psi\right\rangle ^{\otimes2}$ with certainty (for a worst-case $\left\vert
\psi\right\rangle $).
\end{theorem}

\begin{proof}
We will apply Theorem \ref{lbthm}. \ Let the set $\mathcal{O}$\ contain
$U_{\psi}$\ for every possible $n$-qubit state $\left\vert \psi\right\rangle
$. \ Then $S_{U_{\psi}}$\ is just the $1$-dimensional subspace corresponding
to $\left\vert \psi\right\rangle $. \ Also, put $\left(  U_{\psi},U_{\varphi
}\right)  \in R$\ if and only if $\left\vert \left\langle \psi|\varphi
\right\rangle \right\vert =c$, for some $0<c<1$ to be specified later. \ Then
for all $U_{\psi}\in\mathcal{O}$\ and $\left\vert \eta\right\rangle \in
S_{U_{\psi}}^{\bot}$, we have%
\begin{align*}
\operatorname*{E}_{U_{\varphi}~:~\left(  U_{\psi},U_{\varphi}\right)  \in
R}\left[  \left\vert \left\langle \eta|\varphi\right\rangle \right\vert
^{2}\right]   &  =\operatorname*{E}_{\left\vert \varphi\right\rangle
~:~\left\vert \left\langle \psi|\varphi\right\rangle \right\vert =c}\left[
\left\vert \left\langle \eta|\varphi\right\rangle \right\vert ^{2}\right] \\
&  =\operatorname*{E}_{\left\vert v\right\rangle \in S_{U_{\psi}}^{\bot}%
}\left[  \left\vert \left\langle \eta\right\vert \left(  c\left\vert
\psi\right\rangle +\sqrt{1-c^{2}}\left\vert v\right\rangle \right)
\right\vert ^{2}\right] \\
&  =\left(  1-c^{2}\right)  \operatorname*{E}_{\left\vert v\right\rangle \in
S_{U_{\psi}}^{\bot}}\left[  \left\vert \left\langle \eta|v\right\rangle
\right\vert ^{2}\right] \\
&  =\frac{1-c^{2}}{2^{n}-1}.
\end{align*}
So set $\varepsilon:=\frac{1-c^{2}}{2^{n}-1}$. \ If the counterfeiter
succeeds, it must map $\left\vert \psi\right\rangle $\ to some state
$\left\vert f_{\psi}\right\rangle :=\left\vert \psi\right\rangle \left\vert
\psi\right\rangle \left\vert \operatorname*{garbage}_{\psi}\right\rangle $,
and $\left\vert \varphi\right\rangle $\ to $\left\vert f_{\varphi
}\right\rangle :=\left\vert \varphi\right\rangle \left\vert \varphi
\right\rangle \left\vert \operatorname*{garbage}_{\varphi}\right\rangle $.
\ Note that $\left\vert \left\langle f_{\psi}|f_{\varphi}\right\rangle
\right\vert \leq c^{2}$. \ So setting $d:=c^{2}$, Theorem \ref{lbthm}\ tells
us that the counterfeiter must make%
\[
\Omega\left(  \left(  c-c^{2}\right)  \sqrt{\frac{2^{n}-1}{1-c^{2}}}\right)
\]
queries to $U_{\psi}$. \ Fixing (say) $c=1/2$, this is $\Omega\left(
2^{n/2}\right)  $.
\end{proof}

Like Theorem \ref{nocopyc}, Theorem \ref{nocopyq}\ is easily seen to be tight,
since one can use the amplitude amplification algorithm (Lemma \ref{aa}) to
find $\left\vert \psi\right\rangle $, and thereby prepare $\left\vert
\psi\right\rangle ^{\otimes2}$, using $O\left(  2^{n/2}\right)  $\ queries to
$U_{\psi}$.

For completeness, we observe the following generalization of Theorem
\ref{nocopyq}.

\begin{theorem}
\label{kcopiesq}Given $k$ copies of $\left\vert \psi\right\rangle $, as well
as oracle access to $U_{\psi}$, a counterfeiter needs $\Omega\left(
2^{n/2}/\sqrt{k}\right)  $ queries to prepare $\left\vert \psi\right\rangle
^{\otimes k+1}$ with certainty (for a worst-case $\left\vert \psi\right\rangle
$).
\end{theorem}

\begin{proof}
If the counterfeiter succeeds, it must map $\left\vert \psi\right\rangle
^{\otimes k}$\ to some state $\left\vert f_{\psi}\right\rangle :=\left\vert
\psi\right\rangle ^{\otimes k+1}\left\vert \operatorname*{garbage}_{\psi
}\right\rangle $, and $\left\vert \varphi\right\rangle ^{\otimes k}$\ to
$\left\vert f_{\varphi}\right\rangle :=\left\vert \varphi\right\rangle
^{\otimes k+1}\left\vert \operatorname*{garbage}_{\varphi}\right\rangle $.
\ Note that $\left\vert \left\langle f_{\psi}|f_{\varphi}\right\rangle
\right\vert \leq c^{k+1}$. \ So setting $d:=c^{k+1}$, Theorem \ref{lbthm}%
\ tells us that the counterfeiter must make%
\[
\Omega\left(  \left(  c^{k}-c^{k+1}\right)  \sqrt{\frac{2^{n}-1}{1-c^{2}}%
}\right)
\]
queries to $U_{\psi}$. \ Fixing $c:=1-\frac{1}{k}$, the above is%
\[
\Omega\left(  \left(  \frac{1}{e}-\left(  1-\frac{1}{k}\right)  \frac{1}%
{e}\right)  \sqrt{\frac{2^{n}}{1-\left(  1-1/k\right)  ^{2}}}\right)
=\Omega\left(  \frac{2^{n/2}}{\sqrt{k}}\right)  .
\]

\end{proof}

We end this appendix by stating, without proof, three stronger lower bounds
that are the direct analogues of Corollary \ref{smallerror}, Corollary
\ref{largeerror}, and Theorem \ref{averagecase}\ respectively.

\begin{corollary}
\label{smallerrorq}Given one copy of $\left\vert \psi\right\rangle $, as well
as oracle access to $U_{\psi}$, a counterfeiter needs $\Omega\left(
2^{n/2}\right)  $ queries to prepare a state $\rho$\ such that $\left\langle
\psi\right\vert ^{\otimes2}\rho\left\vert \psi\right\rangle ^{\otimes2}%
\geq0.9999$ (for a worst-case $\left\vert \psi\right\rangle $).
\end{corollary}

\begin{corollary}
\label{largeerrorq}Let $1/\varepsilon=o\left(  2^{n}\right)  $. \ Given one
copy of $\left\vert \psi\right\rangle $, as well as oracle access to $U_{\psi
}$, a counterfeiter needs $\Omega\left(  \sqrt{\varepsilon}2^{n/2}\right)  $
queries to prepare a state $\rho$\ such that $\left\langle \psi\right\vert
^{\otimes2}\rho\left\vert \psi\right\rangle ^{\otimes2}\geq\varepsilon$ (for a
worst-case $\left\vert \psi\right\rangle $).
\end{corollary}

\begin{theorem}
\label{averagecaseq}Let $\left\vert \psi\right\rangle $ be an $n$-qubit pure
state chosen uniformly from the Haar measure. \ Given one copy of $\left\vert
\psi\right\rangle $, as well as oracle access to $U_{\psi}$, a counterfeiter
$C$\ needs $\Omega\left(  \sqrt{\varepsilon}2^{n/2}\right)  $ queries to
prepare a $2n$-qubit state $\rho$\ that a projector $V_{\psi}^{\otimes2}%
$\ onto $\left\vert \psi\right\rangle ^{\otimes2}$\ accepts with probability
at least $\varepsilon$, for all $1/\varepsilon=o\left(  2^{n}\right)  $.
\ Here the probability is taken over the choice of $\left\vert \psi
\right\rangle $, as well as the behavior of $C$\ and $V_{\psi}^{\otimes2}$.
\end{theorem}

\section{Acknowledgments}

We thank Andris Ambainis, Boaz Barak, Dmitry Gavinsky, Daniel Gottesman, Aram
Harrow, Yuval Ishai, Shelby Kimmel, Shaunak Kishore, Greg Kuperberg, Andy
Lutomirski, Abel Molina, Rafi Ostrovsky, Amit Sahai, Peter Shor, John
Watrous, and Ronald de Wolf for helpful discussions and correspondence; and the anonymous
reviewers for their comments.

\bibliographystyle{plain}
\bibliography{thesis}

\end{document}